%% file: prov-data-skipping.tex
\documentclass[sigconf, nonacm]{acmart}

\usepackage{balance}  

\usepackage{amsmath}
\usepackage{bm}
\usepackage{listings}
\usepackage{xcolor}
\usepackage{colortbl}
\usepackage{multirow}
\usepackage{hhline}
\usepackage{float}
\usepackage{times}  
\usepackage{textcomp}
\usepackage{adjustbox}
\usepackage[disable]{todonotes}
\newcounter{todocounter}
\usepackage{subcaption}
\usepackage[]{algorithm2e}

\newcommand\vldbdoi{XX.XX/XXX.XX}
\newcommand\vldbpages{XXX-XXX}
\newcommand\vldbvolume{14}
\newcommand\vldbissue{1}
\newcommand\vldbyear{2020}
\newcommand\vldbauthors{\authors}
\newcommand\vldbtitle{\shorttitle}
\newcommand\vldbavailabilityurl{http://vldb.org/pvldb/format_vol14.html}
\newcommand\vldbpagestyle{plain}

\newtheorem{Theorem}{Theorem}
\newtheorem{Definition}{Definition}
\newtheorem{Lemma}{Lemma}

\newtheorem{Example}{Example}

\newenvironment{proofsketch}
{
\noindent \textsc{Proof Sketch.}%
}%
{\qedsymbol}

\usepackage{hyperref}
\usepackage{cleveref}
\usepackage{xspace}

\crefname{Example}{ex.}{ex.}
\Crefname{Example}{Ex.}{Ex.}
\Crefname{figure}{Fig.}{Fig.}
\Crefname{section}{Sec.}{Sec.}
\Crefname{Definition}{Def.}{Def.}
\Crefname{Theorem}{Thm.}{Thm.}
\Crefname{Lemma}{Lem.}{Lem.}

\newcommand{\BG}[1]{\todo[fancyline,inline]{\textbf{Boris says:$\,$} #1}}
\newcommand{\XN}[1]{\todo[color=green!40,fancyline,inline]{\textbf{Xing says:$\,$} #1}}

\renewcommand{\shortauthors}{Niu, et al.}

\input{./sections/macros.tex}
\booltrue{Techreport}

\begin{document}
\ifnottechreport{\title{Provenance-based Data Skipping}}
\iftechreport{\title{Provenance-based Data Skipping (TechReport)}}


\author{Xing Niu$^\ast$, Ziyu Liu$^\ast$, Pengyuan Li$^\ast$, Boris Glavic$^\ast$, Dieter Gawlick$^{\alpha}$, Vasudha Krishnaswamy$^{\alpha}$, Zhen H. Liu$^{\alpha}$, Danica Porobic$^{\alpha}$}
\affiliation{\institution{Illinois Institute of Technology$^\ast$, Oracle$^{\alpha}$}}
\email{{xniu7,zliu102, pli26}@hawk.iit.edu, bglavic@iit.edu}
\email{{dieter.gawlick, vasudha.krishnaswamy, zhen.liu, danica.porobic}@oracle.com}







\begin{abstract}
  Database systems analyze queries to determine upfront which data is needed for answering them and use
   indexes and other physical design techniques to speed-up access to that data.
 However, for important classes of queries, e.g., HAVING and top-k queries, it is impossible to determine up-front what data is \emph{relevant}.
To overcome this limitation, we develop provenance-based data skipping (PBDS), a
novel approach that generates provenance sketches to concisely encode
what data is relevant for a query. Once a provenance sketch has been captured it 
is used to speed up subsequent queries. 
PBDS can exploit
  physical design artifacts such as indexes and zone maps. 
Our approach significantly improves performance for
both disk-based and main-memory database systems.
\end{abstract}

\input{./sections/lstdefs.tex}
\lstset{style=psqlcolor}

\maketitle

\pagestyle{\vldbpagestyle}
\begingroup\small\noindent\raggedright\textbf{PVLDB Reference Format:}\\
\vldbauthors. \vldbtitle. PVLDB, \vldbvolume(\vldbissue): \vldbpages, \vldbyear.\\
\href{https://doi.org/\vldbdoi}{doi:\vldbdoi}
\endgroup
\begingroup
\renewcommand\thefootnote{}\footnote{\noindent
This work is licensed under the Creative Commons BY-NC-ND 4.0 International License. Visit \url{https://creativecommons.org/licenses/by-nc-nd/4.0/} to view a copy of this license. For any use beyond those covered by this license, obtain permission by emailing \href{mailto:info@vldb.org}{info@vldb.org}. Copyright is held by the owner/author(s). Publication rights licensed to the VLDB Endowment. \\
\raggedright Proceedings of the VLDB Endowment, Vol. \vldbvolume, No. \vldbissue\ %
ISSN 2150-8097. \\
\href{https://doi.org/\vldbdoi}{doi:\vldbdoi} \\
}\addtocounter{footnote}{-1}\endgroup

\ifdefempty{\vldbavailabilityurl}{}{
\vspace{.3cm}
\begingroup\small\noindent\raggedright\textbf{PVLDB Artifact Availability:}\\
The source code, data, and/or other artifacts have been made available at \url{\vldbavailabilityurl}.
\endgroup
}

\input{./sections/introduction}
\input{./sections/related_work}
\input{./sections/background}
\input{./sections/ps}
\input{./sections/safety_check}
\input{./sections/template_query}
\input{./sections/capture}
\input{./sections/use}
\input{./sections/implementation}
\input{./sections/experiments}
\input{./sections/conclusion}
\input{./sections/future_work}



\bibliographystyle{ACM-Reference-Format}
\bibliography{prov-data-skipping}

\end{document}
\endinput

%% file: sections/macros.tex
\newcommand{\pk}{PK\xspace}
\newcommand{\nor}{No-PS\xspace}
\newcommand{\termPBDS}{provenance-based data skipping\xspace}
\newcommand{\upcasePBDS}{Provenance-based Data Skipping\xspace}
\newcommand{\pbds}{PBDS\xspace}
\newcommand{\ps}{PS\xspace}

\newcommand{\pss}[1]{{\ps}{#1}}
\newcommand{\runs}{n_{runs}}

\newcommand{\defas}{\coloneqq}

\newcommand{\card}[1]{\vert{#1}\vert}

\newcommand{\mathtext}[1]{\thickspace\text{#1}\thickspace}

\newcommand{\oNotation}[1]{O({#1})}

\newcommand{\relSchema}{{\bf R}}
\newcommand{\schemaOf}[1]{\textsc{SCH}({#1})}
\newcommand{\dbSchema}{\bf \db}
\newcommand{\arity}{arity}
\newcommand{\att}{a}
\newcommand{\attset}{A}
\newcommand{\sa}{X}
\newcommand{\rel}{R}
\newcommand{\rela}[1]{\textsc{#1}}
\newcommand{\db}{D}
\newcommand{\query}{Q}
\newcommand{\tup}{t}

\newcommand{\udom}{\mathbb{U}}
\newcommand{\mult}[2]{{#1}^{#2}}
\newcommand{\concat}{\circ}
\newcommand{\lbbar}{\{\kern-0.5ex|}
\newcommand{\rbbar}{|\kern-0.5ex\}}
\newcommand{\bag}[1]{\lbbar {#1} \rbbar}

\newcommand{\qPopden}{\ensuremath{\query_1}}
\newcommand{\qAvgden}{\ensuremath{\query_2}}

\newcommand{\parti}{F}
\newcommand{\rparti}{F}
\newcommand{\hparti}{F}
\newcommand{\dbpart}{\mathcal{F}}
\newcommand{\frag}{f}

\newcommand{\domain}[1]{\mathcal{D}(#1)}
\newcommand{\range}{r}
\newcommand{\ranges}{\mathcal{R}}
\newcommand{\rlu}[2]{\ensuremath{\text{\upshape{[{#1},{#2}]}}}}

\newcommand{\prov}[2]{P(#1,#2)}

\newcommand{\provSketch}{\mathcal{P}}

\newcommand{\psSet}{\mathcal{PS}}
\newcommand{\provranges}[3]{\ranges(#1,#2,#3)}
\newcommand{\instOf}[1]{\db_{#1}}
\newcommand{\relInst}[1]{\rel_{#1}}

\newcommand{\schema}[1]{\textsc{Sch}(#1)}
\newcommand{\propagate}[1]{\textsc{Prop}(#1)}
\newcommand{\initialize}[2]{\textsc{Init}_{#2}(#1)}
\newcommand{\instrument}[1]{\textsc{Instr}(#1)}

\newcommand{\quse}[2]{{#1}[{#2}]}
\newcommand{\bv}[1]{\texttt{#1}}
\newcommand{\bvsng}[1]{\textsc{sng}({#1})}
\newcommand{\psatt}{\ensuremath{\lambda}}
\newcommand{\psallatt}{\ensuremath{\Lambda}}
\newcommand{\psa}[1]{\psatt_{{#1}}}

\newcommand{\cthead}{\cellcolor{tableHeadGray}}
\colorlet{LightBlue}{blue!30!}
\colorlet{LightRed}{red!20!}
\colorlet{LightGreen}{green!30!}

\newcommand{\ycell}[1]{{\cellcolor{lightyellow}{#1}}}

\colorlet{LRed}{red!60!}
\colorlet{LGreen}{green!70!}
\colorlet{LBlue}{blue!70!}

\definecolor{white}{rgb}{1,1,1}
\definecolor{black}{rgb}{0,0,0}
\definecolor{lgrey}{rgb}{0.9,0.9,0.9}
\definecolor{llgrey}{rgb}{0.93,0.93,0.93}
\definecolor{lllgrey}{rgb}{0.96,0.96,0.96}
\definecolor{grey}{rgb}{0.7,0.7,0.7}
\definecolor{dgrey}{rgb}{0.5,0.5,0.5}
\definecolor{red}{rgb}{1,0,0}
\definecolor{green}{rgb}{0,1,0}
\definecolor{yellow}{rgb}{1.0, 1.0, 0.0}
\definecolor{darkgreen}{rgb}{0,0.5,0}
\definecolor{darkblue}{rgb}{0,0,0.5}
\definecolor{darkpurple}{rgb}{0.5,0,0.5}
\definecolor{darkdarkpurple}{rgb}{0.3,0,0.3}
\definecolor{blue}{rgb}{0,0,1}
\definecolor{shadegreen}{rgb}{0.95,1,0.95}
\definecolor{shadeblue}{rgb}{0.95,0.95,1}
\definecolor{shadered}{rgb}{1,0.85,0.85}
\definecolor{oddRowGrey}{rgb}{0.95,0.95,0.95}
\definecolor{evenRowGrey}{rgb}{0.85,0.85,0.85}
\definecolor{tableHeadGray}{rgb}{0.85,0.85,0.85}
\definecolor{lightgrey}{rgb}{0.88,0.88,0.88}
\definecolor{lightyellow}{rgb}{1.0, 1.0, 0.88}
\definecolor{selectiveyellow}{rgb}{1.0, 0.73, 0.0}

\newcommand{\parttitle}[1]{\smallskip\noindent\textbf{#1.}}

\newcommand{\projection}{\ensuremath{\Pi}}
\newcommand{\selection}{\ensuremath{\sigma}}
\newcommand{\aggregation}{\ensuremath{\gamma}}
\newcommand{\Aggregation}[2]{%
 \IfEqCase{#1}{%
   {}{\gamma_{{#2}}}%
   }[\gamma_{{#1};\,{#2}}]%
 }
\newcommand{\aggf}{f}
\newcommand{\agga}{b}
\newcommand{\grpatts}{G}
\newcommand{\grps}[2]{\textsc{grps}({#1},{#2})}
\newcommand{\gdata}[2]{{#1}_{#2}}

\newcommand{\bitor}{\textsc{bitor}}

\newcommand{\union}{\ensuremath{\cup}}
\newcommand{\intersection}{\ensuremath{\cap}}
\newcommand{\difference}{\ensuremath{-}}

\newcommand{\join}{\ensuremath{\bowtie}}
\newcommand{\crossprod}{\ensuremath{\times}}
\newcommand{\duprem}{\ensuremath{\delta}}

\newcommand{\ordlimit}[2]{\ensuremath{\tau_{#1,#2}}}

\newcommand{\win}{\ensuremath{\omega}}
\newcommand{\Win}[3]{\win_{#1,#2,#3}}

\newcommand{\sasc}[1]{{#1}\uparrow}
\newcommand{\sdesc}[1]{{#1}\downarrow}

\newcommand{\scase}[1]{\textsf{select}({#1})}
\newcommand{\swhen}[2]{{#1} \mapsto {#2}}

\newcommand{\pq}{\mathcal{T}}
\newcommand{\pqa}{\pq_1}
\newcommand{\pqb}{\pq_2}
\newcommand{\apq}{\pq[\pvec]}
\newcommand{\pdom}{\mathbb{P}}
\newcommand{\p}{p}
\newcommand{\pvec}{\vec{\p}}
\newcommand{\pv}{v}
\newcommand{\vvec}{\vec{\pv}}
\newcommand{\pinst}[2]{{#1}[#2]}

\newcommand{\attrs}[1]{\textsc{attrs}(#1)}

\newcommand{\pred}{pred}
\newcommand{\allcond}{conds}
\newcommand{\ucond}{uconds}

\newcommand{\expr}{expr}

\newcommand{\ngpred}{\text{non-grp-pred}}

\newcommand{\true}{\ensuremath{\mathbf{true}}\xspace}

\newcommand{\querya}{Q_1}
\newcommand{\queryb}{Q_2}

\newcommand{\gpe}{ge}
\newcommand{\gc}{gc}
\newcommand{\aComp}{\Psi}
\newcommand{\matchContains}[1]{\precsim_{#1}}
\newcommand{\matchMap}{\mathcal{M}}

\DeclareRobustCommand\fignumref[1]{\raisebox{.5pt}{\textcircled{\raisebox{-.75pt} {\small #1}}}}

\newcommand{\trimfigspace}[1][-3mm]{\vspace{#1}}

\newbool{Techreport}
\newrobustcmd{\ifnottechreport}[1]{\ifbool{Techreport}{}{#1}}
\newrobustcmd{\iftechreport}[1]{\ifbool{Techreport}{#1}{}}


%% file: sections/lstdefs.tex
\definecolor{lstpurple}{rgb}{0.5,0,0.5}
\definecolor{lstred}{rgb}{1,0,0}
\definecolor{lstreddark}{rgb}{0.7,0,0}
\definecolor{lstredl}{rgb}{0.64,0.08,0.08}
\definecolor{lstmildblue}{rgb}{0.66,0.72,0.78}
\definecolor{lstblue}{rgb}{0,0,1}
\definecolor{lstmildgreen}{rgb}{0.42,0.53,0.39}
\definecolor{lstgreen}{rgb}{0,0.5,0}
\definecolor{lstorangedark}{rgb}{0.6,0.3,0}
\definecolor{lstorange}{rgb}{0.75,0.52,0.005}
\definecolor{lstorangelight}{rgb}{0.89,0.81,0.67}
\definecolor{lstbeige}{rgb}{0.90,0.86,0.45}

\DeclareFontShape{OT1}{cmtt}{bx}{n}{<5><6><7><8><9><10><10.95><12><14.4><17.28><20.74><24.88>cmttb10}{}

\lstdefinestyle{psql}
{
tabsize=2,
basicstyle=\footnotesize\upshape\ttfamily,
language=SQL,
morekeywords={PROVENANCE,BASERELATION,INFLUENCE,COPY,ON,TRANSPROV,TRANSSQL,TRANSXML,CONTRIBUTION,COMPLETE,TRANSITIVE,NONTRANSITIVE,EXPLAIN,SQLTEXT,GRAPH,IS,ANNOT,THIS,XSLT,MAPPROV,cxpath,OF,TRANSACTION,SERIALIZABLE,COMMITTED,INSERT,INTO,WITH,SCN,UPDATED,WINDOW},
extendedchars=false,
keywordstyle=\bfseries,
mathescape=true,
escapechar=@,
sensitive=true
}

\lstdefinestyle{psqlcolor}
{
tabsize=2,
basicstyle=\small\upshape\ttfamily,
language=SQL,
morekeywords={PROVENANCE,BASERELATION,INFLUENCE,COPY,ON,TRANSPROV,TRANSSQL,TRANSXML,CONTRIBUTION,COMPLETE,TRANSITIVE,NONTRANSITIVE,EXPLAIN,SQLTEXT,GRAPH,IS,ANNOT,THIS,XSLT,MAPPROV,cxpath,OF,TRANSACTION,SERIALIZABLE,COMMITTED,INSERT,INTO,WITH,SCN,UPDATED,FOLLOWING,RANGE,UNBOUNDED,PRECEDING,OVER,PARTITION,WINDOW},
extendedchars=false,
keywordstyle=\bfseries\color{lstpurple},
deletekeywords={count,min,max,avg,sum,lag,first_value,last_value,rank},
keywords=[2]{count,min,max,avg,sum,lag,first_value,last_value,lead,row_number,rank},
keywordstyle=[2]\color{lstblue},
stringstyle=\color{lstreddark},
commentstyle=\color{lstgreen},
mathescape=true,
escapechar=@,
sensitive=true
}

\lstdefinestyle{datalog}
{
basicstyle=\footnotesize\upshape\ttfamily,
language=prolog
}

\lstdefinestyle{pseudocode}
{
  tabsize=3,
  basicstyle=\small,
  language=c,
  morekeywords={if,else,foreach,case,return,in,or,  and},
  extendedchars=true,
  mathescape=true,
  literate={:=}{{$\gets$}}1 {<=}{{$\leq$}}1 {!=}{{$\neq$}}1 {append}{{$\listconcat$}}1 {calP}{{$\cal P$}}{2},
  keywordstyle=\color{lstpurple},
  escapechar=&,
  numbers=left,
  numberstyle=\color{lstgreen}\small\bfseries,
  stepnumber=1,
  numbersep=5pt,
    classoffset=1, 
    morekeywords={Yes, No},
    keywordstyle=\color{lstgreen},
    classoffset=0,
}

\lstdefinestyle{xmlstyle}
{
  tabsize=3,
  basicstyle=\small,
  language=xml,
  extendedchars=true,
  mathescape=true,
  escapechar=£,
  tagstyle=\color{keywordpurple},
  usekeywordsintag=true,
  morekeywords={alias,name,id},
  keywordstyle=\color{lstred}
}


%% file: sections/introduction.tex
\section{Introduction}
\label{sec:introduction}
Physical design techniques such as index structures, zone maps, and horizontal partitioning have been used to provide fast access to data based on its characteristics. To use any such data structure to answer a query,  database systems statically analyze the query to determine what data is \textit{relevant} for answering it. Based on this information the database optimizer
(i) optimizes the query to filter out irrelevant data as early as possible (e.g., using techniques like selection-pushdown) and (ii) determines 
how to execute this filtering step efficiently. 
\begin{Example}\label{ex:static-analysis-effective}
\ifnottechreport{Consider query $Q_1:$ \lstinline!SELECT city, popden! \\
\lstinline!FROM cities WHERE state = 'CA'! which returns the population density (\texttt{popden}) of cities in California.}
\iftechreport{Query $\qPopden$ shown in \Cref{tab:queries} returns the population density (\texttt{popden}) of cities in California.}
The database system may use an index on column \texttt{state}, if it exists,  to identify cities in California, reducing the I/O cost of the query. 
\end{Example}

In \Cref{ex:static-analysis-effective},  the \lstinline!WHERE! clause condition of the query implies  that only rows fulfilling the condition \lstinline!state = 'CA'! are \textbf{relevant}.
While this approach of statically analyzing a query to determine a declarative description of what data is relevant is effective for some queries, 
it is often not possible to determine relevance statically.
\ifnottechreport{\input{./sections/running_example_short}}
\iftechreport{\input{./sections/running_example_long}}
\begin{Example}\label{ex:static-analysis-ineffective}
Query $\qAvgden$ shown in \Cref{tab:queries} returns the state with the highest average population density per city. Consider the result of this query over an example database (\Cref{tab:result}) shown in \Cref{tab:cities}. California is the state with the highest average population density. 
Thus, only the second and third tuple are needed to produce the query result. One possible declarative description  of the relevant input is \lstinline!state = 'CA'!. However, unlike the previous example, this description is \emph{data-dependent}. For instance, if we delete the fifth row, then New York state has the highest average density and \lstinline!state = 'CA'! is no longer a correct description of  what data is needed. 
\end{Example}

Even though the query in the example above is selective, state-of-the-art systems are incapable of exploiting this selectivity since it is impossible to determine a declarative condition capturing what is relevant by static analysis. In fact, there are many important classes of queries including top-k queries, aggregation queries with \lstinline!HAVING!, and certain types of nested and correlated subqueries, for which static analysis is insufficient to determine relevance. The net effect is that while these queries may be quite selective,  
databases fail to exploit physical design artifacts to speed up their execution.

To overcome this shortcoming of current systems and better utilize existing physical design artifacts, we propose to analyze queries at runtime to determine concise and declarative  descriptions of what data is relevant (sufficient)  for answering a query. We use the provenance of a query to determine such descriptions since
for most provenance models, the provenance \BG{cite} of a query is \emph{sufficient} for answering the query. That is, if we evaluate a query over its provenance 
this yields the same result as evaluating the query over the full database. Specifically, we introduce provenance sketches which are concise descriptions of supersets of the provenance of a query.
Given a horizontal partition of an input table, a provenance sketch records which fragments of the partition contain provenance.
\iftechreport{Since only some tuples of a fragment may be provenance, any provenance sketch encodes a superset of the provenance of a query.}
Similar to query answering with views, a sketch captured for one query is used to speed-up the subsequent evaluation of the same or other queries.

\begin{Example}\label{ex:create-sketch}
Using the Lineage provenance model~\cite{CC09,CW00b}, the provenance of query $\qAvgden$ from \Cref{ex:static-analysis-ineffective} is the set of the tuples highlighted in yellow in~\Cref{tab:cities}. Consider the construction of a provenance sketch based on a range-partition of the input  on attribute \texttt{state} 
which is shown in \Cref{fig:example-range-part}. We assume that states are ordered lexicographically, e.g., \texttt{CA} belongs to the  interval $\rlu{AL}{DE}$. In \cref{tab:cities} we show the fragment that each tuple belongs to on the right, e.g., the first three tuples belong to fragment $f_1$. The provenance sketch $\provSketch_{state}$ of $\qAvgden$ according to this partition is $\{f_1\}$, i.e., 
 $f_1$ is the only fragment that contains provenance.
\end{Example}

\parttitle{Creating Provenance Sketches}
We present techniques 
for instrumenting an input query to compute a provenance sketch that are based on annotation propagation techniques developed for provenance capture (e.g.,~\cite{AF18}\cite{glavic2013using,bhagwat2005annotation}). Our approach has significantly lower runtime and storage overhead than such techniques.
Given a partition for one or more of the relations accessed by a query, we annotate each input row with the fragment it belongs to. These annotations are then propagated through operations ensuring that each intermediate result 
is annotated with a set of fragments that is a superset of its provenance.
\BG{Importantly, the annotation propagation rules are independent of what input partitions are used. Thus, extending our approach to support new types of partitioning input tables is simple: implement the instrumentation needed to annotate tuples with the fragment they belong to.} We  
 use bitvectors to compactly encode sets of fragments and develop fast methods for creating and merging these bitvectors. These optimizations can be implemented using common database extensibility mechanisms.

\parttitle{Using Provenance Sketches}
Once a provenance sketch for a query $\query$ has been created, we would like to use it to speed-up the subsequent execution of $\query$ or queries similar to $\query$. For that we need to be able to instrument a query to restrict its execution to data described by the provenance sketch. How this works specifically depends on what type of partition the provenance sketch is based on.
For example, for range partitioning, we can filter out data not belonging to the sketch using a disjunction of range restrictions.

\begin{Example}\label{ex:use-sketch}
Consider the provenance sketch $\provSketch_{state}$ from \Cref{ex:create-sketch}. It describes a superset of what data is relevant for answering query $\qAvgden$. Thus, we can use it to instrument the query to  filter out irrelevant data early-on. For a query $\query$ and provenance sketch $\provSketch$ we use $\quse{\query}{\provSketch}$ to denote the result of instrumenting $\query$ to filter out data that does not belong to $\provSketch$. To generate $\quse{\qAvgden}{\provSketch_{state}}$ (see \Cref{tab:queries}), we construct a condition over relation \texttt{cities} according to the  provenance sketch. $\provSketch_{state}$ consists of a single fragment $f_1 = [AL,DE]$. The data contained in this fragment can be retrieved using the condition \lstinline!WHERE state BETWEEN 'AL' AND 'DE'!. 
\end{Example}

By translating the sketch into a selection condition, we expose to the database system what data is relevant. Databases are already well-equipped to deal with such conditions and exploit existing physical design, e.g., use an index on attribute \texttt{state}.
\BG{Thus, queries using provenance sketches benefit from index structures, zone maps, etc.} However, provenance sketches with a large number of fragments can result in conditions with a large number of disjunctions that are expensive to evaluate. We present optimizations to speed up 
such expressions.

\parttitle{Sketch Safety}
So far we have assumed that if the provenance for a query is sufficient then so is the superset of the provenance encoded by a sketch. However, this is not always the case.

\begin{Example}\label{ex:not-safe}
Consider the partition of relation  cities on attribute \texttt{popden} shown on the bottom of \Cref{fig:example-range-part}. For convenience we show an identifier for each tuple  on the left (\Cref{tab:cities}). Partitioning the table in this way, the first four tuples belong to $g_2$ since their population density is between $4001$ and $9000$. The remaining three tuples belong to fragment $g_1$: $g_1 = \{t_5, t_6, t_7\}$ and $g_2 = \{t_1, t_2, t_3, t_4\}$. Fragment $g_2$ contains all tuples from the provenance of query $\qAvgden$.  Hence, the provenance sketch corresponding to this  partition is $\provSketch_{popden} = \{ g_2 \}$. If we evaluate $\qAvgden$ over $g_2$, we get a  result \texttt{(NY, 7000)} that is different from the result of the query over the full input relation. The reason is that $g_2$ contains only one tuple with state \texttt{NY} resulting in an average for this state that is higher than the one for 
\texttt{CA}.
\end{Example}

We call a  sketch \emph{safe} if evaluating the query over the data encoded by the sketch yields the same result as evaluating the query over the full input.
We demonstrate that in general it is not possible to determine safety statically at query compile time since safety can be data dependent. In spite of this negative result, we present a \emph{sound} technique that determines safety statically.
\BG{To ensure that only safe provenance sketches are produced by our approach, we need to be able to determine statically which sketch types are safe for a given query.
  As a first result, we prove that any sketch is safe for monotone queries. However, top-k queries and other types of queries that benefit from our approach frequently use non-monotone operators such as aggregation or negation in subqueries. We demonstrate that in general it is not possible to determine statically what sketch types are safe since safety can be data dependent. In spite of this negative result, we present a \emph{sound} technique that determines safety statically.  
  }
\BG{That is, our algorithm may fail to detect that a sketch type is safe for a query, but will never incorrectly label a sketch type as safe when it is not.}

\parttitle{Reusing Sketches for Parameterized Queries}
Parameterized que-\\ries are used to avoid repeated optimization of queries that only difference in constants used in selections. Such queries are frequently used in applications written on top of databases. 
Typically, an application uses a small number of parameterized queries, but executes many instances of each parameterized query. We develop a method that can determine statically whether a provenance sketch captured for one instance of a parameterized query can be used to answer another instance of this query.
Determining whether this is the case requires solving a problem that is closely related to query containment and, thus, unsurprisingly is undecidable in general. We present a sound, but not complete method.
\BG{Example?}

\parttitle{\upcasePBDS}
We develop a framework for creating and using sketches that we refer to as \textit{\termPBDS} (\pbds).
\pbds is used in a self-tuning fashion similar to automated materialized view selection: we 
decide 
when to create and when to use 
provenance sketches 
with the goal to optimize overall query performance. In this work, we assume read-only workloads which is common in OLAP and DISC systems.
We leave maintenance of provenance sketches under updates for future work.
\BG{That being said, to demonstrate the feasibility of a self-tuning solution for provenance sketches, we experimentally evaluate the performance of a simplified version of a self-tuning solution  which operators on  a fixed set of query templates to demonstrate that significant performance improvements are possible when the database system automatically selects when to create and when to use sketches.}


\parttitle{Contributions}
Our main technical contributions are:
\begin{itemize}
\item We introduce \termPBDS, a novel method for analyzing at runtime what data is relevant for a query and introduce provenance sketches as a concise encoding of what subset of the input is relevant for a query. 
\item We develop techniques for 
  capturing provenance sketches by
  instrumenting queries to propagate sketch annotations.
\item We 
  speed up evaluation of queries by instrumenting them
  to filter data based on provenance sketches. By exposing relevance information as  selection conditions to the DBMS,  existing
  physical design can be exploited.
\item We present a sound technique for determining what provenance sketch
  types are safe for which queries.
\item We develop a sound method to determine whether a provenance sketch for an instance $\query_1$ of a parameterized query  can be used to answer an instance $\query_2$ of this query.
\item Using DBMS extensibility mechanisms, we
  implement using query instrumentation.  We
  demonstrate experimentally that it leads to significant performance
  improvements for important query classes  including \lstinline!HAVING! and top-k queries.
\end{itemize}
The remainder of this paper is organized as follows. We review related work in \Cref{sec:related-work} and cover 
background in \Cref{sec:background}. Afterwards, we define provenance sketches in \Cref{sec:prov-sketch}.
In \Cref{sec:safety-check}, we present techniques for determining provenance sketch safety. We investigate how to resue sketches across queries in \Cref{sec:reuse-different}. 
 Our approach for capturing and using sketches is discussed in \Cref{sec:ps-capture,sec:ps-reuse}.  We present experimental results in \Cref{sec:exp} and conclude 
 in \Cref{sec:conclusion}.

%% file: sections/running_example_short.tex
\begin{figure}[t]
\centering
\begin{minipage}{0.95\linewidth}
\centering
   \begin{subfigure}{1\linewidth}
\centering
\begin{adjustbox}{max width=0.8\linewidth}
\centering
\begin{tabular}{|lc|} \hline
  \cthead $\qAvgden$ &
\begin{minipage}{1\linewidth}
   \begin{subfigure}{1\linewidth}
   		\begin{lstlisting}
SELECT state, avg(popden) AS avgden
FROM cities
GROUP BY state
ORDER BY avgden DESC
LIMIT 1;
		\end{lstlisting}
   \end{subfigure}
 \end{minipage} \\ \hline
  \cthead $\quse{\qAvgden}{\provSketch_{state}}$ &
\begin{minipage}{1\linewidth}
   \begin{subfigure}{1\linewidth}
   		\begin{lstlisting}
SELECT state, avg(popden) AS avgden
FROM cities
WHERE state BETWEEN 'AL' AND 'DE'
GROUP BY state
ORDER BY avgden DESC
LIMIT 1;
		\end{lstlisting}
   \end{subfigure}
 \end{minipage} \\ \hline
\end{tabular}
\end{adjustbox}
   	\caption{Queries}
     \label{tab:queries}
   \end{subfigure}
  \end{minipage}

  \vspace{3pt}
  \begin{minipage}{0.5\linewidth}
    \begin{subfigure}{1.0\linewidth}
{\footnotesize
      \begin{tabular}{c|c|c|c|r}
        \cline{2-4}
 & \cthead{popden}       & \cthead{city}     & \cthead{state} &                        \\ \cline{2-4}
$t_1$ & 4200               & Anchorage         & AK             & \multirow{3}{*}{$f_1$} \\
$t_2$ & \ycell        6000 & \ycell San Diego  & \ycell CA      &                        \\
$t_3$ & \ycell        5000 & \ycell Sacramento & \ycell CA      &                        \\ \cline{2-4}
$t_4$ & 7000               & New York          & NY             & \multirow{2}{*}{$f_3$} \\
$t_5$ & 2000               & Buffalo           & NY             &                        \\ \cline{2-4}
$t_6$ & 3700               & Austin            & TX             & \multirow{2}{*}{$f_4$} \\
$t_7$ & 2500               & Houston           & TX             &                        \\ \cline{2-4}
  \end{tabular}
}
      \caption{cities relation}
      \label{tab:cities}
    \end{subfigure}

  \end{minipage}
\begin{minipage}{0.485\linewidth}
 \begin{subfigure}{1\linewidth}
\centering
{\footnotesize
\begin{tabular}{|c|c|} 
    \hline
  \cthead city	& \cthead popden                            \\ \hline
  San Diego & 6000 \\
  Sacramento & 5000 \\ \hline
\end{tabular}
}
     \caption{Result of $\qPopden$}
     \label{tab:result-q1}
   \end{subfigure}
 \begin{subfigure}{1\linewidth}
\centering
{\footnotesize
\begin{tabular}{|c|c|} 
    \hline
  \cthead state	& \cthead avgden                            \\ \hline
  CA & 5500 \\ \hline
\end{tabular}
}
     \caption{Result of $\qAvgden$ (and $\quse{\qAvgden}{\provSketch_{state}}$)}
     \label{tab:result}
   \end{subfigure}

 \end{minipage}
\begin{subfigure}{1.0\linewidth}
  \centering
  $f_1 = \rlu{AL}{DE}$, $f_2 = \rlu{FL}{MI}$, $f_3 = \rlu{MN}{OK}$, $f_4 = \rlu{OR}{WY}$\\[2mm]
  $g_1 = \rlu{1000}{4000}$, $g_2 = \rlu{4001}{9000}$\\
  \caption{A range partition of relation cities on attribute \texttt{state} ($\parti_{state}$, top) and \texttt{popden} ($\parti_{popden}$, bottom)}
  \label{fig:example-range-part}
\end{subfigure}
\trimfigspace
\caption{Running Example}
\label{fig:eg-db}
\end{figure}

%% file: sections/running_example_long.tex
\begin{figure}[t]
\centering
\begin{minipage}{0.95\linewidth}
\centering
   \begin{subfigure}{1\linewidth}
\centering
\begin{adjustbox}{max width=0.9\linewidth}
\centering
\begin{tabular}{|lc|} \hline
 \cthead $\qPopden$ &
\begin{minipage}{1\linewidth}
   \begin{subfigure}{1\linewidth}
   		\begin{lstlisting}
SELECT city, popden
FROM cities
WHERE state = 'CA';
		\end{lstlisting}
   \end{subfigure}
   \end{minipage} \\ \hline
  \cthead $\qAvgden$ &
\begin{minipage}{1\linewidth}
   \begin{subfigure}{1\linewidth}
   		\begin{lstlisting}
SELECT state, avg(popden) AS avgden
FROM cities
GROUP BY state
ORDER BY avgden DESC
LIMIT 1;
		\end{lstlisting}
   \end{subfigure}
 \end{minipage} \\ \hline
  \cthead $\quse{\qAvgden}{\provSketch_{state}}$ &
\begin{minipage}{1\linewidth}
   \begin{subfigure}{1\linewidth}
   		\begin{lstlisting}
SELECT state, avg(popden) AS avgden
FROM cities
WHERE state BETWEEN 'AL' AND 'DE'
GROUP BY state
ORDER BY avgden DESC
LIMIT 1;
		\end{lstlisting}
   \end{subfigure}
 \end{minipage} \\ \hline
\end{tabular}
\end{adjustbox}
   	\caption{Queries}
     \label{tab:queries}
   \end{subfigure}
  \end{minipage}

  \vspace{3pt}
  \begin{minipage}{0.5\linewidth}
    \begin{subfigure}{1.0\linewidth}
{\footnotesize
      \begin{tabular}{c|c|c|c|r}
        \cline{2-4}
 & \cthead{popden}       & \cthead{city}     & \cthead{state} &                        \\ \cline{2-4}
$t_1$ & 4200               & Anchorage         & AK             & \multirow{3}{*}{$f_1$} \\
$t_2$ & \ycell        6000 & \ycell San Diego  & \ycell CA      &                        \\
$t_3$ & \ycell        5000 & \ycell Sacramento & \ycell CA      &                        \\ \cline{2-4}
$t_4$ & 7000               & New York          & NY             & \multirow{2}{*}{$f_3$} \\
$t_5$ & 2000               & Buffalo           & NY             &                        \\ \cline{2-4}
$t_6$ & 3700               & Austin            & TX             & \multirow{2}{*}{$f_4$} \\
$t_7$ & 2500               & Houston           & TX             &                        \\ \cline{2-4}
  \end{tabular}
}
      \caption{cities relation}
      \label{tab:cities}
    \end{subfigure}

  \end{minipage}
\begin{minipage}{0.485\linewidth}
 \begin{subfigure}{1\linewidth}
\centering
{\footnotesize
\begin{tabular}{|c|c|} 
    \hline
  \cthead city	& \cthead popden                            \\ \hline
  San Diego & 6000 \\
  Sacramento & 5000 \\ \hline
\end{tabular}
}
     \caption{Result of $\qPopden$}
     \label{tab:result-q1}
   \end{subfigure}
 \begin{subfigure}{1\linewidth}
\centering
{\footnotesize
\begin{tabular}{|c|c|} 
    \hline
  \cthead state	& \cthead avgden                            \\ \hline
  CA & 5500 \\ \hline
\end{tabular}
}
     \caption{Result of $\qAvgden$ (and $\quse{\qAvgden}{\provSketch_{state}}$)}
     \label{tab:result}
   \end{subfigure}

 \end{minipage}
\begin{subfigure}{1.0\linewidth}
  \centering
  $f_1 = \rlu{AL}{DE}$, $f_2 = \rlu{FL}{MI}$, $f_3 = \rlu{MN}{OK}$, $f_4 = \rlu{OR}{WY}$\\[2mm]
  $g_1 = \rlu{1000}{4000}$, $g_2 = \rlu{4001}{9000}$\\
  \caption{A range partition of relation cities on attribute \texttt{state} ($\parti_{state}$, top) and \texttt{popden} ($\parti_{popden}$, bottom)}
  \label{fig:example-range-part}
\end{subfigure}
\caption{Running Example}
\label{fig:eg-db}
\end{figure}

%% file: sections/related_work.tex
\section{Related Work}
\label{sec:related-work}

Our work is related to provenance capture techniques, compression and summarization of provenance,  maintaining query results using provenance, and physical design and self-tuning techniques. 


\parttitle{Physical Design and Self-tuning}
Physical design has been studied intensively 
including index structures~\cite{AD05,AG08,AS13a,BB00a,CD04,G06,GK10a,HN08,IM11,LK13,LL13,LL14b,SR87,SR87,yu2016two}, horizontal and vertical partitioning techniques~\cite{RJ17,SF16,D13,JD12,ZL10,PA04a,AN04,NR89,ceri1982horizontal}, zone maps~\cite{moerkotte1998small,clarke2013storage}, materialized views~\cite{AK12,halevy2001answering,GL01,AC00,GM99,AD98a,chaudhuri1995optimizing,LM95a}, join indexes~\cite{AD05,LR99,OG95,V87}, and many more.
However, databases fail to exploit existing physical design artifacts for important classes of selective queries such as top-k queries, because
it  is not possible to determine statically for such queries what data is relevant. Our work closes this gap by capturing relevance information at runtime and by translating it into selection conditions that database optimizers are well-equipped to handle.
Self-tuning techniques have a long tradition in database research~\cite{D13,GK10a,CN07}. Most closely related to our work are techniques for automated selection of and query answering with materialized views~\cite{DG17,PJ14,AD13c,GL01,halevy2001answering,AC00,GM99,chaudhuri1995optimizing,LM95a}. In contrast to materialized views, which result in storage overhead proportional to result size of queries, our technique  has negligible storage requirements and can exploit existing physical design artifacts.

\parttitle{Provenance Capture}
How to efficiently capture provenance has been studied extensively.
Many approaches encode provenance as annotations on data
and
 propagate such annotations through queries~\cite{KG12,GA12,XN18,XN17}.  
\BG{This can either be achieved by instrumenting queries to capture provenance and executing these queries on a unmodified database system or by extending an existing database engine. 
The first approach has the advantage that it can fully exploit existing database systems for querying of provenance and for optimizing provenance capture queries~\cite{XN18,XN17}. Building a system from scratch enables novel optimizations that are not possible withing the constraints of a classical relational system.}
A plethora of system that capture provenance for database queries have been introduced in recent years, including Perm~\cite{glavic2013using}, GProM~\cite{AF18}, DBNotes~\cite{bhagwat2005annotation}, LogicBlox~\cite{GA12}, Smoke~\cite{PW18}, Trio~\cite{aggarwal2009trio}, declarative Datalog debugging~\cite{KL12}, ExSPAN~\cite{ZS10}, ProvSQL~\cite{SJ18}, M\"uller et al.'s approach~\cite{MD18}, and  Links~\cite{FC18a}.
Like many approaches from related work, we use query instrumentation to capture provenance.
However, we exploit the fact that we only have to generate a single provenance sketch as the output to significantly simplify the instrumented query.


\parttitle{Compressing, Sketching, and Summarizing Provenance}
Early work on compressing provenance graphs such as~\cite{CJ08a} and~\cite{AB09} avoid storing  common substructures in a provenance graph more than once. Heinis et al.~\cite{HA08} studies how to trade space for computation when querying provenance graphs recursively. Malik et al.~\cite{MN10} use bloom filters to create compact sketches that over-approximate provenance. 
Chen et al.~\cite{CL17} apply similar ideas. 
Olteanu et al.~\cite{OS16,OZ11} study factorized representations of provenance and data. 
This line of work has lead to worst-case optimal algorithms for factorizing data wrt. to queries.
Closely related are techniques for provenance
summarization~\cite{AB15a,LGG18,LX16,GK15,DM19,LL20}, 
 intervention-based methods for explaining aggregate query results~\cite{WM13,RS14,RO15} and other approaches for generating explanations of outcomes~\cite{EA14,EF18}. Some of these techniques use declarative descriptions such as selection-patterns~\cite{EA14,EF18,RS14,LL20}.
 However, the 
summaries produced by these techniques 
are typically not sufficient for our purpose, e.g,. they may not encode a superset of the provenance. 
Our approach uses very compact (10s or 100s of bytes) declarative descriptions of a superset of the provenance which
is sufficient for reproducing a query result.

\parttitle{Optimizing Operations with Provenance}
Early work on data provenance already recognized the potential of using provenance for optimizing performance.
Pandas~\cite{IW10,IS10} uses provenance to selectively update data items in the result of a workflow to reflect changes to the workflow's inputs. Provenance has been used to provision for answering what-if queries~\cite{DI13,AK16a,DM13,DM19}.
Smoke~\cite{PW18}  uses provenance to speed-up queries in interactive visualization.
Assadi et al.~\cite{AK16a} create sketches over provenance to provision for
approximate answering of what-if queries.
In contrast to previous work which 
uses provenance or sketches of provenance instead of the original input, we use sketches as a light-weight index structure that allows us to efficiently access a relevant input data. 
\BG{Discuss Smoke in more detail since it also optimizes follow-up queries}


%% file: sections/background.tex
\section{Background}
\label{sec:background}

In  this section we introduce necessary background provenance and introduce relational algebra used in this work.

\subsection{The Relational Data Model and Algebra}
We use bold face (non-bold) to denote relation and database schemas (instances) respectively.
The arity $\arity(\relSchema)$ of 
$\relSchema$ is the number of attributes in $\relSchema$.
Here we use bag semantics and for simplicity will sometimes assume  a universal domain $\udom$.  That is, a relation $\rel$ for schema $\relSchema$ is a bag of tuples (elements of $\udom^{\arity(\relSchema)}$.
We denote bags as $\bag{\cdot}$). 
We use $\mult{\tup}{n} \in \rel$ to denote that tuple $\tup$ appears with multiplicity $n$ in relation $\rel$. 
Figure~\ref{fig:ex-relational-algebra} shows the bag semantics version of relational algebra used in this work.  
$\schemaOf{\query}$ is used to denote the schema of the result of query $\query$. We use $t.A$ to denote the projection of a tuple on a list of scalar expressions and $\concat$ to denote concatenation of tuples. For convenience, we  define $\tup^0 \in R$ to mean that the tuple $\tup$ is not in $\rel$.
The definitions of selection, projection, cross product, duplicate elimination, and set operations are standard.
\BG{Probably at least a comment set operations}
\BG{ADD BACK IF WE USE THIS SOMEWHERE: Union $R \union S $ returns the bag union relations $R$ and relation $S$. Intersection $R \intersection S$ returns the tuples which both exist in relation $R$ and relation $S$ with the smaller of the two multiplicities.   Difference $R \difference S$ return each tuple $\tup$ with a multiplicity $n - m$ where $\tup^{n} \in R$ and $\tup^m \in S$.}
Aggregation $\Aggregation{f(\att)}{G}(\rel)$ groups the input tuples according to their values in attributes $G$ and then computes the aggregation function $f$ over the bag of values of attribute $\att$ for each group.
\XN{The window operator $\Win{f(a) \to x}{G}{O}$ takes as  parameters an aggregation function $f$, a result attribute name $x$, a list of partition attributes $G$, and a list of order expressions $O$ of the form $\sdesc{\att}$ (sort descending on attribute $\att$) or $\sasc{\att}$ (sort ascending on $\att$) with $G \cap O = \emptyset$. This operator extends each input tuple $\tup \in \rel$ with an additional attribute $x$ whose value is computed by applying the aggregation function $f$ to the window for $\tup$. The window for $\tup$ contains the bag of $\att$-values for all tuples $\tup'$ from $\rel$ which agree with $\tup$ on attributes $G$ and for which $\tup.O \leq \tup'.O$. 
For that we assume 
the domains of $O$ are totally ordered.}
Let $<_O$ denote a total order over the tuples of a relation $R$ sorting a set of attributes $O$ breaking ties using the remaining attributes.
The top-k operator $\ordlimit{O}{C}(\rel)$ returns the $C$ smallest tuples from $\rel$ according to $<_O$.


 \begin{figure}[t]
   \begin{align*}
 \selection_\theta(R)               & = \bag{t^n \mid t^n \in \rel \wedge t \models \theta}
                                    & \projection_A(R) & = \bag{ t^n \mid n = \sum_{u.A = t \wedge u^m \in R} m }
   \end{align*}\\[-3mm]
   \begin{align*}
 \duprem (R)                        & = \bag{ t^1 \mid t^n \in R }
                                    & R \crossprod S   & = \bag{ t \concat s^{n*m} \mid t^n \in R \wedge s^m \in S }
   \end{align*}\\[-5mm]
 \begin{align*}
     R \union S                     & = \bag{ t^{n+m} \mid t^n \in R \wedge t^m \in S }
   \end{align*}\\[-5mm]
   \begin{align*}
     \Aggregation{f(\att)}{G}(\rel) & = \bag{ g \concat f(\gdata{R}{g})^1 \mid g \in \grps{R}{G} } \\
   \end{align*}\\[-9mm]
   \begin{align*}
     \grps{R}{G}                    & = \{ \tup.G \mid t^n \in R \}
                                    & \gdata{R}{g}     & = \bag{ (c)^n \mid n  = \hspace{-0.5cm} \sum_{\mult{t}{m} \in \rel \wedge t.G = g \wedge t.a = c } m }
   \end{align*}\\[-4mm]
   \begin{align*}
   \ordlimit{O}{C}(R) &=  \bag{ t^n \mid t^m \in \rel \land n= max(0,min(m,C - pos(R,O,t))) } \\
   pos(R,O,t) &= \card{ \bag{ t_1^n \mid t_1^n \in \rel \land t_1 <_O t }} 
   \end{align*}\\[-4mm]
 \caption{Bag Relational Algebra}
\label{fig:ex-relational-algebra}
\end{figure}
\subsection{Provenance and Sufficient Inputs}

In the following, we are interested in finding subsets $\db'$ of an input database $\db$ that are sufficient for answering a query $\query$. That is, for which $\query(\db') = \query(\db)$. We refer to such subsets as sufficient inputs.

\begin{Definition}[Sufficient Input]\label{def:sufficient}
\ifnottechreport{Given a query $\query$ and database $\db$, we call $\db' \subseteq \db$ \emph{sufficient} for $\query$ wrt. $\db$ if $Q(\db) = Q(\db')$.}
\iftechreport{Given a query $\query$ and database $\db$, we call $\db' \subseteq \db$ \emph{sufficient} for $\query$ wrt. $\db$ if\\
    $$Q(\db) = Q(\db')$$}
\end{Definition}

Several provenance models for relational queries have been proposed in the literature~\cite{CC09}. Most of these models have been proven to be instances of the semiring provenance model~\cite{GK07,GT17} and its extensions for difference/negation~\cite{GP10} and aggregation~\cite{AD11d} where each tuple is annotated with an arithmetic expression over variables representing input tuples. 
Our main interest in provenance is to be able to extract a sufficient subset of the input database from the provenance. That means, even a simple 
model like Lineage which encodes provenance as a subset of the input database is expressive enough. We use $\prov{\query}{\db}$ to denote the provenance of a query $\query$ over database $\db$ encoded as a set of tuples and assume that $\prov{\query}{\db}$ is sufficient for $\query$ wrt. $\db$. For instance, we may construct $\prov{\query}{\db}$ as the union of the Lineage for all  tuples  $\tup \in \query(\db)$.
Our formal results hold for any provenance model that guarantees sufficiency.

\BG{REMOVED: Data provenance is information about the origin of data and queries, i.e., which inputs contribute to a particular query result. For the sake of this paper, we do not limit ourselves to a particular provenance model. We only require that we should be able to extract a subset of a query's input that is sufficient for producing the result from the provenance of the query. Most provenance models proposed in the literature fulfill this property. For instance, we can use the Lineage or provenance polynomial  model.
Of cours

Here we define the query provenance as the minimal subset of the input tuples is sufficient for query Q over database I. That is running the query on its provenance still returns the correct result and without any tuple in the provenance the result becomes wrong.
}



\BG{The following discussion is not needed I think or at least out of place in this section.}
\BG{Node here we defined is the limitation provenance model which can not represent the complete query provenance. For example, the duplicate removal operator, under our model any duplicated input tuple could represent the provenance of this output tuple. However, this does not affect the query result. The de facto standard for computing query provenance~\cite{KG12,GA12} is to model provenance as annotations on data and define a query semantics to determine how annotations propagate. Under this semantics, each output tuple t of a query Q is annotated with its provenance that a combination of input tuple annotations that explains how these inputs were used by Q to derive t. Many database provenance systems apply this technique such as Perm~\cite{glavic2013using} and GProM~\cite{AF18} which compile queries with annotated semantics into
relational queries that is annotating each input tuple and propagating each annotation by applying a set of generic rewrite rules on each query operator in a relational encoding of the query. For example, after such a transformation, query Q1 used in the Sec.~\ref{sec:introduction} would be rewirted to $Q1^{prov}$.
}

\BG{In every table, each tuple is annotated by all columns named with `prov' prefix and propagate to the top of the query. The result of $Q1^{prov}$ is shown Fig.~\ref{fig:result-prov-q1} which indicates that q1's result tuple (Peter, 1) is computed from the tuple (Peter, 1) in student table and tuple (1, CS) in major table.
}

%% file: sections/ps.tex
\section{Provenance Sketches}
\label{sec:prov-sketch}

As discussed in \Cref{sec:introduction}, we propose provenance sketches to concisely represent a superset of the provenance of a query (a sufficient subset of the input) based on horizontal partitions of the input relations of the query. A sketch contains all fragments which contain at least one row from the provenance of the query. In this work, we limit the discussion to range-partitioning since it allows us to exploit  existing index structures when using a sketch to skip data. However, note that most of the techniques we introduce in this work are independent of the type of partitioning.

\subsection{Range Partitioning}
\label{tab:ps_def}
Given a set of intervals over the domains of a set of partitioning attributes $A \subset \relSchema$, range partitioning determines membership of tuples to fragments based on which interval their values belong to.
For simplicity, we define range partitioning for a single attribute $\att$.

\begin{Definition}[Range partition] \label{def:range}
Consider a relation $\rel$ and $\att \in \relSchema$. Let $\domain{a}$ denote the domain of $\att$. Let $\ranges = \{\range_1, \ldots, \range_n\}$ be a set of intervals $[l,u] \subseteq \domain{\att}$ such that $\bigcup_{i=0}^{n} \range_i = \domain{a}$ and $r_i \cap r_j = \emptyset$ for $i \neq j$. The \emph{range-partition} of $\rel$ on $\att$ according to $\ranges$ denoted as $\rparti_{\ranges,\att}(R)$ is defined as:
  \begin{align*}
    \parti_{\ranges,a}(R)  = \{ \rel_{{r}_{1}}, \ldots, \rel_{{r}_{n}} \} \hspace{2mm}\mathtext{\textbf{where}}\hspace{2mm}
    \rel_{\range}            = \bag{t^n \mid t^n \in \rel \wedge t.a \in  \range}
  \end{align*}
\end{Definition}
\Cref{fig:example-range-part} shows  two range partitions for our running example. 


\iftechreport{

A range partition of a relation divides the tuples of the relation into disjoint groups called \emph{fragments}.
}
\BG{REMOVED: For instance, a partitioning scheme that requires us to record for each tuple which fragment it belongs to would be of little use since the size of this representation would be linear in the size of the input database.}

\BG{To compute a provenance sketch $\provSketch$ for a query $\query$ and
database $\db$ according to $\dbpart$,
we might get different $\dbpart$ which based on the different parititioning methods we used for each table in the database $I$ such that we might get different provenance sketches. Currently, we are supporting range-based, hash-based and hash-page-based partitions. However, range-based partition could correspond directly to the physical design of the database such as index, pyhsical partitioning, zone-map and others which might be a king among other partition methods. Thus in this paper, we mainly focus on discussing range-based provenance sketch which is defined as follows.

\textbf{Range-based}: We define a set of ranges over the value of an attribute a which covers all values of a, the tuples in the same range belong to the same fragment.

\begin{Definition}[Range-partition] 
  Given a table $R$, a domain $\domain{a}: a \in \schema{R}$, and a set of ranges $r_1, \ldots, r_n$ covering all the values in $\domain{a}$ and $r_i \cap r_j = \emptyset: i \neq j$. Then the partitioning $$\parti_{r,a}(R) = \{ S_1, \ldots, S_n \}$$
$$S_{i} = \{ t \mid t \in R \wedge t.a \in  r_i \}$$
\end{Definition}

However, hash-based partition is to compute a hash over a set of attributes, thus the tuples with the same hash value belong to the same fragment and different with hash-based partition, hash-page-based is to hash over the page number which is extracted from the row number.
}

\subsection{Provenance Sketches}\label{sec:prov-sketch-sketches}

Consider a database $\db$,  query $\query$, and a range partition $\parti_{\ranges,\att}$ of $\rel$.
A provenance sketch $\provSketch$ for $\query$  according to $\parti_{\ranges,\att}$ is a subset of the ranges $\ranges$ of $\parti_{\ranges,\att}$ such that the fragments corresponding to the ranges in $\provSketch$ fully cover $\query$'s provenance within $\rel$, i.e., $\prov{\query}{\db} \cap \rel$. We use $\provranges{\db}{\parti_{\ranges,a}(R)}{\query} \subseteq \ranges$ to denote the set of ranges whose fragment contains at least one tuple from $\prov{\query}{\db}$:
   \begin{align*}
      \provranges{\db}{\parti_{\ranges,a}(R)}{\query} &= \{ \range \mid \range \in \ranges \wedge \exists t \in \prov{\query}{\db}: t \in R_{\range} \}
    \end{align*}

\ifnottechreport{
\begin{Definition}[Provenance Sketch]\label{def:provenance-sketch}
  Let $\query$ be a query, $\db$ a database,  $\rel$ a relation accessed by $\query$, and $\parti_{\ranges,a}(R)$ a range partition of $R$.
We call a subset $\provSketch$ of $\ranges$ a \textbf{provenance sketch} iff
$\provSketch \supseteq \provranges{\db}{\parti_{\ranges,a}(R)}{\query}$.
We call a provenance sketch \textbf{accurate} if
$\provSketch =  \provranges{\db}{\parti_{\ranges,a}(R)}{\query}$. 
We use $\relInst{\provSketch}$, called the \textbf{instance} of $\provSketch$, to denote $\bigcup_{\range \in \provSketch} \rel_{\range}$. 
\end{Definition}
}
\iftechreport{
\begin{Definition}[Provenance Sketch]\label{def:provenance-sketch}
  Let $\query$ be a query, $\db$ a database,  $\rel$ a relation accessed by $\query$, and $\parti_{\ranges,a}(R)$ a range partition of $R$.
%
  We call a subset $\provSketch$ of $\ranges$ a \textbf{provenance sketch} iff:
$$\provSketch \supseteq \provranges{\db}{\parti_{\ranges,a}(R)}{\query}$$

We call a provenance sketch \textbf{accurate} if
$$\provSketch =  \provranges{\db}{\parti_{\ranges,a}(R)}{\query}$$ 
%
We use $\relInst{\provSketch}$, called the \textbf{instance} of $\provSketch$, to denote $\bigcup_{\range \in \provSketch} \rel_{\range}$. 
\end{Definition}
}

A provenance sketch $\provSketch$ is a compact and declarative description of a superset of the provenance of a query (the instance $\instOf{\provSketch}$ of $\provSketch$). We call a sketch \textit{accurate} if it only contains ranges whose fragments contain provenance. We use $\psSet$ to denote a set of provenance sketches for a subset of the relations in the database. Consider such a set  $\psSet = \{ \provSketch_1, \ldots, \provSketch_m \}$ where $\provSketch_i$ is a sketch for relation $\rel_i$ and $\rel_i \neq \rel_j$ for $i \neq j$. We use $\instOf{\psSet}$ to denote the database derived from database $\db$ by replacing each relation $\rel_i$ for $i \in \{1, \ldots, n\}$ with $\relInst{\provSketch_i}$. Note that we do not require that all relations of $\db$ are associated with a sketch. Abusing notation, we will use $\instOf{\provSketch}$ to denote $\instOf{\{\provSketch\}}$.
Reconsider the running example in \cref{fig:eg-db}, let $\provSketch$ be the accurate provenance sketch of $\qAvgden$ using the range partition $\hparti_{\ranges,state}(citites)$. Recall that $\prov{\qAvgden}{cities} = \{ t_2, t_3 \}$ and tuples $t_2$ and $t_3$ both belong to fragment $f_1$ since $CA \in [AL,DE]$. Thus, $\provSketch = \{ f_1 \}$.
\XN{old content: As an example, let us create an accurate provenance sketch $\provSketch$ for the running example database and query $\qAvgden$ using the range partition $\hparti_{\ranges,state}(citites)$ from the previous example. Recall that $\prov{\qAvgden}{cities} = \{ t_2, t_3 \}$, i.e., the provenance of query $\qAvgden$ consists of the two tuples highlighted in \Cref{tab:cities}. Tuples $t_2$ and $t_3$ both belong to fragment $f_1$ since $CA \in [AL,DE]$.
Thus, $\provSketch = \{ f_1 \}$. 
In addition to $t_2$ and $t_3$, the sketch also contains $t_1$, the other tuple belonging to $f_1$. Furthermore,  any superset of $\provSketch$, e.g., $\provSketch' = \{f_1, f_2\}$, is also a provenance sketch, albeit not an accurate one.}  


\subsection{Sketch Safety}\label{sec:sketch-safety}
By construction we have $\prov{\query}{\db} \subseteq \instOf{\provSketch} \subseteq \db$. Recall that we assume that $\prov{\query}{\db}$ is sufficient, i.e., $\query(\prov{\query}{\db}) = \query(\db)$. However, in general as shown in \Cref{sec:introduction} this does not guarantee that $\query(\instOf{\provSketch}) = \query(\db)$, even when $\provSketch$ is accurate.
We call a provenance sketch \textbf{safe} for a query $\query$ and database $\db$ if evaluating $\query$ over the data described by the sketch returns the same result  as evaluating it over $\db$. 
\begin{Definition}[Safety]\label{def:safe}
\ifnottechreport{Let $\query$ be a query, $\db$ a database, and $\psSet$ a set of sketches for $\query$. We call $\psSet$ \emph{safe} for $\query$ and $\db$ iff $Q(\instOf{\psSet}) = Q(\db)$. }
 \iftechreport{Let $\query$ be a query, $\db$ be a database, and $\psSet$ a set of sketches for $\query$. We call $\psSet$ \emph{safe} for $\query$ and $\db$ iff:
  \begin{align*}
    Q(\instOf{\psSet}) &= Q(\db)
  \end{align*}}
\end{Definition}

\XN{We call a single provenance sketch $\provSketch$ safe if $\{ \provSketch \}$ is safe.}
Obviously, we are only interested in safe provenance sketches. 
Next we will study how to determine safety statically (without accessing the database). For this we define attributes to be safe if for any database instance, sketches created over these attributes are safe.

\begin{Definition}[Attribute Safety]\label{def:attribute-safety}
  Consider a set of  attributes $\attset$ from a relation schema $\relSchema$ of a  database schema $\dbSchema$. We call $\attset$ \emph{safe} for a query $\query$ if for any instance $\db$ of $\dbSchema$ and range partition $\rparti_{\ranges,\attset}$ of $\rel$, any provenance sketch $\provSketch$ based on $\rparti_{\ranges,\attset}$ is safe.
\end{Definition}
\XN{
Next, in \Cref{sec:safety-check}, we discuss how to determine whether a set of
attributes is safe for a given query independent of the database instance. }
\XN{duplicate:
Next, in \Cref{sec:safety-check}, we discuss how to determine the provenance sketches safety.
In
\Cref{sec:reuse-different}, we investigate under which conditions a sketch can
be reused across multiple instances of a parametrized query.  We then discuss
how to capture accurate provenance sketches efficiently in~\Cref{sec:ps-capture}
and how to use a captured sketch to speed-up a query in \Cref{sec:ps-reuse}. }

%% file: sections/safety_check.tex
\section{Testing Sketch Safety}
\label{sec:safety-check}

\XN{In \Cref{sec:introduction} we demonstrated (\Cref{ex:not-safe}) that even accurate sketches may not be sufficient for answering a query.}
In this section, we develop a sound method that determines based on the input query alone whether every sketches for $\query$ and $\db$ created based on some set of ranges $\ranges$ for a set of attributes $\sa$ are guaranteed to be safe. The rationale for avoiding to access $\db$ is that we want to determine upfront whether the sketches on a set of attributes $\sa$ are safe before paying the cost of creating 
such a sketch. 
Instead of accessing $\db$, we use database statistics.
Before presenting our approach we first state a negative result motivating the decision to develop an algorithm that is only sound, but not complete.\ifnottechreport{ \textbf{For space reason, we outsource the proofs of theorems and lemmas in this paper  to \cite{techreport}.}}

\begin{Theorem}\label{theo:sound-and-safe-impossible}
There cannot exist a sound and complete algorithm that determines safety of a set of attributes for a query $\query$ and $\db$ without accessing $\db$.
\end{Theorem}
\iftechreport{
\begin{proof}
We prove this theorem by demonstrating that there exists an attribute $a$, two databases $\db$ and $\db'$, and a query $\query$ such that an accurate provenance sketch created for $\query$ according to some range partitioning of $a$ over $\db$ is safe while the sketch created for $\db'$ is unsafe. Since the algorithm is not allowed to inspect the database it cannot distinguish between $\db$ and $\db'$ and, thus cannot determine whether the attribute is safe for $\query$ for a given database. Consider query $\qAvgden$ from \Cref{tab:queries} and the range partition $\parti_{popden}$. As discussed in \Cref{sec:introduction}, the sketch $\provSketch_{popden}$ created based on this partition is unsafe for $\qAvgden$ and the instance of the cities relation shown in \Cref{tab:cities}. However, consider the database $\db'$  that only consists of tuples $t_2$ and $t_3$ from \Cref{tab:cities}. The provenance sketch for $\qAvgden$ using $\parti_{popden}$ contains the fragment $\frag$ corresponding to range $g_2$ (all tuples with a population density between 4001 and 9000). Since $\frag = \db'$, we have $\query(\frag) = \query(\db')$.
\end{proof}
}


We now introduce our sound, but not complete, algorithm for determining safety of attributes. For a set of attributes $\sa$ and query $\query$, this algorithm constructs a universally quantified logical formula without free variables such that if this formula evaluates to true,
then $\sa$ is safe for $Q$. Similar to recent work on query equivalence checking~\cite{zhou2019automated}, we utilize an SMT solver~\cite{de2008z3} to check whether the formula is true by rewriting it into negated existential form (a universally quantified formula is true if its negation is unsatisfiable). For example, to test $\forall a: a < 10$, we check whether $a \geq 10$ is unsatisfiable.

\subsection{Generalized Containment}\label{sec:gener-cont}
Our approach utilizes a generalization of the subset relationship between two relations to be able to express that, e.g., a count aggregation returns the same groups, but the counts produced by  $\query(\instOf{\psSet})$ (running the query over the provenance sketches) are smaller than the counts for $\query(\db)$.

\begin{Definition}[Generalized Containment]\label{def:generalized-containm}
Let $R(a_1, \ldots,$ \\
$a_n)$ and $R'(b_1, \ldots, b_n)$ be two relations with the same arity. 
Furthermore, let $\aComp = \bigwedge_{j=1}^{m} a_{i_j} \diamond b_{i_j}$ where $i_j \in [1,n]$, $m \leq n$, and 
 $\diamond \in \{ \leq, =, \geq \}$.   The generalized containment relationship $R \matchContains{\aComp} R'$ based on  $\aComp$ holds for $R$ and $R'$ if there exists a mapping $\matchMap \subseteq R \times R'$  that fulfills all of the following conditions:
\vspace{-4pt}  \iftechreport{\begin{align*}
    &\forall t \in R: \exists t' \in R': \matchMap(t,t')\\
    &\forall t_1, t_2, t_1', t_2': \matchMap(t_1, t_1') \land \matchMap(t_2, t_2') \land (t_1 = t_2 \lor t_1' = t_2')\\
      &\hspace{4cm}\rightarrow t_1 = t_2 \wedge t_1' = t_2'\\
    &\forall (t,t') \in \matchMap: (t,t') \models \aComp
  \end{align*}}
\ifnottechreport{
\begin{align*}
    &\forall t \in R: \exists t' \in R': \matchMap(t,t') \quad \forall (t,t') \in \matchMap: (t,t') \models \aComp \\
    &\forall t_1, t_2, t_1', t_2': \matchMap(t_1, t_1') \land \matchMap(t_2, t_2') \land (t_1 = t_2 \lor t_1' = t_2')\\
      &\hspace{4cm}\rightarrow t_1 = t_2 \wedge t_1' = t_2'
  \end{align*}\trimfigspace[-7mm]
}
\end{Definition}

The first and \ifnottechreport{third}\iftechreport{second} conditions ensure that every tuple from $R$ is ``matched'' to exactly one tuple from $R'$. The second condition ensures that all pairs of matched tuples fulfill condition $\aComp$.
Note that $R \subseteq R'$ is a special case of generalized containment where $\aComp = \bigwedge_{i=1}^{n} a_i = b_i$.
In the following, we will use generalized containment to model the relationship between (intermediate) results of a query over the full input database and over the instance of a set of provenance sketches. In this scenario, the two relations we are comparing have the same schema.
To avoid ambiguities in $\aComp$, if $a$ is an attribute of the LHS, we use $a'$ to refer to the corresponding attribute of the RHS.

For instance, reconsider our running example database from \Cref{fig:eg-db} and assume we are using the range partition based on ranges $\{g_1, g_2\}$ for attribute \texttt{popden} to compute an accurate sketch $\provSketch$ for query $\query_{popState} = \selection_{totden > 10000}(\Aggregation{state}{sum(popden) \to totden}(\rela{cities}))$. Let $\query_{agg}$ be the subquery of $\query_{popState}$ rooted at the aggregation operator. California is the only state that fulfills the selection condition. Thus, the provenance of this query is $\{t_2, t_3\}$ and the sketch consists of range $g_2 = [4001,9000]$. We get $\instOf{\provSketch} = \{t_1,t_2,t_3,t_4\}$. Evaluating the aggregation subquery over this instance, we get the original aggregation result ($11,000$) for California, but since $t_4$ and $t_1$ are also in this instance, a result for New York (Alaska) is produced. The result for NY  is lower  ($7000$)  than  the original result ($9000$) for this group. Generalized containment allows us to reason about such cases, e.g., $\query_{agg}(\instOf{\provSketch}) \matchContains{totden \leq totden'} \query_{agg}(\db)$.
\XN{However, this result did already not fulfill the selection condition over the full database and, thus, is filtered out. Generalizing the example, any accurate sketch on a non-group-by attribute for this query will exhibit the same behavior: the aggregation evaluated over the sketch returns the same result over the sketch as over the input database for groups that are in the final result of $\query_{popState}$ and possibly some of the remaining groups with a \texttt{totden} value that is smaller than or equal to the \texttt{totden} for the full input databases (these groups will be filtered by the selection). Generalized containment allows us to reason about such cases, e.g., we can state our observation for the example as the generalized containment relationship $\query_{agg}(\instOf{\provSketch}) \matchContains{totden \leq totden'} \query_{agg}(\db)$.}

\subsection{Inference Rules}\label{sec:inference-rules}
Given a query $\query$ and a set of attributes $\sa$ from the input relations of $\query$, we construct a logical formula $\gc(\query,\sa)$ to check whether $\query(\instOf{\psSet}) = \query(\db)$ for any database $\db$. Let $\db = \{R_1, \ldots, R_n\}$ and $\sa = X_1 \union \ldots \union X_m$, $\psSet$ is a set of provenance sketches $\{\provSketch_i\}$ build on
$R_i$ over attributes $X_i$ according to $\query$ and $\db$ and some range partitions $ \parti_{\ranges_i,X_i}(R_i)$.
This formula is computed based on a set of rules which is shown in \Cref{tab:subset}. Intuitively, $\gc(\query,\sa)$ does encode constraints that have to hold for attribute values of any tuple produced by $\query(\instOf{\psSet})$ and/or by $\query(\db)$. For instance, if the query contains a selection on a condition $a< 10$ then all result tuples of the selection are guaranteed to fulfill $a < 10$. In addition, $\gc(\query,\sa)$ encodes relationships between attribute values of $\query(\instOf{\psSet})$ and $\query(\db)$, e.g., because $\instOf{\psSet}$ is a subset of $\db$, any count aggregate computed over $\instOf{\psSet}$ is smaller than or equal to the same aggregate computed over $\db$.
\XN{Consider a query $\query$ and a set of attributes $\sa$ from these relations. Let $\{R_1, \ldots, R_n \}$ denote the set of relations that contain at least one attribute from $\sa$. We use $\sa_i$ to denote subset of $\sa_i$ that belongs to relation $R_i$. For convenience we will assume attribute names are not repeated across relations.
We now introduce a set of rules that given $\query$ and $\sa$, generate a formula $\gc(\query,\sa)$.

Intuitively, the formula $\gc(\query,\sa)$ we construct checks that for any
database $\db$ we have $\query(\instOf{\psSet_{\db}}) = \query(\db)$ where
$\psSet_{\db}$ is a set of provenance sketches $\{\provSketch_i\}$ build on
$R_i$ over attributes $\sa_i$ according to $\query$ and $\db$ and some range partitioning $ \parti_{\ranges_i,\sa_i}(R_i)$.
Variables in the formula represent attribute values from the LHS (the result of the query over the provenance sketch instance) or the RHS (the result of the query over the databases). We constrain variables based on known restrictions on attribute values (e.g., we know that for any tuple in the output of  a selection $\selection_{a < 10}(R)$ we have $a < 10$) and relationships between tuples from the LHS and RHS (e.g., for $\query \defas R$, we know that for any provenance sketch set $\query(\instOf{\psSet_{\db}}) \subseteq \query(\db)$). As we will prove in~\Cref{sec:correctness-proof}, $\gc(\query,\sa)$ implies that the generalized containment relationship $\query(\instOf{\psSet}) \matchContains{\aComp_{\query,\sa}} \query(\db)$ holds where $\aComp_{\query,\sa}$ is a condition constructed by our algorithm. We construct $\gc(\query,\sa)$ in a bottom-up traversal of query $\query$, incrementally
adding constraints to the formula that reason about how intermediate query
results for $\db$ and $\instOf{\psSet_{\db}}$ are related to each
other.

The rules we use to compute $\gc(\query,\sa)$ are shown in \Cref{tab:subset}.
}
In the following, we first introduce some 
auxiliary constructs which are used to define $\gc(\query,\sa)$ and then discuss the rules that define $\gc$ (\Cref{tab:subset}).
For simplicity of exposition we assume that attribute names are unique.\BG{Do we need to explain that this extends to intermediate results and how this works, e.g., $\projection_{a} (R(a,b))$ is fine, but $\Aggregation{b}{max(c) \to a}(R(a,b,c))$ is not. Furthermore, all variables that occur freely in the formulas are assumed to be universally quantified, e.g., $a < 10$ should be interpreted as $\forall a: a < 10$.}

\parttitle{${\bf \pred(\query)}$} We use $\pred$ to record conditions which have to be fulfilled by all tuples produced by query $\query$ and its subqueries. $\pred$ is defined recursively as shown below.
For instance, selection and join conditions are added to $\pred$, since all tuples produced by such operators have to fulfill these conditions. As an example consider the query $\query \defas \selection_{a = 5} (\projection_{a}(\selection_{b < 4}(\rel)))$, we get $\pred(\query) = (a = 5 \land b < 4)$.
Note that we are using database statistics to bound the values of tuples from input relations. $min(a)$ ($max(a)$) denotes the smallest (largest) value in attribute $a$.
\vspace{-1mm}
\resizebox{.9\linewidth}{!}{
  \begin{minipage}{\linewidth}
\begin{align*}
  \pred(\query)                                                   & =
             \begin{cases}
              \bigwedge_{ a \in SCH(R)} a \geq min(a) \land a \leq max(a)               & \mathbf{if}\,Q \text{ is a relation}     \\
             \pred(\querya) \wedge \pred(\queryb)               & \mathbf{if}\,Q = \querya \crossprod \queryb     \\
             \pred(\querya) \wedge \pred(\queryb) \wedge \theta & \mathbf{if}\,Q = \querya \join_{\theta} \queryb  \\
             \pred(\querya) \land \theta                         & \mathbf{if}\, Q = \selection_{\theta}(\querya)   \\
             \pred(\querya) \lor \pred(\queryb)               & \mathbf{if}\,Q = \querya \union \queryb   \\
             \pred(\querya)                                      & \mathbf{otherwise}                               \\ 
           \end{cases}
\end{align*}
  \end{minipage}
}
\BG{What about set operations?}

\input{sections/gc_cond.tex}

\parttitle{${\bf \expr(\query)}$}
This formula encodes relationships  between values of attributes in the result of the query and its subqueries. For every generalized projection, we record how the value of attributes in the output of the projection are related to the values of attributes in its input.
For example, for $\query \defas \projection_{a+b \rightarrow x, c+d \rightarrow y}$, we get $\expr(\query) = (a+b=x \land c+d=y)$.

\vspace{-2mm}
\resizebox{.9\linewidth}{!}{
  \begin{minipage}{\linewidth}
\begin{align*}
\expr(\query)                                                         & =
                \begin{cases}
                 \emptyset              & \mathbf{if}\,Q \text{ is a relation}     \\
                  \expr(\querya) \land \expr(\queryb)               & \mathbf{if}\,\query = \querya \crossprod \query \,\,\mathbf{or}\,\, \querya \join_\theta \queryb \\
                  \expr(\querya) \land \bigwedge_{i=1}^{n} e_i = b_i & \mathbf{if}\,\query = \projection_{e_1 \to b_1, \ldots, e_n \to b_n}(\querya)          \\
                  \expr(\querya) \lor \expr(\queryb)               & \mathbf{if}\,Q = \querya \union \queryb  \\
                  \expr(\querya)                                     & \mathbf{otherwise}\\
                \end{cases}
\end{align*}
  \end{minipage}
}

\parttitle{${\bf \aComp_{\query,\sa}}$ and $\gc(\query,\sa)$}
In $\gc(\query,\sa)$ we make use of $\aComp_{\query,\sa}$ which relates attributes from $\query(\instOf{\psSet})$ to attributes from $\query(\db)$. Furthermore, we use $\allcond(\query)$ to denote $\pred(\query) \land \expr(\query)$.
We will show that the generalized containment $\query(\instOf{\psSet}) \matchContains{\aComp_{\query,\sa}} \query(\db)$ holds if $\gc(\query,\sa)$ is valid. Importantly, we will then prove that this implies that $\query(\instOf{\psSet}) = \query(\db)$. One subtlety related to that is that for any subquery $\query_{sub}$ of the query $\query$ for which are testing safety, $\gc(\query_{sub},\sa) \Rightarrow  \query_{sub}(\instOf{\psSet}) \matchContains{\aComp_{\query_{sub},\sa}} \query_{sub}(\db)$. However, $\gc(\query_{sub}$ \\
$,\sa) \Rightarrow \query_{sub}(\instOf{\psSet}) = \query_{sub}(\db)$ does not hold in general.\BG{Not sure whether this needs to be explained here or later.}
The rules for deriving $\aComp_{\query,\sa}$ and $\gc(\query,\sa)$ as shown in \Cref{tab:subset}.
For example, $\aComp_{R,\sa}$ is equality on all attributes of $R$, because $\instOf{\psSet} \subseteq \db$. Recall that we use $a, b, \ldots$ to denote attributes from $\query(\instOf{\psSet})$ and $a',b',c', \ldots$ to denote the corresponding attributes from $\query(\db)$. We apply the same notation for queries and conditions, e.g.,  $\query$ and $\query'$ denote the same query over $\instOf{\psSet}$ and $\db$, respectively.
We use $\attrs{\query}$ to denote the set of attributes of the relations accessed by $\query$. Furthermore, for a subquery $\query_1$ ($\query_2)$ we use $\sa_{1}$ ($\sa_2$) to denote the subset of $\sa$ contained in relations accessed by $\query_1$ ($\query_2$).
For any subquery $\query$ that does not contain of attributes for which we want to test safety ($\sa = \emptyset$), we know that $\instOf{\psSet}$ contains the original relations from $\db$. Thus, $\query(\instOf{\psSet}) = \query(\db)$ for any such subquery and we set $\aComp{\query,\sa}$ to the equality on all attributes and $\gc(\query,\sa) = \true$.
\ifnottechreport{Because of space restrictions, we only briefly discuss some of the rules. For a
detailed description of all rules see \cite{techreport}. For most operators,
$\aComp_{\query,X}$ is the same as $\aComp_{\query_1,\sa_1}$ (or
$\aComp_{\query_1,\sa_1} \land \aComp_{\query_2,\sa_2}$ for joins).} 
\iftechreport{For selection, duplicate removal and top-k operators, no additional attributes are created by these operators, thus $\aComp_{\query,X}$ is the same as $\aComp_{\query_1,\sa_1}$.  For a projection $\projection_A(\query_1)$ where $A = e_1 \to b_1, \ldots, e_n \to b_n$, some additional attributes might be created based on existing attributes by renaming expressions. For these attributes, we could use $\aComp_{\query_1,\sa_1}$ and $\expr(\query_1) \land \expr(\query_1')$ to decide the relationship of them between $\projection_A(\query_1(\instOf{\psSet}))$ and $\projection_A(\query_1(\db))$.} For a union
$\query_1 \union \query_2$, only constraints that hold in both
$\aComp_{\query_1,X_1}$ and $\aComp_{\query_2,\sa_2}$ hold for the result of the
union. For aggregation, we consider several cases: (i) if $\sa_1$ ($\sa$
restricted to attributes from relations below the aggregation) is a subset of
the aggregation's group-by attributes, then calculating the aggregation function
over the sketch instance yields the same result as over the full database,
because each group is contained in exactly one fragment of the partition on
which the provenance sketch is build on. Thus, either all or none of the tuples
of a group are included in $\instOf{\psSet}$ and for all groups included in
$\instOf{\psSet}$, the aggregation function result will be same in
$\query(\instOf{\psSet})$ and $\query(\db)$; (ii) for aggregation functions that are monotone (e.g., count, max, or sum over positive numbers) we know that the
aggregation function result produced for a group that occurs in
$\query(\instOf{\psSet})$ has to be smaller than or equal to the result for the
same group in $\query(\db)$. Thus, if the constraints we have derived for the
input of the aggregation imply that the input attribute $a$ for the aggregation
function $a$ is larger than $0$, then $b < b'$ holds; (iii) the third case
handles min and sum aggregation over negative numbers.

The $\gc(\query,\sa)$ condition for most operators requires that $\gc$ holds for the operator's input. Additionally, operator-specific conditions apply. For example, for selection the condition $\theta$ evaluated on any tuple from $\instOf{\psSet}$ has to imply that the condition holds for the corresponding tuple over $\db$ ($\theta'$) (We use  $\db$ ($\theta'$) to represent the evaluation of $\theta'$ over $\db$), because that implies generalized containment. For aggregation, generalized containment holds if the group-by attributes are the same over $\instOf{\psSet}$ and $\db$. Similarly, for the top-k operator\footnote{Note that for the top-k operator $\ordlimit{O}{C}$, our rules require  the assumption that the input of this operator returns at least $C$ tuples. For such queries we have to re-validate the result of our safety check at runtime.} the order-by attributes have to the same to ensure that the smallest items are the same for $\query$ and $\query'$. For union and cross-product we only need to require that $\gc$ holds for both inputs of these operators. For join, if the join attributes for  $\instOf{\psSet}$ and $\db$ are equal then (and $\gc$ holds for both inputs), then it also holds for the result of the join.
\BG{Join is needed?}
\BG{$\sa \cap \attrs{\query} = \emptyset$ represents no provenance sketches are applied to this subquery $\query$. For convenience, we let $\sa_1 = \sa \cap \attrs{\query_1}$, for binary operator, we let $\sa_2 = \sa \cap \attrs{\query_2}$. Also, we assume the schemas are same between the both inputs of union operator. For aggregation, we reason different cases for deciding the relationship of the function attributes between LHS and RHS. For the top-k operator $\ordlimit{O}{C}$, our rules under the assumption that the query returns exactly $C$ tuples. Otherwise, it is data dependent.}
\XN{
\parttitle{${\bf \aComp_{\query,\sa}}$}
Recall that this formula was introduced in the generalized containment which is constructed by a  conjunction of comparison atoms, e.g., $\query(\instOf{\psSet}) \matchContains{\aComp_{\query,\sa}} \query(\db)$ keeps that all "matched" tuples between $\query(\instOf{\psSet})$ and $\query(\db)$ have to satisfy $\aComp_{\query,\sa}$. 
Since we want to use $\gc(\query,\sa)$ to reason about $\query(\instOf{\psSet}) \matchContains{\aComp_{\query,\sa}} \query(\db)$, when the rules in \Cref{tab:subset} imply $\gc(\query,\sa)$ holds, they also compute the $\aComp_{\query,\sa}$ it holds for. As an example, $\aComp_{\query,\sa}$ is initialized for leaf nodes of the algebra tree of a query (relation accesses $R$) as equality on all attributes of the relation accessed by the leaf node. Then, $R_{\psSet} \matchContains{\aComp_{R,\sa}} R$ and all "matched" tuples between $R_{\psSet}$ and $R$ satisfy for $\aComp_{R,\sa}$. However, for some operators like aggregation, we need to reason different cases to decide $\aComp$. }

\XN{
This formula is a conjunction of comparison atoms such that if $\gc(\query,\sa)$ holds, then $\query(\instOf{\psSet_{\db}}) \matchContains{\aComp_{\query,\sa}} \query(\db)$ holds for any database $\db$ and set of provenance sketches $\psSet_{\db}$ according to $\sa$. Recall that $X$ is the set of attributes for which we have provenance sketches in $\psSet_{\db}$. For binary operators, we use $X_1$ ($X_2$) to denote the subsets of $X$ that belong to the subtree rooted at the left (right) input of the operator. $\aComp_{\query,\sa}$ is initialized for leaf nodes of the algebra tree of a query (relation accesses) as equality on all attributes of the relation accessed by the leaf node. For any relation access $R$, the query evaluated over $\psSet_{\db}$ returns either the same result as the input query (if $\psSet_{\db}$ does not contain a sketch for $R$) or a subset. As mentioned above $R \subseteq R'$ is the same as generalized containment on a conjunction of equality constraints (one for each attribute of $R$). For any other operator that  is the root of a subquery that does not contain any attribute from $X$, all relation accesses in this subtree return the  same result over $\instOf{\psSet_{\db}}$ as in the original query, because no provenance sketches are applied to this subtree. In this case we get the same result for both queries which implies that generalized containment for equalities for all attributes of the operator's schema hold. For operators for which this condition does not hold, we propagate or combine the conditions of their inputs. The only exception is aggregation.  For aggregation we can infer that  the aggregate function results for all tuples shared among the outputs of both queries are equal if the provenance sketch attributes are a subset of the group-by attributes for the query. This is true, because in this case, the provenance sketch instance contains a subset of the groups of the input query and, thus, both queries compute the same results for any such group. If this is not the case, then the query over the provenance sketches may return subsets of groups produced by the query over the full database. For aggregation function $count$ which is monotone in the number of inputs, this implies that the aggregation function result returned for a subset is less than the aggregation function result for the whole group. The same applies for $max$ and for $sum$, but only if the inputs are all positive. For aggregation function min and sum  over negative inputs, the results would guaranteed to be greater than the result over the full group. If no of these cases applies, then we do not know how the aggregation function results over the  provenance sketch instance and over the full  input database relate  to each other.}
\XN{
\parttitle{Computing ${\bf \gc}$}
Using these three auxiliary properties, we define a set of rules in \Cref{tab:subset} to examine  whether $\gc(\query,\sa)$ holds.  }
\XN{Using these three auxiliary properties we define $\gc(\query,X)$, which encodes a sufficient condition for the safety of $X$ for query $\query$.}
\BG{We need to say that we constructing a condition for every attribute $x$, make conditions parameterized by attribute $x$}
\BG{For multiple attributes currently not yet covered.}
\BG{Is in general it always safe to combine two provenance sketches that are safe individually, e.g., $R.a$ is safe and $S.b$ is safe. In general, this may fail: aggregation with filter on smaller than and before join of $R$ and $S$. Individually, we may still produce ``enough'' join partners so that groups that are not in the provenance are still filtered out, but in combination this may not be the case. However, our sound condition should ensure that combinations are safe! (conjecture)}
\BG{Better clarify the purpose of each of the ``properties''}
%
\BG{Check, the following may be outdated.}
%
\BG{update notation in the example}
\begin{Example}\label{ex:safety-inference-rules}
Reconsider our running example database in~\Cref{fig:eg-db} and assume we are using the range partition based on ranges $\{g_1, g_2\}$ on attribute \texttt{popden} to compute an accurate sketch $\provSketch$ for query $\query_{popState} = \selection_{totden < 7000}(\Aggregation{state}{sum(popden) \to totden}(\rela{cities}))$. 
To determine whether \texttt{popden} is a safe attribute, we calculate $\gc(\query_{pop}$ \\$_{State},\{popden\})$ using the rules from~\Cref{tab:subset}. 
We start from relation $\rela{cities}$. Since $popden>0$, we get $\pred(cities) = popden > 0$, $\expr(cities) = \emptyset$, $ \aComp_{cities,\{popden\}} = popden=popden' \land city=city' \land state=state'$, and $ \gc(cities,\{popden\})$ evaluates to true. Let $\query_{agg}$ be the subquery of $\query_{popState}$ rooted at the aggregation operator. Since $state=state'$,  based on the rules for aggregation rules (\Cref{tab:subset}),
\iftechreport{$\gc(Q_{agg},\{popden\})$ holds for $\aComp_{Q_{agg},\{popden\}} = popden=popden' \land city=city' \land state=state' \land f \leq f'$ where $f \leq f'$ is computed based on aggregation rule in  \Cref{tab:acomp}. Also, $\pred(Q_{agg}) = popden > 0$ and $\expr(Q_{agg}) = \O$.}
\ifnottechreport{we add $totden \leq totden'$ to $\aComp_{Q_{agg},\{popden\}}$.}
At last, applying the rule for selection, we get $totden \leq totden' \land totden < 7000 \not \Rightarrow totden' < 7000$. Thus,  $\gc(\query_{popState},\{popden\})$ evaluates to false. That is, as expected based on our previous discussion of this example, we determine that \texttt{popden} may be unsafe (our rules are only sound, but not complete).
\end{Example}
\iftechreport{
  \input{sections/safety_correctness.tex}
  }
\ifnottechreport{
  We now proceed to formally state the correctness of our safety checking algorithm (\Cref{theo:safety-check-is-correct}), i.e., if $\gc(\query,\sa)$ holds, then $\sa$ is a safe set of attributes for $\query$.
\begin{Theorem}[$\gc(\query,\sa)$ implies safety of $\sa$]\label{theo:safety-check-is-correct}
  Let $\query$ be a query, $\db$ be a database, and $\sa = \bigcup_{1}^{n} \sa_i$ a set of attributes where each $\sa_i$ belongs to a relation $R_i$ accessed by $\query$ such that $R_i \neq R_j$ for $i \neq j$. 
If $\gc(\query,\sa)$ holds, then $\sa$ is a safe set of attributes for $\query$.
\end{Theorem}
 \begin{proofsketch}
The claim is proven by first proving two lemmas that state that (i) $\gc(\query,\sa)$ implies $\gc(\query',X')$ for any subquery of $\query$ (this follows trivially from the definition of $\gc$) and that (ii) $\gc(\query',\sa)$ implies $\query'(\instOf{\psSet}) \matchContains{\aComp} \query'(\db)$ for any subquery $\query'$ of $\query$. (ii) is proven by induction of the structure of a query. Then based on these results we prove the theorem by demonstrating that $\gc(\query,\sa)$ together with the fact that $\instOf{\psSet}$ contains that provenance of $\query$ implies the claim.
\end{proofsketch}
}


%% file: sections/gc_cond.tex
\begin{figure*}
\fbox{
\begin{minipage}{1\linewidth}
\begin{minipage}{1\linewidth}
 \begin{subfigure}{0.41\linewidth}
\vspace{-3mm}
\begin{adjustbox}{max width=1\linewidth}
{
\begin{tabular}{r@{ = }l}
 \multicolumn{2}{l}{\underline{\textbf{if $\sa = \emptyset$}}} 
  \\
 $\aComp_{\query,\sa}$    & $\bigwedge_{a \in \schemaOf{\query}} a = a'$                                                                       \\
 $\gc(\query,\sa)$ & $\true$ \\
\multicolumn{2}{l}{\underline{\textbf{otherwise}}}                                                                                                  \\
  $\aComp_{R,\sa} $                           & $\bigwedge_{a \in \schemaOf{R}} a = a'$                                                             \\
  $\aComp_{\selection_{\theta}(\querya),\sa}$ & $ \aComp_{\projection_A(\querya),\sa} =  \aComp_{\querya,\sa_1}$                                    \\
  $\aComp_{\duprem(\querya),\sa}$             & $ \aComp_{\ordlimit{O}{C}(\querya),\sa} = \aComp_{\querya,\sa_1}$                                   \\
  $\aComp_{\querya \crossprod \queryb,\sa} $   & $\aComp_{\querya \join_{\theta} \queryb,\sa} = \aComp_{\querya,\sa_1} \land \aComp_{\queryb,\sa_2}$ \\
  $\aComp_{\querya \union \queryb,\sa}$ & $\bigwedge_{i=1}^{n} (\aComp_{\querya,\sa_1} \rightarrow a_i = a_i' \land \aComp_{\querya,\sa_2} \rightarrow b_i = b_i') \rightarrow a_i = a_i'$\\
\multicolumn{2}{l}{\hspace{9mm}\textbf{where} $\schemaOf{\querya} = (a_1, \ldots, a_n)$ \textbf{and} $\schemaOf{\queryb} = (b_1, \ldots, b_n)$}\\
\end{tabular}
}
\end{adjustbox}
\end{subfigure}
 \begin{subfigure}{0.59\linewidth}
\centering
\begin{adjustbox}{max width=1\linewidth}
{
\begin{tabular}{|ll|}
\hline
\rowcolor{lightgrey}
Query $\query$ & $\gc(\query,\sa)$ \\ 
		$R$ &  $\true$  \\
  $\selection_{\theta}(\querya)$ &   $\gc(\querya,\sa_1) \land (\aComp_{\querya,\sa_1} \land \allcond(\querya') \land \allcond(\querya) \land \theta \rightarrow \theta') $ 
  \\
  $\Aggregation{\aggf(a) \rightarrow \agga}{\grpatts}(\querya)$ & $\gc(\querya,\sa_1) \land (\forall g \in G : \aComp_{\querya,\sa_1} \land \allcond(\querya') \land \allcond(\querya)  \rightarrow g = g')  $ 
  \\
  $\delta(\querya)$ & $\gc(\querya,\sa_1) \land (\forall a \in \schemaOf{\querya} :\aComp_{\querya,\sa_1} \land \allcond(\querya') \land \allcond(\querya) \rightarrow a = a') $ 
  \\
$\projection_{A}(\querya)$ & $\gc(\querya,\sa_1)$
  \\
  $\ordlimit{O}{C}(\querya)$ & $\gc(\querya,\sa_1) \land (\forall o \in O : \aComp_{\querya,\sa_1} \land \allcond(\querya') \land \allcond(\querya) \rightarrow o = o')
                               $   
  \\
    $\querya \union \queryb$ & $\gc(\querya,\sa_1) \land \gc(\queryb,\sa_2)$ \\
		$ \querya \crossprod \queryb$ & $\gc(\querya,\sa_1) \land \gc(\queryb,\sa_2)$
  \\ 
  $\querya \join_{a=b} \queryb$ & $\gc(\querya,\sa_1) \land \gc(\queryb,\sa_2) \land $ 
 $	\aComp_{\querya,\sa_1} \land \allcond(\querya') \land \allcond(\querya) \rightarrow a=a'  \land$ \\
&  $	 \aComp_{\queryb,\sa_2} \land \allcond(\queryb') \land \allcond(\queryb) \rightarrow b=b'$ 
  \\
		\hline
\end{tabular}
}
\end{adjustbox} \\[-1mm]
     \caption{$\gc(\query,\sa)$}
     \label{tab:gc}
\end{subfigure}
 \end{minipage}  \\[-5mm]
  \begin{subfigure}{1\linewidth}
 \resizebox{0.8\linewidth}{!}{
\begin{minipage}{1\linewidth}
\begin{align*}
\aComp_{\Aggregation{\aggf(a) \to \agga}{\grpatts}(\querya),\sa} &=
  \begin{cases}
     \aComp_{\querya,\sa_1} \land \agga = \agga' & \text{\textbf{if} } \forall x \in\sa_1 \exists g \in G: \allcond(\query_1) 
     \rightarrow x = g \\ 
    \aComp_{\querya,\sa_1} \land \agga \leq \agga' & \text{\textbf{if} } \exists x: x \in \sa_1 \land x \not \in G \land (\aggf=count \lor (\aggf \in \{sum,max\} \land (\allcond(\querya)  
    \rightarrow a \geq 0)))\\
      \aComp_{\querya,\sa_1} \land \agga \geq \agga' & \text{\textbf{if} } \exists x: x \in \sa_1 \land x \not \in G \land (\aggf \in \{sum,min\} \land (\allcond(\querya) 
      \rightarrow a \leq 0)) \\
    \aComp_{\querya,\sa_1} & \text{\textbf{otherwise}}
    \end{cases}
\end{align*}
 \end{minipage}
 }\\[-4mm]
       \caption{$\aComp_{\query,\sa}$}
      \label{tab:acomp}
 \end{subfigure}
  \end{minipage}
  }
  \trimfigspace
      \caption{Bottom-up inference of condition $\gc(\query,\sa)$. This condition implies $\query(\instOf{\psSet}) \matchContains{\aComp_{\query,\sa}} \query(\db)$ if $\query$ is a subquery of the query for which we want to determine sketch safety and $\query(\instOf{\psSet}) = \query(\db)$ iif $\query$ is the query for which we want to determine sketch safety.}
      \label{tab:subset}
\end{figure*}


%% file: sections/safety_correctness.tex
  \subsection{Correctness Proof}
\label{sec:correctness-proof}
\iftechreport{
We now proceed to prove the correctness of our safety checking algorithm (\Cref{theo:safety-check-is-correct}), i.e., if $\gc(\query,\sa)$ is valid, then $\sa$ is a set of attributes safe for $\query$. Before proving our main result we will prove three lemmas that will utilize in this proof.
First we prove that, for any query $\query = op(\querya, \ldots, \query_m)$ where  $op$
is an  operator, by construction, $\aComp_{\query,\sa}$ implies $\aComp_{\query_i,\sa_i}$ and that $\gc(\query,\sa) \Rightarrow \gc(\query_i,\sa_i)$ where $\sa_i$ is the set of attributes from $\sa$ that belong to relations accessed by $\query_i$. This will be used in the proof of following two lemmas.
We then prove that  $\gc(\query,\sa)$ implies generalized containment: $\gc(\query,\sa) \Rightarrow \query(\instOf{\psSet}) \matchContains{\aComp} \query(\db)$. Afterwards, we prove that $\gc(\query,\sa)$ implies that $ \query(\db)$ is a subset of $\query(\instOf{\psSet})$: $\gc(\query,\sa) \Rightarrow \query(\instOf{\psSet}) \supseteq \query(\db)$. }
\ifnottechreport{
  We now proceed to prove the correctness of our safety checking algorithm (\Cref{theo:safety-check-is-correct}), i.e., if $\gc(\query,\sa)$ is valid, then $\sa$ is a safe set of attributes for $\query$. Before proving our main result we will prove two lemmas that will utilize in this proof.
\Cref{lem:gc-imply-gcontainment} shows that $\gc(\query,\sa) \Rightarrow \query(\instOf{\psSet}) \matchContains{\aComp_{\query,\sa}} \query(\db)$ and \Cref{lem:gc-regenerate} that $\gc(\query,\sa) \Rightarrow \query(\db) \subseteq \query(\instOf{\psSet})$. 

}
\iftechreport{
\begin{Lemma}\label{lem:gc}
  Let $\db$ be a database, $\query = op(\querya, \ldots, \query_n)$ a query, and $\sa$ a set of attributes from the schema of $\db$. Furthermore, let $\sa_i \subseteq \sa$ be the subset of $\sa$ contained in the relations accessed by $\query_i$. Then, for all $i \in \{1,\ldots, n\}$, 
 \begin{align*}
\gc(\query, \sa)  \Rightarrow \gc(\query_i, \sa_i)     \quad
\aComp_{\query,\sa}                  \Rightarrow \aComp_{\query_i,\sa_i}
  \end{align*}
\end{Lemma}
\ifnottechreport{
\begin{proof}
For the full proof see \cite{techreport}.
\end{proof}}
\iftechreport{
\begin{proof}
For the rules in \Cref{tab:subset}, $\gc(\query, \sa)$ is based on $\gc(\querya, \sa_1) \land \ldots \land \gc(\query_n, \sa_n) $. Thus, this lemma holds.
\end{proof}}
}
\BG{state how this will be used afterwards.}
Recall in \Cref{sec:gener-cont} we defined general containment to 
model the relationship between (intermediate) results of $\query(\instOf{\psSet})$ and $\query(\db)$ and designed $\gc(\query,\sa)$ rules in \Cref{tab:subset} to trace the evolution of the general containment. Thus now we prove that $\gc(\query,\sa) \Rightarrow \query(\instOf{\psSet}) \matchContains{\aComp_{\query,\sa}} \query(\db)$.

\begin{Lemma}\label{lem:gc-imply-gcontainment}
  Let $\query$ be a query, $\db$ be a database, and $\sa = \bigcup_{1}^{n} \sa_i$ a set of attributes where each $\sa_i$ belongs to a relation $R_i$ accessed by $\query$ such that $R_i \neq R_j$ for $i \neq j$. Given a set of provenance sketches $\psSet = \{ \provSketch_i \}$  for  $\query$ over $\db$ with respect to a set of range partitions $\{ \parti_{\ranges_i,\sa_i}(R_i) \}$, then
  \ifnottechreport{$\gc(\query,\sa) \Rightarrow \query(\instOf{\psSet}) \matchContains{\aComp_{\query,\sa}} \query(\db)$.}
  \iftechreport{
\[ \gc(\query,\sa) \Rightarrow \query(\instOf{\psSet}) \matchContains{\aComp_{\query,\sa}} \query(\db) \]}
%
\end{Lemma}
\ifnottechreport{
\begin{proof}
For the full proof see \cite{techreport}.
\end{proof}
}
\iftechreport{
\begin{proof}
We prove this by induction.
For convenience, given tuple $t \in \query(\instOf{\psSet})$ and $t' \in \query(\db)$, 
we say $t$ \emph{mapped} with $t'$, if $(t,t') \in \mathcal{M}$ where $\mathcal{M} \subseteq \query(\instOf{\psSet}) \times \query(\db)$.
Also, given a set of attributes $A = \{a_1, \ldots, a_n\}$ and  $A' = \{a_1', \ldots, a_n'\}$, for convenience, we say $A = A'$ if $a_1 = a_1', \ldots, a_n = a_n'$.

\underline{Base case:} We start from the relation access $R$. Assume $R_{\psSet}$ represents the provenance sketches of $R$ with respect to a set of range partitions $\{ \parti_{\ranges,\sa}(R) \}$, then $R_{\psSet} \subseteq R$, which is a special case of  $R_{\psSet} \matchContains{\aComp_{R,\sa}} R$ where $\aComp$ only contains equalties for each column between the schema of $R_{\psSet}$ and $R$. Therefore, $R_{\psSet} \matchContains{\aComp_{R,\sa}} R$.

\underline{Inductive step:}
Assume query $\query_n$, a sub query of $\query$, with depth 
less than or equal to $n$.
Here we use the depth to represent the levels in the relational algebra tree of a query, e.g., $\query_1$ represents the  relation access $R$.
Assume we have proven that 
$\gc(\query_n, \sa_n) \Rightarrow \query_n(\instOf{\provSketch}) \matchContains{\aComp_{\query_n,\sa_n}} \query_n(\db)$.
Based on the induction hypothesis, we have to prove 
$\gc(\query_{n+1}, \sa_{n+1}) \Rightarrow \query_{n+1}(\instOf{\provSketch}) \matchContains{\aComp_{\query_{n+1},\sa_{n+1}}} \query_{n+1}(\db)$.   
For notational convenience, we use $op_{n}$($op_{n+1}$) represent any operator of depth $n$ ($n + 1$), i.e., 
$\query_{n+1} = op_{n+1}(\query_{n})$.
From \Cref{lem:gc}, $\gc(\query_{n+1}, \sa_{n+1}) \Rightarrow \gc(\query_n, \sa_n)$. Then based on the assumption, $\query_n(\instOf{\provSketch}) \matchContains{\aComp_{\query_n,\sa_n}} \query_n(\db)$. For join operator, we assume the left and right input of $op_{n+1}$ are $\query_{nL}(\db)$ and $\query_{nR}(\db)$ respectively, then $\gc(\query_{nL}, \sa_{nL}) = \true$ and $\gc(\query_{nR}, \sa_{nR}) = \true$, and  \\
$\query_{nL}(\instOf{\provSketch}) \matchContains{\aComp_{\query_{nL},\sa_{nL}}} \query_{nL}(\db)$ and $\query_{nL}(\instOf{\provSketch}) \matchContains{\aComp_{\query_{nL},\sa_{nL}}} \query_{nL}(\db)$.
Now we consider different $op_{n+1}$:

\underline{$\selection_{\theta}$:} Since 
$\query_n(\instOf{\psSet}) \matchContains{\aComp_{\query_n,\sa_n}} \query_n(\db)$, then assume tuple $t \in \query_n({\instOf{\psSet}})$, then there have to be a tuple $t' \in \query_n(\db)$ such that 
$t$ mapped with $t'$ and $(t,t') \models \aComp_{\query_n,\sa_n}$.
Then $t$ satisfies $ \allcond(\query_n)$ and $t'$ satisfies $\allcond(\query_n')$.
Assume after $\selection_{\theta}$, $t$ exists in the result of $\selection_{\theta}(\query_n({\instOf{\psSet}}))$, that is $t \in \selection_{\theta}(\query_n({\instOf{\psSet}}))$, thus $t$ satisfies $\theta$.
Since $\aComp_{\query_n,\sa_n} \land \allcond(\query_n') \land \allcond(\query_n) \land \theta \rightarrow \theta'$ in \Cref{tab:subset}, then $t'$ satisfies $\theta'$, i.e., $t' \in \selection_{\theta}(\query_n(\db))$. Therefore, $op_{n+1}(\query_n(\instOf{\psSet})) \matchContains{\aComp_{\query_{n+1},\sa_{n+1}}} op_{n+1}(\query_n(\db))$ where $\aComp_{\query_{n+1},\sa_{n+1}} = \aComp_{\query_n,\sa_n}$ and $\sa_{n+1} = \sa_n$. 

\underline{$\delta$:} Assume $\exists t \in \delta(\query_n({\instOf{\psSet}}))$, then $t \in \query_n({\instOf{\psSet}})$.
Since \\
$\query_n(\instOf{\psSet}) \matchContains{\aComp_{\query_n,\sa_n}} \query_n(\db)$,
then there have to be a tuple $t' \in \query_n(\db)$ such that 
$t$ mapped with $t'$ and $(t,t') \models \aComp_{\query_n,\sa_n}$.
Thus, in \Cref{tab:subset}, the duplicate removal rule guarantees that $ \forall a \in \schemaOf{\query_n} :\aComp_{\query_n,\sa_n} \land \allcond(\query_n') \land \allcond(\query_n) \Rightarrow a = a' $ , that is $t = t'$.
And because $t' \in \delta(\query_n({\db}))$, $op_{n+1}(\query_n(\instOf{\psSet})) \subseteq op_{n+1}(\query_n(\db))$. Then same with the base case which is a special case of the general containment, thus $op_{n+1}(\query_n(\instOf{\psSet})) \matchContains{\aComp_{\query_{n+1},\sa_{n+1}}} op_{n+1}(\query_n(\db))$ where $\aComp_{\query_{n+1},\sa_{n+1}}$ only contains equalties. 

\underline{$\projection_A$:} Similarly,  assume $\exists t_{proj} \in \projection_A(\query_n({\instOf{\psSet}}))$. Then,  there have to be a tuple $t \in \query_n({\instOf{\psSet}})$ such that $\projection_A(\{t\}) = \{t_{proj}\}$.
Since 
$\query_n(\instOf{\psSet}) \matchContains{\aComp_{\query_n,\sa_n}} \query_n(\db)$,
then there have to be a tuple $t' \in \query_n(\db)$ such that 
$t$ mapped with $t'$ and $(t,t') \models \aComp_{\query_n,\sa_n}$.
That it, there have to be a tuple $tt'$ satisfies $\{t_{proj}'\}= \projection_A(\{t'\})$ such that $(t_{proj},t_{proj}')$ satisfies $\aComp_{\query_{n},\sa_{n}}$.
Thus, $op_{n+1}(\query_n(\instOf{\psSet})) \matchContains{\aComp_{\query_{n+1},\sa_{n+1}}} op_{n+1}(\query_n(\db))$ where $\aComp_{\query_{n+1},\sa_{n+1}} = \aComp_{\query_{n},\sa_{n}}$.

\underline{$\query_{Ln}(\db) \union \query_{Rn}(\db)$:}
Assume $\exists t_L \in \query_{Ln}(\instOf{\psSet})$, since \\ $\query_{Ln}(\instOf{\psSet}) \matchContains{\aComp_{\query_{Ln},\sa_{Ln}}} \query_{Ln}(\db)$, then there have to be a tuple $t_L' \in \query_{Ln}(\db)$ such that $t_L$ mapped with $t_L'$ and $(t_L,t_L') \models \aComp_{\query_{Ln},\sa_{Ln}}$. After union, $t_L \in \query_{Ln}(\instOf{\psSet}) \union \query_{Rn}(\instOf{\psSet})$ and $t_L' \in \query_{Ln}(\db) \union \query_{Rn}(\db)$.
 Assume $\exists t_R \in \query_{rn}(\instOf{\psSet})$, similarly, we could get that there have to be a tuple $t_R' \in \query_{Rn}(\db)$ such that $t_R$ mapped with $t_R'$ and $(t_R,t_R') \models \aComp_{\query_{Rn},\sa_{Rn}}$. After union, $t_R \in \query_{Ln}(\instOf{\psSet}) \union \query_{Rn}(\instOf{\psSet})$ and $t_R' \in \query_{Ln}(\db) \union \query_{Rn}(\db)$. Then $\{t_L,t_R\} \subseteq \query_{Ln}(\instOf{\psSet}) \union \query_{Rn}(\instOf{\psSet})$ and $\{t_L',t_R'\} \subseteq \query_{Ln}(\db) \union \query_{Rn}(\db)$. So now the generalized containment will only hold for the common part between $\aComp_{\query_{Ln},\sa_{Ln}}$ and $\aComp_{\query_{Rn},\sa_{Rn}}$. That is, $op_{n+1}(\query_n(\instOf{\psSet})) \matchContains{\aComp_{\query_{n+1},\sa_{n+1}}} op_{n+1}(\query_n(\db))$ where $\aComp_{\query_{n+1},\sa_{n+1}} = \bigwedge_{a=a' \in \aComp_{\query_{Ln},\sa_{Ln}} \land a=a' \in \aComp_{\query_{Rn},\sa_{Rn}}} a=a'$.

\underline{$\Aggregation{\aggf(a) \rightarrow \agga}{\grpatts}$:}  
Since 
$\query_n(\instOf{\psSet}) \matchContains{\aComp_{\query_n,\sa_n}} \query_n(\db)$, then assume $\exists t \in \query_n({\instOf{\psSet}})$, there has to be a tuple $t' \in \query_n(\db)$ such that 
$t$ mapped with $t'$ and $(t,t') \models \aComp_{\query_n,\sa_n}$.
And because of the aggregation rule in \Cref{tab:gc} which keeps that $\forall g \in G: g=g'$, then after aggregation,
we will still get $\query_{n+1}(\instOf{\psSet}) \matchContains{\aComp_{\query_{n+1},\sa_{n+1}}} \query_{n+1}(\db)$ on $\aComp_{\query_{n+1},\sa_{n+1}} = \aComp_{\query_{n},\sa_{n}}$.
To learn the relation between $b$ and $b'$, we consider different cases:
\textbf{CASE 1:} 
$\forall x \in\sa_1 \exists g \in G: \allcond(\query_1) \rightarrow x = g $. 
(i) We first prove that $b = b'$ when all the operators under $op_{n+1}$ are monotone operators. Based on rules in \Cref{tab:subset}, only aggregation could generate inequality in $\aComp_{\query_i,\sa_i}$ ($1 \leq i \leq n$), then $\aComp_{\query_{n},\sa_{n}}$ only contains equalities, so $\query_n(\instOf{\psSet}) \subseteq \query_n(\db)$. 
Now we discuss $\query_{n+1}(\instOf{\psSet})$ and $\query_{n+1}(\db)$. 
To make $\query_{n+1}(\instOf{\psSet}) \subseteq \query_{n+1}(\db)$, there should not exist the case that $\exists t_1 \in \instOf{\psSet}$ and $\exists t_2 \in \db - \instOf{\psSet}$ where $t_1.x = t_2.x$. However, this case will not happen since $\forall x \in\sa_{n+1} \exists g \in G: \allcond(\query_{n+1}) \rightarrow x = g $. 
Thus, $\query_{n+1}(\instOf{\psSet}) \subseteq \query_{n+1}(\db)$ and $b = b'$.
(ii) We next prove that $b = b'$ when there is only one aggregation under $op_{n+1}$. We represents this aggregation as $\Aggregation{\aggf_1(a_1) \rightarrow \agga_1}{\grpatts_1}$ 
and use $\query_{G_1}$ represent the subquery rooted at this aggregation. By continuing apply \Cref{lem:gc}, we could get $\query_{G_1}(\instOf{\psSet}) \matchContains{\aComp_{\query_{G_1},\sa_{G_1}}} \query_{G_1}(\db)$ since only monotone operators under this aggregation.
For convenience we assume that all columns' name appeared in a query are unique. Because $\sa_{n+1} \subseteq \schemaOf{\db}$ and 
 $\forall x \in\sa_{n+1} \exists g \in G: \allcond(\query_n) \rightarrow x = g $, then  $\forall x \in\sa_{G_1} \exists g \in G_1: x = g $, otherwise, we will lose some columns used in $\sa_{n+1}$ before reaching $op_{n+1}$. 
Thus, after $\Aggregation{\aggf_1(a_1) \rightarrow \agga_1}{\grpatts_1}$, we get $\agga_1 = \agga_1'$. Then, $\aComp_{\query_{G_1},\sa_{G_1}}$ would only contain equalities and thus $\query_{G_1}(\instOf{\psSet}) \subseteq \query_{G_1}(\db)$. 
And only monotone operators exist between $_{F_1(a) \rightarrow f_1} \aggregation _{G_1}$ and $op_{n+1}$, which would not change the $\aComp$, thus $\aComp_{\query_{n},\sa_{n}}$ would only contain equalities. 
Since  $\query_n(\instOf{\psSet}) \matchContains{\aComp_{\query_n,\sa_n}} \query_n(\db)$,  
 $\query_n(\instOf{\psSet}) \subseteq \query_n(\db)$. 
 Then, similar to (i), we get $b = b'$. 
(iii) At last, we prove that $b = b'$ when there are many aggregations under $op_{n+1}$. Assume these aggregations are $\Aggregation{\aggf_1(a_1) \rightarrow \agga_1}{\grpatts_1}, \ldots, \Aggregation{\aggf_k(a_k) \rightarrow \agga_k}{\grpatts_k}$ from bottom to up respectively. The matched subqueries are $\query_{G_1}, \ldots, \query_{G_k}$. Since $\sa_{n+1} \subseteq \schemaOf{\db}$ and 
$\forall x \in\sa_{n+1} \exists g \in G: \allcond(\query_n) \rightarrow x = g $, then $\forall x \in\sa_{G_1} \exists g \in G_1: \rightarrow x = g, \ldots,  \forall x \in\sa_{G_k} \exists g \in G_k: \rightarrow x = g$, otherwise, we will lose some columns in $\sa$ before reaching $op_{n+1}$. Then by reapplying (ii), we would get $b = b'$. Note that for non-monotone operators, we only consider aggregation in this paper. 
Thus, $\sa_{n+1} = \sa_n \land b = b'$.
\textbf{CASE 2:} $\exists x: x \in \sa_1 \land x \not \in G \land (f=count \lor (f \in \{sum,max\} \land (\allcond(\query_1) \rightarrow a \geq 0)))$.
Recall we have proven that $\query_{n+1}(\instOf{\psSet}) \matchContains{\aComp_{\query_{n+1},\sa_{n+1}}} \query_{n+1}(\db)$ where $\aComp_{\query_{n+1},\sa_{n+1}} = \aComp_{\query_{n},\sa_{n}}$.
Assume exists a pair of matched tuples $(t,t')$ in  $\query_{n+1}(\instOf{\psSet}) \matchContains{\aComp_{\query_{n+1},\sa_{n+1}}} \query_{n+1}(\db)$ where $\{t\} = op_{n+1}(\query_{n}(\{t_1,\ldots,t_m\}))$ and $\{t'\} = op_{n+1}(\query_{n}(\{t_1',\ldots,t_n'\}))$. That is, $\{t_1,\ldots,t_m\} \subseteq \query_{n+1}(\instOf{\psSet})$ and $\{t_1',\ldots,t_n'\} \subseteq \query_{n+1}(\db)$.
If $\exists x: x \in \sa_n \land x \not \in G$, then $m \leq n$. Since $f=count \lor (f \in \{sum,max\} \land (\allcond(\query_n) \rightarrow a \geq 0))$, $b \leq b'$. Thus, $\sa_{n+1} = \sa_n \land b \leq b'$.
\textbf{CASE 3:} $\exists x: x \in \sa_1 \land x \not \in G \land (f \in \{sum,min\} \land (\allcond(\query_1) \rightarrow a \leq 0))$. Similar to CASE 2, if $\exists x: x \in \sa_n \land x \not \in G$, then $m \leq n$.
Since $f \in \{sum,min\} \land (\allcond(\query_n) \rightarrow a \leq 0)$, $b \geq b'$. Thus, $\sa_{n+1} = \sa_n \land b \geq b'$.
\textbf{CASE 4:} Otherwise, we are unable to decide the relationship between $b$ and $b'$. Thus, $op_{n+1}(\query_n(\instOf{\psSet})) \matchContains{\aComp_{\query_{n+1},\sa_{n+1}}} op_{n+1}(\query_n(\db))$ where  $\aComp_{\query_{n+1},\sa_{n+1}} = \aComp_{\query_n,x}$ and $\sa_{n+1} = \sa_n$.

\underline{$\ordlimit{O}{C}$:}
Assume $\exists t \in \query_n(\instOf{\psSet})$, since
$\query_n(\instOf{\psSet}) \matchContains{\aComp_{\query_n,\sa_n}} \query_n(\db)$,
then there have to be a tuple $t' \in \query_n(\db)$ such that
$t$ mapped with $t'$ and $(t,t') \models \aComp_{\query_n,\sa_n}$.
Assume $t'$ satisfies the limit condition, i.e., $t' \in  \ordlimit{O}{C}(\query_n(\db))$. 
Since $\forall o \in O: \aComp_{\query_1,\sa_1} \land \allcond(\query_1) \land \allcond(\query_1') \rightarrow o = o'$, then $t.C = t'.C$. Thus, $t \in  \ordlimit{O}{C}(\query_n(\instOf{\psSet}))$. 
Thus, $op_{n+1}(\query_n(\instOf{\psSet})) \matchContains{\aComp_{\query_{n+1},\sa_{n+1}}} op_{n+1}(\query_n(\db))$ where \\$\aComp_{\query_{n+1},\sa_{n+1}} = \aComp_{\query_{n},\sa_{n}}$.	

\underline{$\query_{Ln}(\db) \crossprod \query_{Rn}(\db)$:} 
Let $\query_{Ln}(\db)$ and $\query_{Rn}(\db)$ be the left and right children of $op_{n+1}$ respectively.
Since $\gc(\query_{n+1}, \sa_{n+1}) = \true$, based on \Cref{lem:gc}, $\gc(\query_{Ln}, \sa_{Ln}) = \true$ and $\gc(\query_{Rn}, \sa_{Rn}) = \true$.  Thus $\query_{Ln}(\instOf{\psSet}) \matchContains{\aComp_{\query_{Ln},\sa_{Ln}}} \query_{Ln}(\db)$ and $\query_{Rn}(\instOf{\psSet}) \matchContains{\aComp_{\query_{Ln},\sa_{Rn}}} \query_{Rn}(\db)$.  Then \\$\query_{Ln}(\db_{\psSet}) \crossprod  \query_{Rn}(\db_{\psSet})  \matchContains{\aComp_{\query_{Ln} \crossprod \query_{Rn},\sa_{\query_{Ln} \crossprod \query_{Rn}}}} \query_{Ln}(\db) \crossprod  \query_{Rn}(\db)$ where $\aComp_{\query_{Ln} \crossprod \query_{Rn},\sa_{\query_{Ln} \crossprod \query_{Rn}}} =  \aComp_{\query_{Ln}} \land \aComp_{\query_{Rn}}$ and $\sa_{\query_{Ln} \crossprod \query_{Rn}} = \sa_{\query_{Ln}} \land \sa_{\query_{Rn}}$. 

\underline{$\query_{Ln}(\db) \join_{a=b}  \query_{Rn}(\db)$:} Similar to $\crossprod$, for $\query_{Ln}$ and $\query_{Rn}$ we have $\query_{Ln}(\instOf{\psSet}) \matchContains{\aComp_{\query_{Ln},\sa_{Ln}}} \query_{Ln}(\db)$ and $\query_{Rn}(\instOf{\psSet}) \matchContains{\aComp_{\query_{Ln},\sa_{Rn}}} \query_{Rn}(\db)$. Since $ \aComp_{\query_{Ln},\sa_{Ln}} \land \allcond(\query_{Ln}') \land \allcond(\query_{Ln}) \Rightarrow a=a' $ and $ \aComp_{\query_{Rn},\sa_{Rn}} \land \allcond(\query_{Rn}') \land \allcond(\query_{Rn}) \Rightarrow b=b' $, then for each tuple $t \in \query_{Ln}(\db_{\psSet})  \join_{a=b}   \query_{Rn}(\db_{\psSet})$ where $t.a = t.b = C$ ($C$ is a constant value), there have to be a tuple $t' \in \query_{Ln}(\db)  \join_{a=b}  \query_{Rn}(\db)$ where $t'.a = t'.b= C$ such that 
$t$ mapped with $t'$ and $(t,t') \models \aComp_{\query_{Ln},\sa_{Ln}} \land \aComp_{\query_{Rn},\sa_{Rn}}$. Therefore, \\
$\query_{Ln}(\db_{\psSet})  \join_{a=b}  \query_{Rn}(\db_{\psSet})  \matchContains{\aComp_{\query_{Ln}  \join_{a=b} \query_{Rn},\sa_{\query_{Ln}  \join_{a=b} \query_{Rn}}}} \query_{Ln}(\db)$ $ \join_{a=b}  \query_{Rn}(\db)$ where $\aComp_{\query_{Ln}  \join_{a=b} \query_{Rn},\sa_{\query_{Ln} \crossprod \query_{Rn}}} =  \aComp_{\query_{Ln}} \land \aComp_{\query_{Rn}}$ and $\sa_{\query_{Ln}  \join_{a=b} \query_{Rn}} = \sa_{\query_{Ln}} \land \sa_{\query_{Rn}}$.
\end{proof}
}
\BG{state how this will be used afterwards.}
In \Cref{lem:gc-imply-gcontainment} we proved that $\gc(\query,\sa) \Rightarrow \query(\instOf{\psSet}) \matchContains{\aComp_{\query,\sa}} \query(\db)$.
Since $\instOf{\psSet} \supseteq P(\query,\db)$, now we reason about that non-provenance tuples in $\instOf{\psSet}$ would not 
result in any result tuple which is different with the result in $ \query(\db)$. That is, if we generate a tuple in $ \query(\db)$, we would also generate this tuple in $\query(\instOf{\psSet})$.
\begin{Lemma}\label{lem:gc-regenerate}
Let $\query$ be a query, $\db$ be a database, 
and $\sa = \bigcup_{1}^{n} \sa_i$ a set of attributes where each $\sa_i$ belongs to a relation $R_i$ accessed by $\query$ such that $R_i \neq R_j$ for $i \neq j$. 
Given a set of provenance sketches $\psSet = \{ \provSketch_i \}$  for  $\query$ over $\db$ with respect to a set of range partitions $\{ \parti_{\ranges_i,\sa_i}(R_i) \}$, then
\ifnottechreport{$\gc(\query,\sa) \Rightarrow \query(\db) \subseteq \query(\instOf{\psSet}). $}
\iftechreport{\[ \gc(\query,\sa) \Rightarrow \query(\db) \subseteq \query(\instOf{\psSet}) \]}
\end{Lemma}
\ifnottechreport{
\begin{proof}
For the full proof see \cite{techreport}.
\end{proof}
}
\iftechreport{
\begin{proof}
$\query(\db) \subseteq \query(\instOf{\psSet})$ represents that if exists a tuple $t \in \query(\db)$, then $t \in \query(\instOf{\psSet})$. That is, we need to prove that $\gc(\query,\sa) \land \exists t \in \query(\db) \Rightarrow t \in \query(\instOf{\psSet})$.  In the following, we prove this by induction through proving this formula holds for each subquery of $\query$. However, not all the subquery's result satisfy this formula,
since there is a implicit condition $P(t)\subseteq P(\query,\db)$ under $\gc(\query,\sa) \land \exists t \in \query(\db) \Rightarrow t \in \query(\instOf{\psSet})$. 
Recall that we use $P(\query,\db)$ to represent the provenance of $\query$ over $\db$, for convenience, we use $P(t)$ to represent the provenance to derive tuple $t$.
Then the meaning is that only the intermediate result tuples which contribute to the final result $\query(\db)$ would satisfy this formula. Thus, what we need to prove is that $ \gc(\query,\sa) \land \exists t \in \query(\db) : P(t)\subseteq P(\query,\db)  \Rightarrow t \in \query(\instOf{\psSet})$.

\underline{Base case:}
When $Q$ is table access operator $R$, if $t \in R$, $P(t) = \{t\}$. Recall that provenance sketches are a superset of provenance, that is $P(t) \subseteq R_{\psSet}$. Thus $t \in R_{\psSet}$.


\underline{Inductive step:}
Assume query $\query_n$, a sub query of $\query$, with depth less than or equal to $n$.
Here we use 
\emph{depth} to represent the levels in the relational algebra format of a query, e.g., $\query_1$ represents the base table $R$ with depth 1.
Assume we have proven that $ \gc(\query_n,\sa_n) \land \exists t \in \query_n(\db) : P(t)\subseteq P(\query,\db)  \Rightarrow t \in \query_n(\instOf{\psSet})$.
Based on the induction hypothesis, we have to prove that the same holds for sub query $\query_{n+1}$ with depth less than or equal to $n+1$, that is, we need to prove  $ \gc(\query_{n+1},\sa_{n+1}) \land \exists t \in \query_{n+1}(\db) : P(t)\subseteq P(\query,\db) \Rightarrow t \in \query_{n+1}(\instOf{\psSet})$. For notational convenience, we use $op_{n}$($op_{n+1}$) represent any operator of depth $n$ ($n + 1$), i.e., $\query_{n+1} = op_{n+1}(\query_{n})$. 

Since $\gc(\query_{n+1},\sa_{n+1}) = \true$, base on \Cref{lem:gc}, then $\gc(\query_n, \sa_n) = \true$. For join operator, we assume the left and right input of $op_{n+1}$ are $\query_{nL}(\db)$ and $\query_{nR}(\db)$ respectively, then $\gc(\query_{nL}, \sa_{nL}) = \true$ and $\gc(\query_{nR}, \sa_{nR}) = \true$.


Now we discuss different $op_{n+1}$:

\underline{For $\projection_A$:} 
Assume $\exists t_{proj} \in \query_{n+1}(\db) : P(t_{proj}) \subseteq P(\query,\db)$, then there have to be a tuple $t \in \query_n(\db)$ such that $\projection_A(\{t\}) = \{t_{proj}\}$. Since $P(t_{proj}) = P(t)$, $P(t) \subseteq P(\query,\db)$. And because
$\gc(\query_n, \sa_n) = \true$, based on assumption, $t \in \query_n(\instOf{\psSet})$.
Thus, from $\query_{n+1}(\instOf{\psSet}) = \projection_A(\query_{n}(\instOf{\psSet}))$,
$t \in \query_n(\instOf{\psSet})$ and $\projection_A(\{t\}) = \{t_{proj}\}$, we get $t_{proj} \in \query_{n+1}(\instOf{\psSet})$.


\underline{For $\selection_{\theta}$:} Assume $\exists t \in \query_{n+1}(\db) : P(t) \subseteq P(\query,\db)$, since $\selection_{\theta}(\{t\}) = \{t\}$, $t \in \query_{n}(\db)$.
Because $\gc(\query_n, \sa_n) = \true$, based on assumption, $t \in \query_n(\instOf{\psSet})$. Thus, $t \in \query_{n+1}(\instOf{\psSet})$.

\underline{For $\delta$:} Assume $\exists t \in \query_{n+1}(\db) : P(t) \subseteq P(\query,\db)$, there have to exist at least one tuple $t' \in \query_{n}(\db)$ such that $t = t'$, that is $P(t') \subseteq P(t)$, thus $P(t') \subseteq P(\query,\db)$.
And because $\gc(\query_n, \sa_n) = \true$, based on assumption, $t' \in \query_n(\instOf{\psSet})$, that it  $t \in \query_n(\instOf{\psSet})$. Thus, $t \in \query_{n+1}(\instOf{\psSet})$.

\underline{$\ordlimit{O}{C}$:}
Assume $\exists t \in \query_{n+1}(\db) : P(t) \subseteq P(\query,\db)$, then $t \in \query_{n}(\db)$.   And because $\gc(\query_n, \sa_n) = \true$, based on assumption, $t \in \query_n(\instOf{\psSet})$. Let $T$ represent all these $t$, that is $T \subseteq \query_{n+1}(\db)$. Now we decide whether $T \subseteq \query_{n+1}(\instOf{\psSet})$.
Recall \Cref{lem:gc-imply-gcontainment} proved that $\gc(\query_n, \sa_n) \Rightarrow \query_n(\instOf{\psSet}) \matchContains{\aComp_{\query_n,\sa_n}} \query_n(\db)$. Also we have   $T \subseteq \query_{n}(\instOf{\psSet})$, $T \subseteq \query_{n}(\db)$ and $\forall o \in O : o = o'$ from the rules in \Cref{tab:subset}. Thus, if $T \subseteq \query_{n+1}(\db)$, then $T \subseteq \query_{n+1}(\instOf{\psSet})$.  

\underline{$\query_{Ln}(\db) \union \query_{Rn}(\db)$:} Assume $\exists t \in \query_{n+1}(\db) : P(t) \subseteq P(\query,\db)$, then $t \in \query_{Ln}(\db)$ or $t \in \query_{Rn}(\db)$.
And because $\gc(\query_{nL}, \sa_{nL}) = \true$ and $\gc(\query_{nR}, \sa_{nR}) = \true$, based on assumption, $t \in \query_{Ln}(\instOf{\psSet})$ or $t \in \query_{Rn}(\instOf{\psSet})$. Thus, $t \in \query_{n+1}(\instOf{\psSet})$.

\underline{For $\crossprod, \join$:} Let $\query_{nL}(\db)$ and $\query_{nR}(\db)$ represent the left and right input of $op_{n+1}$ respectively. Assume $\exists t \in \query_{n+1}(\db) : P(t) \subseteq P(\query,\db)$,  then there have to be a tuple $t'_L \in \query_{nL}(\db)$ and a tuple $t'_R \in \query_{nR}(\db)$ such that $\{t'_L\} op_{n+1} \{t'_R\} = \{t\}$.
And $P(t'_L) \union P(t'_R) = P(t)$, thus   $P(t'_L) \subseteq P(\query,\db)$ and $P(t'_R) \subseteq P(\query,\db)$.
And because $\gc(\query_{nL}, \sa_{nL}) = \true$ and $\gc(\query_{nR}, \sa_{nR}) = \true$, based on assumption, $t'_L \in \query_{nL}(\instOf{\psSet})$ and $t'_R \in \query_{nR}(\instOf{\psSet})$. Thus  $t \in \query_{n+1}(\instOf{\psSet})$.

\underline{For $\Aggregation{\aggf(a) \to \agga}{\grpatts}$:}
Assume $\exists t \in \query_{n+1}(\db) : P(t) \subseteq P(\query,\db)$ and $\{t_1,...,t_m\} \subseteq \query_{n}(\db)$ such that $op_{n+1}(\{t_1,...,t_m\}) = \{t\}$.
Then $P(t_1) \union \dots \union P(t_m) = P(t)$, thus $P(t_1) \subseteq P(\query,\db), \dots, P(t_m) \subseteq P(\query,\db) $.
And because $\gc(\query_n, \sa_n) = \true$, based on assumption, $\{t_1,...,t_m\} \subseteq \query_n(\instOf{\psSet})$.
To let $t \in \query_{n+1}(\instOf{\psSet})$, we have to keep that not exists a set of tuples $T \subseteq \query_n(\instOf{\psSet})$ such that $op_{n+1}(\{t_1,...,t_m\} \union T) = \{t'\} $ and $t' != t$ where $t' \in \query_{n+1}(\db)$. Then, the question transforms to prove no such $T$ in $\query_n(\instOf{\psSet}) - \{t_1,...,t_m\}$.
Since $\gc(\query_n, \sa_n) = \true$, based on \Cref{lem:gc-imply-gcontainment}, then $\query_n(\instOf{\psSet}) \matchContains{\aComp_{\query_n,\sa}} \query_n(\db)$.
Assume $op_{n+1}$ groups on columns $\{a_1, \dots, a_k\}$ and for $\{t_1,...,t_m\}$, the values on these columns are $\{v_1, \dots, v_k\}$. Then, that is in $\query_n(\db) - \{t_1,...,t_m\}$, there are no tuples on these columns have same values with $\{v_1, \dots, v_k\}$.
Because  $\gc(\query_{n+1}, \sa_{n+1}) = \true$, the aggregation rule in \Cref{tab:gc} are satisfied, then we know $\forall g \in G : g = g'$, that is $a_1 = a_1', \dots, a_k = a_k'$.
Since $\query_n(\instOf{\psSet}) \matchContains{\aComp_{\query_n,\sa}} \query_n(\db)$, if no tuples in $\query_n(\db) - \{t_1,...,t_m\}$ have same values with $\{v_1, \dots, v_k\}$ on $\{a_1, \dots, a_k\}$, then there no tuples in $\query_n(\instOf{\psSet}) - \{t_1,...,t_m\}$
have the same values with $\{v_1, \dots, v_k\}$ on $\{a_1, \dots, a_k\}$.
Thus, $T$ does not exist, we have $t \in \query_{n+1}(\instOf{\psSet})$.
\end{proof}
}

\BG{State
the intuitive reason why generalized containment implies equality for query
results. We need to explain what else is checked in the conditions that ensures
equality because generalized containment does not imply equality in general.}
In \Cref{lem:gc-imply-gcontainment} we proved that $\query(\instOf{\psSet}) \matchContains{\aComp_{\query,\sa}} \query(\db)$ and \Cref{lem:gc-regenerate} keeps the matched tuples are equal, that is, $\query(\instOf{\psSet}) \subseteq \query(\db)$. Then, we have $\query(\instOf{\psSet}) = \query(\db)$.
\begin{Theorem}[$\gc(\query,\sa)$ implies safety of $\sa$]\label{theo:safety-check-is-correct}
  Let $\query$ be a query, $\db$ be a database, and $\sa = \bigcup_{1}^{n} \sa_i$ a set of attributes where each $\sa_i$ belongs to a relation $R_i$ accessed by $\query$ such that $R_i \neq R_j$ for $i \neq j$. 
If $\gc(\query,\sa)$ holds, then $\sa$ is a safe set of attributes for $\query$.
\end{Theorem}
\ifnottechreport{
\begin{proof}
For the full proof see \cite{techreport}.
\end{proof}
}
\iftechreport{
\begin{proof}
Recall that a set of attribute $\sa$ is safe, if for databases $\db$ and all sets of provenance sketches $\psSet = \{ \provSketch_i \}$  for  $\query$ over $\db$  with respect to a set of range partitions $\{ \parti_{\ranges_i,\sa_i}(R_i) \}$, we have 
\iftechreport{\[\gc(\query,\sa) \Rightarrow \query(\instOf{\psSet}) = \query(\db)\]}

  Since $\gc(\query,\sa) = \true$, based on \Cref{lem:gc-imply-gcontainment}, $\query(\instOf{\psSet}) \matchContains{\aComp_{\query,\sa}} \query(\db)$; based on \Cref{lem:gc-regenerate}, $\query(\instOf{\psSet}) \supseteq \query(\db)$. Thus, for each pair of matched tuples $(t,t')$ in $\query(\instOf{\psSet}) \matchContains{\aComp_{\query,\sa}} \query(\db)$ where $t \in \query(\instOf{\psSet})$ and $t' \in \query(\db)$, $t = t'$, that is, $\query(\instOf{\psSet}) \subseteq \query(\db)$. Thus, $\query(\instOf{\psSet}) = \query(\db)$.
\end{proof}
}
Furthermore, since our rules are independent of the number of fragments, adding fragments to a safe sketch, the resulting sketch is guaranteed to be safe too.
\BG{Some of the commmented out stuff below can be moved to TR}




%% file: sections/template_query.tex
\section{Reusing Provenance Sketches for Parameterized Queries}\label{sec:reuse-different}

Given a set of accurate provenance sketches $\psSet$  captured for a query $\query$, we would like to be able to use it answer future queries $\query'$.
However, to determine whether this is possible, we need to determine whether $\instOf{\psSet}$ is sufficient for $\query'$.
This is similar to checking query containment which is known to be undecidable for the class of queries we are interested in~\cite{chandra1977,sagiv1980,klug1988}.
\BG{Undecidability through reduction from Hilbert?}
Here we focus on developing a solution for a narrower problem: reusing sketches across multiple instances of a parameterized query~\cite{amiri2003}. Given the prevalence of parameterized queries in applications and reporting tools that access a database, this is an important special case. Note that even for ad hoc analytics, it is common that query patterns repeat if the number of queries is sufficiently large. While such queries are typically not expressed as parameterized queries we can treat them as such by replacing all constants in selection conditions with parameters. \iftechreport{Note that the problem studied in this section can also be interpreted as a generalization of safety checking for sketches with the difference that we determine the safety of a sketch for a different query instead of for the query it was captured for.} The major result of this section is a sufficient condition for checking whether a sketch can be reused that is rooted in the safety conditions we introduced in~\Cref{sec:safety-check}.

Let $\pdom$ be a countable set of variables called parameters.
A \emph{parameterized query} $\apq$ for $\pvec = (\p_1, \ldots, \p_n)$ and $\p_i \in \pdom$ is a relational algebra expression where conditions of selections  may refer to a parameters from the set $\{\p_i\}$. We assume that each parameter from $\pvec$ is referenced at least once by $\apq$. A parameter binding $\vvec$ for $\apq$ is a vector of constants, one for each parameter $\p_i$ from $\pvec$. The \emph{instance} $\pinst{\pq}{\vvec}$ of $\apq$ for $\vvec$ is the query resulting from  substituting $\p_i$ with $\pv_i$ in $\pq$. For instance, the parameterized SQL query \lstset{mathescape=false}\lstinline!SELECT * FROM R WHERE a < $1! can be written as $\pq[\p_1] = \selection_{a  < \p_1}(R)$.
We define the \textbf{sketch reusability problem} as: given a parameterized query $\pq$, two instances $\query$ and $\query'$ for $\pq$, and a safe set of provenance sketches $\psSet$ for $\query$, determine whether $\instOf{\psSet}$ is sufficient for $\query'$. In the remainder of this section we develop a \emph{sufficient} condition for sketch reusability.
\BG{In reality, the same query has lower chance to be run again and again, how to use existing provenance sketches for different incoming queries becomes more important.
Assume given a query $Q$, the database $\db$, the provenance sketches of $Q$ over database $\db$ is $\provSketch$. Then when a new query $Q'$ coming in, we want to determine whether $\instOf{\provSketch}$ could be used to compute $Q'$, that is whether $Q'(\instOf{\provSketch}) = Q'(\db)$. To solve this problem we might need to find the relationship between these two queries $Q$ and $Q'$ which might be related to the query containment and equivalence problems\cite{chandra1977,sagiv1980,klug1988} which is hard. Thus, to simplify this problem, we start from the \textbf{template-based queries} that the selection conditions in the \lstinline!WHERE! clause share the same structure among the queries and differ only in a few numeric or string constants ~\cite{amiri2003}. For example, queries $Q=$ \lstinline!SELECT * FROM R WHERE a < 10! and $Q'=$ \lstinline!SELECT * FROM R WHERE a < 5! share the same structure but differ on selection conditions which are $a<10$ and $a<5$ respectively. Thus to determine whether  $Q'(\instOf{\provSketch}) = Q'(\db)$, we could use the conditions in the \lstinline!WHERE! to make the decision.
}
Before 
presenting our sufficient condition, 
we first state three lemmas that we will use to develop this condition.
First, 
observe that the same sets of attributes are safe for all instances of a parameterized query.

\input{sections/ge_cond}

\begin{Lemma}\label{lem:ps-safe-same-template}
Let $\query$ and $\query'$ be two instances of a parameterized query $\pq$, then any set of attributes $\sa$ that is safe for $\query$ is safe for $\query'$.
\end{Lemma}
\iftechreport{
\begin{proof}
The definition of safety from \Cref{sec:safety-check} does not consider constants in selection conditions. Since $\query$ and $\query'$ only differ in such constants, any sketch type  that is safe for $\query$ is also safe for $\query'$.
\end{proof}
}

Furthermore, adding additional fragments to a safe sketch $\provSketch$ for a query $\query$ yields a safe sketch (\Cref{sec:safety-check}). 

\begin{Lemma}\label{lem:additional-fragments}
Let $\provSketch$ be a safe sketch for a query $\query$. Then any sketch $\provSketch' \supseteq \provSketch$ is safe for $\query$.
\end{Lemma}
\iftechreport{
\begin{proof}
Since our safety rules introduced in \cref{sec:safety-check} only care about which columns we do range partition on, that is if we say an accurate provenance sketch safe, any superset of it is still safe.
\end{proof}
}
Recall that accurate provenance sketches are sketches which do only contain ranges whose fragments contain provenance. Consider
a database $\db$ and two queries $\query$ and $\query'$ and denote
the provenance of $\query$ ($\query'$) as $\prov{\query}{\db}$ ($\prov{\query'}{\db}$).
Furthermore, consider
two sets of accurate provenance sketches $\psSet$ and $\psSet'$ build over the same attributes $\sa$ and partitions where $\psSet$ ($\psSet'$) is a sketch for $\query$ ($\query'$).
If $\prov{\query}{\db} \supseteq \prov{\query'}{\db}$ then $\psSet \supseteq \psSet'$ and, thus, also $\instOf{\psSet} \supseteq \instOf{\psSet'}$.

\begin{Lemma}\label{lem:prov-to-ps}
  Consider two queries $\query$ and $\query'$ and database $\db$ with $n$ relation and let $\psSet = \{\provSketch_1, \ldots, \provSketch_m\}$ and $\psSet' = \{\provSketch_1', \ldots, \provSketch_m'\}$
be  two set of accurate provenance sketches for $\query$ ($\query'$)
  that are both based on 
the same set of range partitions $\{\parti_{\ranges_1,a_1}(R_1), \ldots, $
$\parti_{\ranges_m,a_m}(R_m) \}$ where $\{R_1, \ldots, R_m \} \subseteq \schemaOf{\db}$. We have,
\ifnottechreport{
  $\prov{\query}{\db} \supseteq \prov{\query'}{\db} \Rightarrow \psSet \supseteq \psSet' \land \instOf{\psSet} \supseteq \instOf{\psSet'}$.
  }
\iftechreport{
\[
\prov{\query}{\db} \supseteq \prov{\query'}{\db} \Rightarrow \psSet \supseteq \psSet' \land \instOf{\psSet} \supseteq \instOf{\psSet'}
\]
}
\end{Lemma}
\iftechreport{
\begin{proof}
Since $\prov{\query}{\db} \supseteq \prov{\query'}{\db}$, for each tuple $t \in \prov{\query'}{\db}$, $t \in \prov{\query}{\db}$. Assume $t \in \rel_{\range}$ where $\rel \in \db$
, then for each $\range \in \provSketch'$, $\range \in \provSketch$, that is $\provSketch \supseteq \provSketch'$. Since it is hold for each $\provSketch$, we say $\psSet \supseteq \psSet'$ and thus $\instOf{\psSet} \supseteq \instOf{\psSet'}$.
\end{proof}
}

These lemmas imply that a provenance sketch for any instance $\query$ of a parameterized query $\pq$ is safe for another instance $\query'$ of $\pq$, if $\prov{\query}{\db} \supseteq \prov{\query'}{\db}$. We refer to this as \textit{provenance containment}. In the following we develop a sufficient condition that guarantees  provenance containment
%
for all input databases $\db$. That is, we reason about the results of $\query$
and $\query'$ over all possible input databases.  
We again use an SMT
solver similar to how we checked safety in~\Cref{sec:safety-check}.  Our
condition consists of two parts $\ucond(\query', \query)$ and $\gpe(\query',
\query)$. Condition $\gpe$ fulfills a similar purpose as $\gc$ in our safety
condition. As $\gc$ it is define recursively over the structure of a query and we construct a conjunction  $\aComp_{\query',\query}$ of comparisons between attributes from $\query$ and $\query'$ such that $\gpe$ (together with the condition $\ucond$ explained below) implies general containment ($\query'(\db) \matchContains{\aComp_{\query',\query}} \query(\db)$). We will use $a$ to refer to attributes from $\query$ and $a'$ to refer to the corresponding attribute from $\query'$. Similarly, if $\theta$ is a condition in $\query$, then $\theta'$ denotes the corresponding condition in $\query'$.

The main difference to $\gc$ is that we are now dealing with two different
queries instead of one query. The selection conditions of the two queries that
restrict values of an attribute may be spread over multiple operators in the
queries. It is possible that the conditions of all selections of $\query'$ imply
the conditions of all selections of $\query$ even though this does not hold for
all individual selections of these two queries. As a trivial example consider
$\query =\selection_{a = 20}(\selection_{a > 30}(R))$ and $\query' =
\selection_{a = 20}(\selection_{a > 10}(R))$. Subquery $\selection_{a > 10}(R)$
is not contained in $\selection_{a > 30}(R)$, but $\query$ and $\query'$ are
equivalent. To avoid failing to determine generalized containment, because it
does not hold for a subquery, we do not test generalized containment for
selections in $\gpe$. Instead we use condition $\ucond(\query',\query)$ to test whether all conditions in $\pred(\query')$ imply $\pred(\query)$:
\[
  \ucond(\query', \query) = \aComp_{\query',\query} \land \pred(\query') \land \expr(\query') \land \expr(\query) \rightarrow \pred(\query)
\]
For the example shown above this means we test $a = a' \land a' = 20 \land a' > 10 \rightarrow a = 20 \land a > 30$ instead of testing $a = a' \land a' > 10 \rightarrow a > 30$ first (which would fail). A similar problem arises when testing whether the input groups for an aggregation are the same for both queries. To avoid failing, because we may not have seen all restriction for the values of group-by attributes, we only check the restrictions on non-group-by attributes enforced by the two queries.
Here $\ngpred(\query)$ denotes the result of removing from $\pred(\query)$ all conjuncts that only reference group-by attributes, e.g., given $\pred(\query) = a>10 \land g<5$ where $g$ is  group-by attribute, we get $\ngpred(\query) = a > 10$.
We construct two conditions $\textcircled{1}$  and $\textcircled{2}$ to test whether it is the case that for any group that exists in both $\querya'(\db)$  and $\querya(\db)$, the group for $\querya'$ contains a subset of the tuples of the corresponding group for $\querya$ (or vice versa). If both  $\textcircled{1}$ and  $\textcircled{2}$ hold, then $\querya'$ and $\querya$ produce the same result for every group that exists in both query results. Thus, the aggregation function result produced for these groups by the two queries are equal (we can add $b = b'$ to $\aComp_{\query',\query}$). For the 2nd and 3rd case, we check whether the tuples in the group for  $\querya'$ is a subset of the tuples for same group in $\querya$.If this is the case and we using $min$ or $sum$ over negative numbers than then the aggregation function result for $\querya'$ is smaller than the one for $\querya$. The 3rd case is the symmetric case for $sum$ over positive numbers or $max$ aggregation.   

\ifnottechreport{
\begin{Example}\label{ex:reuse}
Consider the parameterized query $\pq = \selection_{cnt>\$2}($
$ \Aggregation{state}{count(\ast) \rightarrow cnt}(\selection_{popden>\$1}(cities)))$. This query returns states that have more than \texttt{\$2} cities with at least \texttt{\$1} inhabitants. 
Assume $\query$ and $\query'$ are two instances of $\pq$ with parameters binding $(100,10)$ and $(100,15)$ for $(\$1,\$2)$ respectively. We use $\query_{agg}$ and $\query_{agg}'$ to denote the subqueries rooted at the aggregation operator. To determine whether a set of sketches $\psSet$ for $\query$ can be used to answer $\query'$, we get the conditions shown below. We use $p$, $c$, and $s$ to denote $popden$, $city$, and  and $state$, respectively.
\begin{align*}
  \pred(\query) &= p > 100 \land cnt > 10
  &\pred(\query') &= p'> 100 \land cnt' > 15
\end{align*}\\[-10mm]
\begin{align*}
  \aComp_{\query',\query} &= p = p' \land c = c' \land s = s' \land cnt = cnt'
\end{align*}
Since this query does not contain any projections, $\expr(\query)$ and $\expr(\query')$ are empty. The condition $\gpe(\query', \query)$ constructed for this query tests the relationship between group-by attributes in the inputs of the aggregation subqueries $\query_{agg}$ and $\query_{agg'}$. Since $\aComp_{\query_{agg}',\query_{agg}}$ contains $s = s'$, $\gpe(\query', \query)$ holds. Furthermore, both \textcircled{1} and \textcircled{2} hold and, thus, we add $cnt = cnt'$ to $\aComp_{\query', \query}$. Finally, $\ucond(\query',\query)$ tests that\\[-4mm]
\[
\aComp_{\query',\query} \land \pred(\query') \land \expr(\query') \land \expr(\query) \rightarrow \pred(\query)
\]
Substituting the conditions shown above we get  $p = p' \land cnt = cnt' \land p > 100 \land cnt' > 15 \land p'> 100 \land cnt > 10$. Since this condition holds for all possible values of the variables in the formula (recall that free variables are assumed to be universally quantified), we can use $\psSet$ to answer $\query'$.
\end{Example}
}

\iftechreport{
\begin{Lemma}\label{lem:ge}
  Let $\db$ be a database, $\query$ and $\query'$ be two instances from the same parameterized query $\pq$. And $\query = op(\querya, \ldots, \query_n)$ and $\query' = op(\querya', \ldots, \query_n')$.
Then, for all $i \in \{1,\ldots, n\}$, 
 \begin{align*}
\gpe(\query', \query)  \Rightarrow \gpe(\query_i', \query_i)    \quad
\aComp_{\query',\query}                  \Rightarrow \aComp_{\query_i',\query_i}
  \end{align*}
\end{Lemma}
\begin{proof}
For the rules in \cref{tab:gen-eq}, 
$\gpe(\query', \query)$ is based on $\gpe(\query_1', \query_1) \land \ldots \land \gpe(\query_n', \query_n)$. Thus, this lemma holds.
\end{proof}
}
\iftechreport{
\begin{Lemma}\label{lem:gpe}
  Given  two instances $\query$ and $\query'$ from the same parameterized query $\pq$, a database $\db$. Then, 
 \begin{gather*}
\gpe(\query', \query) \land \ucond(\query',\query)
    \Rightarrow  \query'(\db) \matchContains{\aComp_{\query',\query}} \query(\db)
  \end{gather*}
Let $P(t)$ ($P(t')$) denote the provenance of tuple $t$ ($t'$).  For any $\mathcal{M} \subseteq \query'(\db) \times \query(\db)$ such that $\query'(\db) \matchContains{\aComp_{\query',\query}} \query(\db)$ holds based on $\mathcal{M}$ we have:
 \begin{gather*}
\gpe(\query', \query) \land \ucond(\query',\query)
    \Rightarrow 
    \forall (t',t) \in \mathcal{M}: P(t') \subseteq P(t)
  \end{gather*}
\end{Lemma}

}
\iftechreport{
\begin{proof}
We prove this by induction, that is we need to prove this holds for each subquery of $\pq$ between $\query$ and $\query'$.
However, 
$\ucond(\query',\query)$ and $\query'(\db) \matchContains{\aComp_{\query,\query}} \query(\db)$ might not hold for every subquery, there is a hidden constraint that $\forall (t',t) \in \mathcal{M}$ where $\mathcal{M} \subseteq \query'(\db) \times \query(\db)$, $P(t) \subseteq P(\query,\db)$ and $P(t') \subseteq P(\query',\db)$. Thus, what we need to prove in the induction is that for each subquery $\pq_{sub}$, let $S = \{t | t \in \query_{sub}(\db) \land P(t) \subseteq P(\query,\db)\}$, $S' = \{t' | t' \in \query_{sub}'(\db) \land P(t') \subseteq P(\query',\db)\}$ and $\mathcal{M} \subseteq S' \times S$ such that $S' \matchContains{\aComp_{\query_{sub}',\query_{sub}}} S$ holds based on $\mathcal{M}$, then
\begin{gather*}
\gpe(\query_{sub}', \query_{sub}) \land \ucond(\query',\query) \\ 
\Rightarrow \\
S' \matchContains{\aComp_{\query_{sub}',\query_{sub}}} S \land \forall (t',t) \in \mathcal{M}: P(t') \subseteq P(t).
\end{gather*}
Similar to $P(t)$ and $P(t')$, $P(S)$ ($P(S')$) represents the provenance of set $S$ ($S'$). And thus, we have $P(\query,\db) = P(S)$ and $P(\query',\db) = P(S')$.

\underline{Base case:}
When $\pq$ is table access operator $R$, then $\query(\db) = \query'(\db) = R$. Thus, $\query'(\db) \subseteq \query(\db)$ which is the special case of $\query'(\db) \matchContains{\aComp_{\query',\query}} \query(\db)$ where $\aComp_{\query',\query}$ contains equalities on all columns.
And $\forall (t',t) \in \mathcal{M}$ where $\mathcal{M} \subseteq \query'(\db) \times \query(\db)$, we have $P(t) = t$, $P(t') = t'$ and $t = t'$. Hence $P(t') = P(t)$, also $P(t') \subseteq P(t)$.



\underline{Inductive step:}
Assume query $\query_n$, a sub query of $\query$, with depth less than or equal to $n$.
Here we use 
\emph{depth} to represent the levels in the relational algebra format of a query, e.g., $\query_1$ represents the base table $R$ with depth 1.
Assume we have proven that 
$\gpe(\query_n', \query_n)  \land \ucond(\query',\query) \Rightarrow S_n' \matchContains{\aComp_{\query_n',\query_n}} S_n \land \forall (t_n',t_n): P(t_n') \subseteq P(t_n)$, where $(t_n',t_n) \in \mathcal{M}_n$ and $\mathcal{M}_n \subseteq \query_n'(\db) \times \query_n(\db)$.
Based on the induction hypothesis, we have to prove that the same holds for sub query $\query_{n+1}$ with depth less than or equal to $n+1$, that is, we need to prove  
$\gpe(\query_{n+1}', \query_{n+1})  \land \ucond(\query',\query)  \Rightarrow S_{n+1}' \matchContains{\aComp_{\query_{n+1}',\query_{n+1}}} S_{n+1} \land \forall (t'_{n+1},t_{n+1}):  P(t_{n+1}') \subseteq P(t_{n+1})$, where $(t_{n+1}',t_{n+1}) \in \mathcal{M}_{n+1}$ and $\mathcal{M}_{n+1} \subseteq \query_{n+1}'(\db) \times \query_{n+1}(\db)$.
For notational convenience, we use $op_{n}$($op_{n+1}$) represent any operator of depth $n$ ($n + 1$), i.e., $\query_{n+1} = op_{n+1}(\query_{n})$.
Based on \cref{lem:ge}, $\gpe(\query_{n+1}', \query_{n+1}) \Rightarrow \gpe(\query_{n}', \query_{n})$. Thus, based on assumption, $\gpe(\query_n', \query_n) \land \ucond(\query',\query) \Rightarrow S_n' \matchContains{\aComp_{\query_n',\query_n}} S_n \land \forall (t_n',t_n) \in \mathcal{M}_n:P(t_n') \subseteq P(t_n)$.
For join operator, we assume the left and right input of $op_{n+1}$ are $\query_{nL}(\db)$ and $\query_{nR}(\db)$ respectively, then above is hold for both $\query_{nL}$ and $\query_{nR}$.  Thus, $S_{nL}$ ($S_{nL}'$) is the $S$ ($S'$) in $\query_{nL}$ and $S_{nL}$ ($S_{nL}'$) is the $S$ ($S'$) in $\query_{nL}$. Similarly for $\mathcal{M}$ that $\mathcal{M}_{nL} \subseteq \query_{nL}'(\db) \times \query_{nL}(\db)$ and $\mathcal{M}_{nR} \subseteq \query_{nR}'(\db) \times \query_{nR}(\db)$.
In the following, we consider different $op_{n+1}$:

\underline{For $\projection_A$:}
Since $\forall (t_n',t_n) \in \mathcal{M}_n$, there have to be one tuple $t_{n+1} \in \query_{n+1}(\db)$ which satisfies $\projection_A\{t_n\} = \{t_{n+1}\}$ and one tuple $t_{n+1}' \in \query_{n+1}'(\db)$ which satisfies $\projection_A\{t_n'\} = \{t_{n+1}'\}$, and $P(t_n) = P(t_{n+1})$ and $P(t_n') = P(t_{n+1}')$.  Because $S_n' \matchContains{\aComp_{\query_n',\query_n}} S_n$, then
$S_{n+1}' \matchContains{\aComp_{\query_{n+1}',\query_{n+1}}} S_{n+1}$ where  $\aComp_{\query_{n+1}',\query_{n+1}} = \aComp_{\query_{n}',\query_{n}}$. In addition, since $P(t_n') \subseteq P(t_n)$, $P(t_{n+1}') \subseteq P(t_{n+1})$.

\underline{For $\selection_{\theta}$:}
Since $P(\query,\db) = P(S_n)$ and  $P(\query',\db) = P(S_n')$, then every tuple $t_n \in S_n$ and every tuple $t_n' \in S_n'$ have to be survived after selection, that is $t_n \in \query_{n+1}(\db)$ and $t_n' \in \query_{n+1}'(\db)$. Because $S_n' \matchContains{\aComp_{\query_n',\query_n}} S_n$, then
$S_{n+1}' \matchContains{\aComp_{\query_{n+1}',\query_{n+1}}} S_{n+1}$ where  $\aComp_{\query_{n+1}',\query_{n+1}} = \aComp_{\query_{n}',\query_{n}}$.
In addition, since $P(t_n') \subseteq P(t_n)$, $t_n = t_{n+1}$ and $t_n' = t_{n+1}'$, $P(t_{n+1}') \subseteq P(t_{n+1})$.

\underline{For $\delta$:} Assume $t_{n+1} \in \query_{n+1}(\db)$ and $\{t_{n_{1}}, \ldots, t_{n_{k}}\} \subseteq \query_{n}(\db)$ such that $\{t_{n+1}\} = \delta \{t_{n_{1}}, \ldots, t_{n_{k}}\}$. Similarly, assume $t_{n+1} \in \query_{n+1}'(\db)$ and $\{t_{n_1}', \ldots, t_{n_m}'\} \subseteq \query_{n}'(\db)$ such that $\{t_{n+1}'\} = \delta \{t_{n_1}', \ldots, t_{n_k}'\}$. Because $S_n' \matchContains{\aComp_{\query_n',\query_n}} S_n$, then
$S_{n+1}' \matchContains{\aComp_{\query_{n+1}',\query_{n+1}}} S_{n+1}$ where  $\aComp_{\query_{n+1}',\query_{n+1}} = \aComp_{\query_{n}',\query_{n}}$. Also because $S_n' \matchContains{\aComp_{\query_n',\query_n}} S_n$, $m < k$. And $P(t_{n_{1}}') \subseteq P(t_{n_{1}}), \ldots, P(t_{n_{m}}') \subseteq P(t_{n_{m}})$, then $P(t_{n+1}') \subseteq P(t_{n+1})$.
%

\underline{$\union$:} $\query_{Ln}(\db) \union \query_{Rn}(\db)$
Since $S_{nL}' \matchContains{\aComp_{\query_{nL}',\query_{nL}}} S_{nL} \land \forall (t',t) \in \mathcal{M_{nL}}: P(t') \subseteq P(t)$ and $S_{nR}' \matchContains{\aComp_{\query_{nR}',\query_{nR}}} S_{nR} \land \forall (t',t) \in \mathcal{M_{nR}}: P(t') \subseteq P(t)$, after union, the generalized containment will only hold for the common part between $\aComp_{\query_{nL}',\query_{nL}}$ and $\aComp_{\query_{nR}',\query_{nR}}$. That is, $S_{n+1}' \matchContains{\aComp_{\query_{n+1}',\query_{n+1}}} S_{n+1} \land \forall (t',t) \in \mathcal{M_{n+1}}: P(t') \subseteq P(t)$ where $\aComp_{\query_{n+1}',\query_{n+1}} = \bigwedge_{a'=a \in \aComp_{\query_{nL}',\query_{nL}} \land a'=a \in \aComp_{\query_{nR}',\query_{nR}}} a'=a$.

\underline{For $\crossprod$:} Let $\{t_{n+1}\} = \{t_{nL}\} \crossprod \{t_{nR}\}$ where $t_{nL} \in S_{nL}$ and $t_{nR} \in S_{nR}$. Similarly, let $\{t_{n+1}'\} = \{t_{nL}'\} \crossprod \{t_{nR}'\}$ where $t_{nL}' \in S_{nL}'$ and $t_{nR}' \in S_{nR}$. Since $S_{nL}' \matchContains{\aComp_{\query_{nL}',\query_{nL}}} S_{nL}$ and $S_{nR}' \matchContains{\aComp_{\query_{nR}',\query_{nR}}} S_{nR}$, then $S_{n+1}' \matchContains{\aComp_{\query_{n+1}',\query_{n+1}}} S_{n+1}$ where  $\aComp_{\query_{n+1}',\query_{n+1}} = \aComp_{\query_{nL}',\query_{nL}} \land \aComp_{\query_{nR}',\query_{nR}}$. Since $P(t_{nL}') \subseteq P(t_{nL})$ and $P(t_{nR}') \subseteq P(t_{nR})$, then $P(t_{n+1}') \subseteq P(t_{n+1})$.

\underline{For $\join_{a=b}$:} Let $\{t_{n+1}\} = \{t_{nL}\} \join_{a=b} \{t_{nR}\}$ where $t_{nL} \in S_{nL}$ and $t_{nR} \in S_{nR}$. Similarly, let $\{t_{n+1}'\} = \{t_{nL}'\} \join_{a=b} \{t_{nR}'\}$ where $t_{nL}' \in S_{nL}'$ and $t_{nR}' \in S_{nR}$. Our join rule in \cref{tab:gen-eq} keeps that $t_{nL}.a = t_{nL}'.a$ and $t_{nL}.b = t_{nR}'.b$, thus $(t_{n+1}',t_{n+1}) \in \mathcal{M}_{n+1}$ and $\mathcal{M}_{n+1} \subseteq \query_{n+1}'(\db) \times \query_{n+1}(\db)$. That is, $S_{n+1}' \matchContains{\aComp_{\query_{n+1}',\query_{n+1}}} S_{n+1}$ where  $\aComp_{\query_{n+1}',\query_{n+1}} = \aComp_{\query_{nL}',\query_{nL}} \land \aComp_{\query_{nR}',\query_{nR}}$.Since $P(t_{nL}') \subseteq P(t_{nL})$ and $P(t_{nR}') \subseteq P(t_{nR})$, then $P(t_{n+1}') \subseteq P(t_{n+1})$.

\underline{For $\Aggregation{\aggf(a) \to \agga}{\grpatts}$:}
Assume $t_{n+1} \in \query_{n+1}(\db)$ and $\{t_{n+1}\} = $ \\$\Aggregation{\aggf(a) \to \agga}{\grpatts}(\{t_{n_1},...,t_{n_m}\})$ where $\{t_{n_1},...,t_{n_m}\} \subseteq \query_{n}(\db)$, then $t_{n+1}$ either contributes to the result of $\query(\db)$ or not. Since $P(\query,\db) = P(S_n)$, if yes, $\{t_{n_1},...,t_{n_m}\} \subseteq S_n$; if not, $\{t_{n_1},...,t_{n_m}\} \subseteq \query_{n}(\db) - S_n$. That is, if $t_{in} \in S_n$ and $t_{out}\in \query_{n}(\db) - S_n$, then 
$\forall g \in G: t_{in}.g \neq t_{out}.g$.
The same holds for $t_{n+1}' \in \query_{n+1}'(\db)$. 
And because $S_n' \matchContains{\aComp_{\query_n',\query_n}} S_n$, based on the aggregation rule in \cref{tab:gen-eq-ge}, $\forall g \in G: g=g'$,  then $S_{n+1}' \matchContains{\aComp_{\query_{n+1}',\query_{n+1}}} S_{n+1}$ where  $\aComp_{\query_{n+1}',\query_{n+1}} = \aComp_{\query_{n}',\query_{n}}$ and $P(t_{n+1}') \subseteq P(t_{n+1})$.
Now we discuss the relationship between $b$ and $b'$. 
Case 1: $\textcircled{1} \land \textcircled{2}$ implies that either  $\query_n'(\db) = \query_n(\db)$ or $\query_n'(\db)$ and $\query_n(\db)$ different on the predicates of containing \lstinline!GROUP BY! attributes. If all of the operators in $\pq_n$ are monotone operators, then for the tuples in $\db$ with the same values on \lstinline!GROUP BY! attributes, either all of them will survive together until the aggregation ($op_{n+1}$) or all of them are filtered out together by the predicates before the aggregation. 
Then each common group of $op_{n+1}$ over $S_n'$ and $S_n$ are derived from the same tuples from $\db$, thus $b = b'$.
If $\pq_n$ contains aggregations, then the \lstinline!GROUP BY! attributes of each of these aggregations should include the the \lstinline!GROUP BY! attributes of $op_{n+1}$, then the case discussed above still holds. Thus $b = b'$.
Case 2: since  $\textcircled{2}$ implies that for each of these groups mentioned above, the group in $\query_n'(\db)$ contains less tuples compared with the group in $\query_n(\db)$, and thus $\agga \leq \agga'$ if $((f=sum \lor min) \land (\allcond(\querya) \rightarrow a<0))$.  
Case 3: if $(f=count \lor ((f=sum \lor max) \land (\allcond(\querya) \rightarrow a > 0)))$, then $\agga \geq \agga'$.
Case 4: Otherwise, we let relationship between $b$ and $b'$ be undecidable. 

\end{proof}
}
\iftechreport{
\begin{Lemma}\label{lem:gpe-prov-contain}
Given $\query$ and $\query'$ two instances of a parameterized query $\pq$, a database $\db$,  
consider $\mathcal{M} \subseteq \query'(\db) \times \query(\db)$ such that $\query'(\db) \matchContains{\aComp_{\query',\query}} \query(\db)$ holds based on $\mathcal{M}$, 
then
\begin{gather*}
\query'(\db) \matchContains{\aComp_{\query,\query'}} \query(\db) \land \forall (t',t) \in \mathcal{M}: P(t') \subseteq P(t) \\
    \Rightarrow \prov{\query'}{\db} \subseteq \prov{\query}{\db}
\end{gather*}
\end{Lemma}
}
\iftechreport{
\begin{proof}
Since $\query'(\db) \matchContains{\aComp_{\query,\query'}} \query(\db)$, that is $\forall t' \in \query'(\db) : \exists t \in \query(\db) : \mathcal{M}(t',t)$. Because $P(t') \subseteq P(t)$, $\prov{\query'}{\db} \subseteq \prov{\query}{\db}$.
\end{proof}
}

\iftechreport{Thus, based on the discussion by now, we could infer in a reverse order that whether the provenance sketches of $\query$ could answer $\query'$.}


\ifnottechreport{
We now demonstrate that our approach is sound. 
}

\begin{Theorem}\label{theo:reuse-safe-rule}
   Let $\query$ and $\query'$ be two instances of a parameterized query $\pq$, $\db$ be a database, 
   and $\psSet$ a set of safe provenance sketches of $\query$ with respect to $\db$. 
   \[
     \gpe(\query', \query) \land \ucond(\query',\query) \Rightarrow \psSet\,\,\text{\textbf{is safe for}}\,\, \query'
   \]
\end{Theorem}
\ifnottechreport{
  \begin{proofsketch}
    We first demonstrate that $\gpe(\query', \query) \land \ucond(\query',\query)$ implies that the generalized containment $\query'(\db) \matchContains{\aComp_{\query',\query}} \query(\db)$ holds, thus, establishing a connection between all tuples in $\query'(\db)$ and tuples of $\query(\db)$. We then show that given an arbitrary mapping $\mathcal{M}$ based on which
$\query'(\db) \matchContains{\aComp_{\query',\query}} \query(\db)$, for any $(t',t) \in \mathcal{M}$, the provenance of
$t' \in \query'(\db)$) is a subset of the provenance of $t \in \query(\db)$. By definition of generalized containment, for all $t' \in \query'(\db)$ there has to exist $t \in \query(\db)$ such that $(t',t) \in \mathcal{M}$ which immediately implies that $\prov{\query'}{\db} \subseteq \prov{\query}{\db}$. Since we have shown before that provenance containment implies that $\psSet$ is safe for $\query'$, this concludes the proof. For the detailed proof, please see \cite{techreport}.
\end{proofsketch}
}
\iftechreport{
\begin{proof}
In turn in \cref{lem:gpe} and \cref{lem:gpe-prov-contain} we get $\prov{\query'}{\db} \subseteq \prov{\query}{\db}$. Let $\psSet_{ac}$ and $\psSet_{ac}'$ be the accurate provenance sketches of $\query$ and $\query'$ with respect to $\db$. They are the same type with $\psSet$. From  \cref{lem:prov-to-ps}, $\prov{\query'}{\db} \subseteq \prov{\query}{\db} \Rightarrow \instOf{\psSet_{ac}} \supseteq \instOf{\psSet_{ac}'}$. Because $\instOf{\psSet} \supseteq \instOf{\psSet_{ac}}$, $\instOf{\psSet} \supseteq \instOf{\psSet_{ac}'}$. At last, in turn in \cref{lem:additional-fragments} and \cref{lem:ps-safe-same-template} that if $\instOf{\psSet}$ is safe to answer $\query$, then it is also safe to answer $\query'$.
\end{proof}
}
\iftechreport{
\begin{Example}\label{ex:reuse}
Consider the parameterized query $\pq$ and two instances $\query$ and $\query'$ shown in \Cref{fig:reuse-eg}.
This query returns states that have more than \texttt{\$2} cities with at least \texttt{\$1} inhabitants. $\pq$ has two conjuncts $popden \geq \$1$ and $cntcity \geq \$2$. Instances $\query$ and $\query'$ only differ in their $\$2$ bindings. 
We use $\query_{agg}$ and $\query_{agg}'$ to denote the subqueries rooted at the aggregation operator. To determine whether a set of sketches $\psSet$ for $\query$ can be used to answer $\query'$, we get the conditions shown below. We use $p$, $c$, and $s$ to denote $popden$, $city$, and  and $state$, respectively.
\begin{align*}
  \pred(\query) &= p > 100 \land cnt > 10
  &\pred(\query') &= p'> 100 \land cnt' > 15
\end{align*}\\[-10mm]
\begin{align*}
  \aComp_{\query',\query} &= p = p' \land c = c' \land s = s' \land cnt = cnt'
\end{align*}
Since this query does not contain any projections, $\expr(\query)$ and $\expr(\query')$ are empty. The condition $\gpe(\query', \query)$ constructed for this query tests the relationship between group-by attributes in the inputs of the aggregation subqueries $\query_{agg}$ and $\query_{agg'}$. Since $\aComp_{\query_{agg}',\query_{agg}}$ contains $s = s'$, $\gpe(\query', \query)$ holds. Furthermore, both \textcircled{1} and \textcircled{2} hold and, thus, we add $cnt = cnt'$ to $\aComp_{\query', \query}$. Finally, $\ucond(\query',\query)$ tests that\\[-4mm]
\[
\aComp_{\query',\query} \land \pred(\query') \land \expr(\query') \land \expr(\query) \rightarrow \pred(\query)
\]
Substituting the conditions shown above we get  $p = p' \land cnt = cnt' \land p > 100 \land cnt' > 15 \land p'> 100 \land cnt > 10$. Since this condition holds for all possible values of the variables in the formula (recall that free variables are assumed to be universally quantified), we can use $\psSet$ to answer $\query'$.

\end{Example}
}
\iftechreport{
\input{sections/reuse_eg_figs}
}

%% file: sections/ge_cond.tex
\begin{figure*}
\fbox{
\begin{minipage}{1\linewidth}
\begin{minipage}{1\linewidth}
 \begin{subfigure}{0.41\linewidth}
\vspace{-3mm}
\begin{adjustbox}{max width=1\linewidth}
{
\begin{tabular}{r@{ = }l}
  $\aComp_{R',R} $ & $\bigwedge_{a \in \schemaOf{R}} a = a'$\\
  $\aComp_{\selection_{\theta'}(\querya'),\selection_{\theta}(\querya)}$ & $\aComp_{\querya',\querya}$ \\
  $ \aComp_{\projection_A(\querya'),\projection_A(\querya)}$ & $\aComp_{\querya',\querya}$\\
     $ \aComp_{\duprem(\querya'),\duprem(\querya)}$ &$ \aComp_{\querya',\querya}$\\
 $\aComp_{\querya' \crossprod' \queryb',\querya \crossprod \queryb}$  & $\aComp_{\querya',\querya} \land \aComp_{\queryb',\queryb}$\\
 $\aComp_{\querya' \join_{\theta'} \queryb',\querya \join_{\theta} \queryb}$ & $\aComp_{\querya',\querya} \land \aComp_{\queryb',\queryb}$\\
  $\aComp_{\querya' \union' \queryb',\querya \union \queryb}$ & $\bigwedge_{i=1}^{n} (\aComp_{\querya',\querya} \rightarrow a_i = a_i' \land \aComp_{\queryb',\queryb} \rightarrow b_i = b_i') $\\
\multicolumn{2}{l}{\hspace{25mm} $\rightarrow a_i = a_i'$ \textbf{where} $\schemaOf{\querya} = (a_1, \ldots, a_n)$ }\\
\multicolumn{2}{l}{\hspace{42mm} \textbf{and} $\schemaOf{\queryb} = (b_1, \ldots, b_n)$}\\
\end{tabular}
}
\end{adjustbox}
\end{subfigure}
 \begin{subfigure}{0.59\linewidth}
\centering
\begin{adjustbox}{max width=1\linewidth}
{
\begin{tabular}{|ll|}
\hline
\rowcolor{lightgrey}
Query $\pq$                                                             & $\gpe(\query',\query)$ \\
		$R$                                                             & $\true$                \\
    $\selection_{\theta}(\pqa)$                                      & $\gpe(\querya', \querya) 
                                                                       $ 
                                                                                                 \\
    $\Aggregation{\aggf(a) \rightarrow \agga}{\grpatts}(\pqa)$       & $\gpe(\querya', \querya) \land (\forall g \in G : \aComp_{\querya',\querya} \land \allcond(\querya) \land \allcond(\querya') \rightarrow g = g')$  
   \\
    $\delta(\pqa)$                                                   & $ \gpe(\querya', \querya) \land (\forall a \in \schemaOf{\querya} :\aComp_{\querya',\querya} \land \allcond(\querya) \land \allcond(\querya') \rightarrow a = a')$ 
  \\
 $\projection_{A}(\pqa)$                                             & $\gpe(\querya', \querya)$
  \\
  $ \pqa \union \pqb$                                             & $\gpe(\querya', \querya) \land  \gpe(\queryb', \queryb)$ \\
		$ \pqa \crossprod \pqb$                                   & $\gpe(\querya', \querya) \land  \gpe(\queryb', \queryb)$ 
  \\ 
  $\pqa \join_{a=b} \pqb$                                         & $ \gpe(\querya', \querya) \land  \gpe(\queryb', \queryb) \land  (\aComp_{\querya',\querya} \land \allcond(\querya) \land \allcond(\querya') \rightarrow a=a') \land$ \\
                                                                        & $  (\aComp_{\queryb',\queryb} \land \allcond(\queryb) \land \allcond(\queryb') \rightarrow  b=b')$ 
  \\
		\hline
\end{tabular}
}
\end{adjustbox} \\[-1mm]
     \caption{$\gpe(\query',\query)$}
     \label{tab:gen-eq-ge}
\end{subfigure}
 \end{minipage}  \\[-6mm]
  \begin{subfigure}{1\linewidth}
   \resizebox{0.75\linewidth}{!}{
\begin{minipage}{0.88\linewidth}
\begin{align*} 
\aComp_{\Aggregation{\aggf(a) \to \agga}{\grpatts}'(\querya'),\Aggregation{\aggf(a) \to \agga}{\grpatts}(\querya)} &=
  \begin{cases}
     \aComp_{\querya',\querya} \land \agga = \agga' & \text{if }   \textcircled{1} \land \textcircled{2}   \\
        \aComp_{\querya',\querya} \land \agga \leq \agga' & \text{else if }  \textcircled{2}  \land ((f=sum \lor min) \land (\allcond(\querya) \rightarrow a<0))\\
      \aComp_{\querya',\querya} \land \agga \geq \agga' & \text{else if } \textcircled{2} \land (f=count \lor ((f=sum \lor max) \land (\allcond(\querya) \rightarrow a > 0))) \\
    \aComp_{\querya',\querya} & \text{otherwise}
    \end{cases}
\end{align*}
 \end{minipage}
 }
\\[1mm]
  \begin{minipage}{1\linewidth}
  \centering
  \resizebox{0.8\textwidth}{!}{
    \textbf{\textcircled{1}}  $ \aComp_{\querya', \querya}  \land \ngpred(\querya) \land \expr(\querya) \land \expr(\querya') \rightarrow \ngpred(\querya')$  \hspace{5mm}  \textbf{\textcircled{2}}   $ \aComp_{\querya', \querya}  \land \ngpred(\querya') \land \expr(\querya') \land \expr(\querya) \rightarrow \ngpred(\querya)$
    }
\end{minipage}\\[-4mm]
       \caption{$\aComp_{\query',\query}$}
      \label{tab:gen-eq-acomp}
 \end{subfigure}
  \end{minipage}
  }
  \trimfigspace
  \caption{Rules defining $\gpe(\query', \query)$ and $\aComp_{\query',\query}$ which are used to test reusability}
     \label{tab:gen-eq}
\end{figure*}


%% file: sections/reuse_eg_figs.tex
\begin{figure}[t]
\centering
\begin{minipage}{0.5\linewidth}
\centering
\begin{lstlisting}
SELECT state, count(city) AS cntcity
FROM cities
GROUP BY state
WHERE popden > $1
HAVING cntcity > $2
		\end{lstlisting}
\end{minipage} 
   \begin{minipage}{0.4\linewidth}\vspace{0.5cm}
\resizebox{0.65\textwidth}{!}{
  \begin{tabular}{|c||cc|}
    \cline{2-3} 
  \multicolumn{1}{c|}{$\,$} & $\query$ & $\query'$ \\ \hline
 \cthead{$\mathtt{\$1}$}      & 100               & 100                \\ 
  \cthead{$\mathtt{\$2}$}     & 10                & 15                 \\
  \hline
\end{tabular}
}
\end{minipage}\vspace{-3mm}
\caption{Parameterized Query $\pq$ and instances $\query$ and $\query'$}
   \label{fig:reuse-eg}
\end{figure}

%% file: sections/capture.tex
\section{Provenance Sketch Capture}
\label{sec:ps-capture}
	We now discuss how to capture provenance sketches through query instrumentation.
    We first review how queries are instrumented to propagate provenance using Lineage~\cite{CW00b,CC09,GT17} where the provenance of a query result is the set of input tuples that were used to derive the result.
\XN{   old content:
    Typically, tuples are represented through unique identifiers such as the ones shown to the left of each tuple in \Cref{tab:cities}. Some approaches encode the Lineage of a tuple, a sets of tuple identifiers, using  non-1NF representations~\cite{SJ18}. Other approaches use a denormalized flat-relational representation~\cite{AF18,glavic2013using}.
Most approaches
    operate in two phases: 1) \textbf{annotate} each input tuple with a singleton set containing its identifier and 2) \textbf{propagate} these annotation through each operator of a query such that each (intermediate) query result is annotated with its provenance. 
    }
Most approaches
    operate in two phases: 1) \textbf{annotate} each input tuple with a singleton set containing its identifier, e.g., the ones shown to the left of each tuple in \Cref{tab:cities} and 2) \textbf{propagate} these annotations through each operator of a query such that each (intermediate) query result is annotated with its provenance.

\begin{Example}\label{ex:prov-and-sketch-capture}
To capture the Lineage of each result tuple of the query $\qAvgden$ from \Cref{ex:create-sketch} we annotate each tuple $t_i$ from the cities table  (see \Cref{tab:cities})  with a singleton set $\{t_i\}$. Then annotations are propagated through the operators of $\qAvgden$. At last we get one result tuple with annotation $\{t_2, t_3\}$ (see \Cref{tab:result}).  The annotation $\{t_2, t_3\}$ means that the result tuple (CA, 5500) of $\qAvgden$ was produced by combining input tuples (6000, San Diego, CA) and (5000, Sacramento, CA). 
\end{Example}
\XN{old content:
The major difference is that we annotate tuples with a sets of fragments from a partition $\parti$ instead of sets of tuple identifiers. Thus, the size of the annotation is determined by $\card{\parti}$, i.e., the number of fragments of $\parti$, 
instead of the size of the database.
}
Our approach for computing provenance sketches also operates in two phases. The major difference is that we annotate each tuple with the fragment the tuple belongs to instead of the tuple identifier. Assume the fragment is from a partition $\parti$, then the size of the annotation is determined by $\card{\parti}$, i.e., the number of fragments of $\parti$, 
instead of the size of the database. Since the partition is fixed for a query, the annotations used for capture only need to record which fragments are present. This can be done compactly using bit sets. A partition with $n$ fragments is encoded as a vector of $n$ bits. We refer to this as the \emph{bitset encoding} of a sketch.
For instance, for the range-partition on attribute \texttt{state} from \Cref{fig:example-range-part}, the fragments $f_1$ and $f_3$ 
would be represented as \bv{1000} and \bv{0010}. 
\iftechreport{
Assume we range-partition relation cities on attribute \texttt{state} using the intervals shown in~\Cref{fig:example-range-part}. Each input is annotated with the singleton fragment it belongs to ($\{f_1\}$, $\{f_3\}$ and $\{f_4\}$) as shown in \Cref{tab:cities}.} Then, the result tuple of $\qAvgden$ (\Cref{tab:result}) is generated based on input tuples $\{t_2, t_3\}$ which are both annotated with $\{f_1\}$. Thus,  this tuple is annotated with the final sketch $\{f_1\}$. 
\subsection{Initializing Annotations}
\label{sec:annotate}

We now discuss how to seed the tuple annotations for a query $\query$ according to a set of range partitions $\dbpart = \{ \parti_1, \ldots, \parti_m \}$ over database $\db$. Let $\parti_i$ be the partition for relation $\rel_i$ where $i \in [1,m]$. Note that if the query is not selective in terms of provenance wrt. to $\rel_i$, then it may not be beneficial to include this relation in the set of sketches. 
To simplify the presentation, we assume that no relation is accessed more than once by $\query$, but our approach also handles multiple accesses.
Recall that a range partition (\Cref{def:range})  
assigns tuples to fragments based on their value in an attribute $\att$ and a set of ranges
 over the domain of $\att$.  
We add a projection on top of $\rel_i$ to compute and store the fragment in a column $\psa{\rel_i}$ through a \lstinline!CASE! expression.
We use $\initialize{\rel_i}{\parti_i}$ to denote this instrumentation step.
In relational algebra, we use
$\scase{\swhen{\theta_1}{e_1}, \ldots, \swhen{\theta_n}{e_n}}$ to denote an expression that returns
the result of the first $e_i$ for which condition $\theta_i$ evaluates to true and returns null if all $\theta_i$ fail. We use $\bvsng{i}$ to denote the singleton bit set for $\{f_i\}$. For a range partition 
 $\parti_{\ranges,a}(R)$ with ranges $\ranges = \{\range_1, \ldots, \range_n\}$ we generate query: 
\XN{ old contetnt:
We now discuss how to seed the tuple annotations  for a query $\query$ according to a database partition $\dbpart = \{ \parti_1, \ldots, \parti_m \}$.
Note that for the sake of creating a sketch it is sufficient to have partitions for a subset of the relations accessed by query $\query$. For example, if the query is not selective in terms of provenance wrt. to a relation then it may not be beneficial to include this relation in the sketch. If a sketch is not covering a relation $\rel$ accessed by $\query$, this just means that we would not be able to apply this sketch to filter $\rel$. }
\BG{We may use different partitions for each mention of a relation in the query since these may correspond to different access patterns.}
\BG{Note that while we defined a database partition to consist of one partition for each table in the database, for creating a provenance skecth for a query we obviously only need a partition for each relation accessed by the query. Furthermore, we are free to only create a sketch that covers a subset of these relations, e.g., if the query is not selective in terms of provenance wrt. to a relation then it would not be beneficial to include this relation in the sketch. Additionally, we may use different partitions for each mention of a relation in the query since these may correspond to different access patterns.} 
\XN{ old contetnt:
Let $\rel$ be an input relation  of a query $\query$  and $\parti$ be the partition from $\dbpart$  for this relation.
To instrument the query to annotate each tuple from $\rel$ with the fragment of $\parti$ it belongs to,  we add a projection on top of the table access $\rel$ to compute and store this fragment in a column $\psa{\rel}$. To simplify the presentation, we assume that no relation is accessed more than once by $\query$, but our approach also handles multiple accesses.}
\XN{old content:
\parttitle{Range-based Partitioning}
A range partition (\Cref{def:range})  
assigns tuples to fragments based on their value in an attribute $\att$ and a set of ranges
 over the domain of $\att$.  
We use $\initialize{\rel}{\parti}$ to denote the instrumented query produced for the range partition $\parti$ of relation $\rel$.
In SQL we can express this using \lstinline!CASE!. For relational algebra we have to allow for conditional expressions in projections. We use
$\scase{\swhen{\theta_1}{e_1}, \ldots, \swhen{\theta_n}{e_n}}$ to denote an expression that returns
the result of the first $e_i$ for which condition $\theta_i$ evaluates to true and returns null if all $\theta_i$ fail. We use $\bvsng{i}$ to denote the singleton bit set for $\{f_i\}$. For a range partition 
 $\parti_{\ranges,a}(R)$ with ranges $\ranges = \{\range_1, \ldots, \range_n\}$ we generate the following query:}
\begin{equation}
 \label{eq:annot-init-range}
  \initialize{\rel}{\parti} \defas \projection_{\scase{\swhen{\att \in \range_1}{\bvsng{1}}, \ldots, \swhen{\att \in \range_n}{\bvsng{n}}}}(R)
\end{equation}

For example, 
to instrument the relation access from query $\qAvgden$ from  \Cref{tab:cities} using the partition $\rparti_{state}$  
(\Cref{fig:example-range-part}), we generate query $\query_{INIT}$  shown below (written in SQL for legibility).
Based on the value of attribute \texttt{state} we assign tuples to fragments of $\rparti_{state}$.

\begin{adjustbox}{max width=0.97\linewidth}
\begin{minipage}{0.89\linewidth}
 \lstset{tabsize=4,style=psqlcolor,basicstyle=\scriptsize\upshape\ttfamily}
\begin{lstlisting}
SELECT popden, city, state,
	CASE  WHEN  state >= 'AL' AND state <= 'DE' THEN '1000'
				WHEN  state >= 'FL' AND state <= 'MI' THEN '0100'
				WHEN  state >= 'MN' AND state <= 'OK' THEN '0010'
				WHEN  state >= 'OR' AND state <= 'WY' THEN '0001'
    END AS $\psa{state}$
FROM cities
\end{lstlisting}
\end{minipage}
\begin{minipage}{0.09\linewidth}
\fbox{$\bf{\query_{INIT}}$}
\end{minipage}
\end{adjustbox}
\newcommand{\ruler}{\ensuremath{r_0}\xspace}
\newcommand{\ruleproj}{\ensuremath{r_1}\xspace}
\newcommand{\rulesel}{\ensuremath{r_2}\xspace}
\newcommand{\ruleagg}{\ensuremath{r_3}\xspace}
\newcommand{\rulejoin}{\ensuremath{r_4}\xspace}
\newcommand{\ruleord}{\ensuremath{r_5}\xspace}
\newcommand{\rulewin}{\ensuremath{r_5}\xspace}
\newcommand{\ruleunion}{\ensuremath{r_6}\xspace}
\newcommand{\ruleinstr}{\ensuremath{r_7}\xspace}

\subsection{Propagating Annotations}
\label{tab:propagate}

We now discuss how to instrument a query to propagate annotations to generate a single output tuple storing the provenance sketch for the query. We denote the set of attributes of an instrumented query storing provenance sketches as $\psallatt$. Given a database partition $\dbpart$ and query $\query$, we use $\instrument{\dbpart,\query}$ to denote the result of instrumenting the query to capture a provenance sketch for $\dbpart$. For two lists of attributes $A= (a_1 \ldots, a_n)$ and $B = (b_1, \ldots, b_n)$ we write $A = B$ as a shortcut for $\bigwedge_{i \in \{1, \ldots, n\}} a_i = b_i$.
We apply similar notation for bulk renaming $A \to B$ and function application, e.g., $f(A)$ denotes $f(a_1), \ldots, f(a_n)$. We assume the existence of an aggregation function $\bitor$ which computes the bit-wise or of a set of bit sets.\footnote{For example, in Postgres this function exists under the name \texttt{bit\_or}.}
The rules defining $\instrument{\cdot}$ are shown in \Cref{fig:cap-rules}.
As the last step of the rewritten query $\instrument{\cdot}$ we apply $\bitor$ aggregation to merge the sketches annotations of the results of the query (rule $\ruleinstr$). The input to this aggregation is generated using $\propagate{\dbpart,\query}$ which recursively replaces $\query$ with an instrumented version. 

\ifnottechreport{
\begin{figure}[t]
\centering
\begin{adjustbox}{max width=1\linewidth}
  \begin{minipage}{1.6\linewidth}
  \begin{minipage}{0.44\linewidth}
       \begin{align*}
\propagate{\dbpart,R}                                                                                                & = \initialize{R}{\dbpart} \tag{\ruler}                                                                      \\
\propagate{\dbpart,\projection_{A}(Q)}                                                                               & = \projection_{A, \psallatt}(\propagate{\dbpart,Q}) \tag{$r_1$}                                             \\
\propagate{\dbpart,\selection_{\theta}(Q)}                                                                           & = \selection_{\theta}(\propagate{\dbpart,Q}) \tag{$r_2$}
         \end{align*}
      \end{minipage}
        \begin{minipage}{0.56\linewidth}
      \begin{align*}
         \propagate{\dbpart, Q_1 \crossprod Q_2}                                                                         & = \propagate{\dbpart,Q_1} \crossprod \propagate{\dbpart,Q_2} \tag{$r_4$}                                    \\
     \propagate{\dbpart, \ordlimit{O}{C}(\query)}                                                             & = \ordlimit{O}{C}(\propagate{\dbpart,Q}) \tag{$r_5$}                         \\
\     \propagate{\dbpart, Q_1 \union Q_2}                                                                            & =
\propagate{\dbpart,Q_1} \union \propagate{\dbpart,Q_2} \tag{$r_6$}
\end{align*}
 \end{minipage}\\[-1mm]
 \begin{minipage}{1\linewidth}
         \begin{align*}
     \propagate{\dbpart, \Aggregation{f(a)}{G}(Q)}                                                                   & =
                                                 \begin{cases}
       \projection_{a, G, \psallatt} ( \Aggregation{f(a)}{G}(Q) \join_{f(a) = a \wedge G = G} \propagate{\dbpart,Q}) & \text{if}\, f = min \vee f = max                                                                            \\
                                                    \Aggregation{f(a),bitor(\psallatt)}{G}(\propagate{\dbpart,Q})    & \text{otherwise}                                                                                            \\
                                                                                                          \end{cases}                                                                                                  \tag{$r_3$}
         \end{align*}\\[-7mm]
\begin{align*}
\instrument{\dbpart, Q}                                                                                              & =
\aggregation_{bitor(\psallatt) \rightarrow \psallatt}(\propagate{\dbpart,\query}) \tag{$r_7$}
\end{align*}                                                                                                                                                                                                                    \\[-14mm]
\end{minipage}
  \end{minipage}
  \end{adjustbox}
    \caption{Instrumentation rules for sketch capture}
   \label{fig:cap-rules}
 \end{figure}
}
\iftechreport{
\begin{figure}[t]
\centering

\begin{adjustbox}{max width=1\linewidth}
  \begin{minipage}{1.6\linewidth}
       \begin{align*}
\propagate{\dbpart,R}                                                                                                & = \initialize{R}{\dbpart} \tag{\ruler}                                                                      \\
\propagate{\dbpart,\projection_{A}(Q)}                                                                               & = \projection_{A, \psallatt}(\propagate{\dbpart,Q}) \tag{$r_1$}                                             \\
\propagate{\dbpart,\selection_{\theta}(Q)}                                                                           & = \selection_{\theta}(\propagate{\dbpart,Q}) \tag{$r_2$}
         \end{align*}
         \begin{align*}
     \propagate{\dbpart, \Aggregation{f(a)}{G}(Q)}                                                                   & =
                                                 \begin{cases}
       \projection_{a, G, \psallatt} ( \Aggregation{f(a)}{G}(Q) \join_{f(a) = a \wedge G = G} \propagate{\dbpart,Q}) & \text{if}\, f = min \vee f = max                                                                            \\
                                                    \Aggregation{f(a),bitor(\psallatt)}{G}(\propagate{\dbpart,Q})    & \text{otherwise}                                                                                            \\
                                                                                                          \end{cases}                                                                                                  \tag{$r_3$}
         \end{align*}
\begin{align*}
         \propagate{\dbpart, Q_1 \crossprod Q_2}                                                                         & = \propagate{\dbpart,Q_1} \crossprod \propagate{\dbpart,Q_2} \tag{$r_4$}                                    \\
     \propagate{\dbpart, \ordlimit{O}{C}(\query)}                                                             & = \ordlimit{O}{C}(\propagate{\dbpart,Q}) \tag{$r_5$}                         \\
\     \propagate{\dbpart, Q_1 \union Q_2}                                                                            & =
\propagate{\dbpart,Q_1} \union \propagate{\dbpart,Q_2} \tag{$r_6$}                                                                                                                                                                 \\[2mm]
\instrument{\dbpart, Q}                                                                                              & =
\aggregation_{bitor(\psallatt) \rightarrow \psallatt}(\propagate{\dbpart,\query}) \tag{$r_7$}
\end{align*}                                                                                                                                                                                                                    \\[-10mm]

  \end{minipage}
  \end{adjustbox}
    \caption{Instrumentation rules for sketch capture}
   \label{fig:cap-rules}
 \end{figure}
}

%
%

 Rule $r_0$ initializes the sketch annotations for relation $R$ using $\initialize{\cdot}{}$ as introduced in \Cref{sec:annotate}. For projection we only need to add the $\psallatt$ columns from its input to the result schema ($r_1$). Selection is applied unmodified to the instrumented input ($r_2$).
 A result tuple of an aggregation operator with group-by is produced by evaluating the aggregation function(s) over all tuples from the group. Thus, if each tuple $\tup$ from a group is annotated with a set of fragments that is sufficient to produce $\tup$, then the union of these fragments is sufficient for reproducing the result for this group. Since aggregation is non-monotone, this is only correct if we use a partitioning that has been determined to be \emph{safe} for the query using the method from ~\Cref{sec:safety-check}.
\BG{REMOVED, IF WE WANT THIS WE SHOULD MOVE IT TO THE SAFETY SECTION:  This can happen when only some tuples from a group belong to a fragment containing provenance.} Hence, we union the provenance sketches for each group using the bitwise or  aggregation function $bitor$ ($r_3$), e.g., \bv{1000} and \bv{0010} will be merged producing \bv{1010}. 
For aggregation functions $min$ and $max$ it is sufficient to only include tuples with the min/max value in attribute $\att$.  We implement this by selecting a single tuple with the min/max value for each group. \BG{Sketch quality may be improved if we select this smartly by using an aggregation function that returns the smallest sketch from the input}  
For cross product we compute the cross product of the instrumented inputs ($r_4$).
For order by with limit we compute we order the instrumented inputs with limit ($r_5$).
\XN{
The window operator partitions its input on a set of attributes $G$ and within each fragment sorts the tuples on a list of sort expressions $O$. Each input tuple $\tup$ is then extended with the result of aggregation $f$ applied to the set of tuples from the window for $\tup$ which are all tuples with the same group-by value as $\tup$ that are smaller than or equal to $\tup$ wrt. to the sort order. To reproduce a window we just have to ensure that all tuples from the window are in the input and that no new tuples are produced that did originally not belong to the window. We satisfy the $1^{st}$ condition  by instrumenting the window operator to union the sketches for each tuple in the window. The $2^{nd}$ condition as long as 
the partitions are safe according to  our rules from \Cref{sec:safety-check} ($r_5$). }
For union we union the instrumented inputs ($r_6$).
\begin{Theorem}\label{theo:instrument-rule}
  Consider a query $\query$, database $\db$, and a safe partitioning $\dbpart$ for $\db$. Then $\instrument{\dbpart, \query}(\db)$ produces a safe sketch.
\end{Theorem}

\iftechreport{
\begin{proof}
\BG{Polish proof}
We prove this by induction over 
the relational algebra expression.
\underline{Base case ($r_0$ [table access]):} Since $\initialize{\rel}{\parti}$ annotates each tuple $t$ in R with the fragment $\frag$ it belongs to, i.e., $t \in \frag$, the sketch produced by $\instrument{\query, \dbpart}$ contains all tuples from $R$ and, thus, is trivially safe.
\BG{$\frag$ is the provenance sketch of tuple $t$ noted as $D_{\provSketch_{t}}$, thus for $t \in R$, exists $t \in D_{\provSketch_{t}}$.}
\underline{Inductive Step:}
assume given a query $Q'$ and database $D$, if exists tuple $t'$ such that $t' \in Q'(D)$, then $t' \in Q'(D_{\provSketch_{t'}})$. Next, We prove this is also hold for the query $Q$ where $Q = op(Q')$ after applying our instrumentation rules for this query operator (op). In the following, we prove each op individually.
\underline{$r_1$ [$\projection_{A}$]:} here $A$ represents the set of attributes used in the projection. Assume
$t \in \projection_{A}(Q'(D))$, then we have $t = t'.A$ and $D_{\provSketch_{t}} = D_{\provSketch_{t'}}$,
thus $t' \in Q'(D_{\provSketch_{t}})$. Since $t = t'.A$, we get $t \in \projection_{A}(Q'(D_{\provSketch_{t}}))$.
\underline{$r_2$ [$\selection_{\theta}$]:} Assume $t \in \selection_{\theta} (Q'(D))$, then we have $t = t'$ and $D_{\provSketch_{t}} = D_{\provSketch_{t'}}$, similarly,
$t \in \selection_{\theta} (Q'(D_{\provSketch_{t}}))$.
\underline{$r_3$ [$\aggregation$]:} assume $t \in \aggregation(Q'(D))$ and $t$ is derived from $t'_{1}, t'_{2},...,t'_{n}$ where $t'_{i} \in Q'(D)$. Then $t = \aggregation(\{t'_{1},t'_{2},...,t'_{n}\})$ and $D_{\provSketch_{t}} = D_{\provSketch_{t'_{1},t'_{2},...,t'_{n}}} = D_{\provSketch_{t'_{1}}} \cup D_{\provSketch_{t'_{n}}} \cup ... \cup D_{\provSketch_{t'_{n}}}$. Since $\{t'_{1},t'_{2},...,t'_{n}\} \in Q'(D_{\provSketch_{t'_{1}}}) \cup Q'(D_{\provSketch_{t'_{n}}}) \cup ... \cup Q'(D_{\provSketch_{t'_{n}}}) = Q'(D_{\provSketch_{t}}) $, we get $t \in \aggregation(Q'(D_{\provSketch_{t}}))$. When $\aggregation \in $ \{min, max\}, only one $t'$ in $\{t'_{1}, t'_{2},...,t'_{n}\}$ is enough to derive $t$, the proof is same.
\underline{$r_4$ [$Q'_{L}(D) \crossprod Q'_{R}(D)$]:} assume $t'_{L} \in Q'_{L}(D)$ and $t'_{R} \in Q'_{R}(D)$, then $t = t'_{L} \crossprod t'_{R}$ and $D_{\provSketch_{t}} = D_{\provSketch_{t'_{L}}} \cup D_{\provSketch_{t'_{R}}}$ , thus $t'_{L} \crossprod t'_{R} \in Q'_{L}(D_{\provSketch_{t'_{L}}}) \crossprod Q'_{R}(D_{\provSketch_{t'_{R}}}) = Q'_{L}(D_{\provSketch_{t}}) \crossprod Q'_{R}(D_{\provSketch_{t}})$, that is, $t \in Q'_{L}(D_{\provSketch_{t}}) \crossprod Q'_{R}(D_{\provSketch_{t}})$.
\underline{$r_5$ [$\ordlimit{O}{C}(D)$]:} Assume $t \in \ordlimit{O}{C} (Q'(D))$, then we have $t = t'$ and $D_{\provSketch_{t}} = D_{\provSketch_{t'}}$, similarly,
$t \in \ordlimit{O}{C} (Q'(D_{\provSketch_{t}}))$.
\underline{$r_6$ [$Q'_{L}(D) \union  Q'_{R}(D)$]:} assume $t \in Q'_{L}(D) \union Q'_{R}(D)$ and $t = t'$ where $t' \in Q'_{L}(D)$, then $D_{\provSketch_{t}} = D_{\provSketch_{t'}}$ such that $t' \in Q'_{L}(D_{\provSketch_{t}}) \in  Q'_{L}(D_{\provSketch_{t}}) \union Q'_{R}(D_{\provSketch_{t}})$.
Thus, $t \in Q'_{L}(D_{\provSketch_{t}}) \union Q'_{R}(D_{\provSketch_{t}})$.
\end{proof}
}
\iftechreport{
\begin{Example}
  Reconsider query $\qAvgden$ from our running example. We now explain the steps of generating $\instrument{\qAvgden,\parti_{state}}$. \Cref{fig:eg-q-capture}
  shows the resulting query and \Cref{fig:eg-capture-inter-results} shows its intermediate and final result(s).
  In relational algebra, query $\qAvgden$ can be written as
  $\projection_{state, avgden}(\ordlimit{\sdesc{avgden}}{1}(\Aggregation{avg(popden) \to avgden}{state}(cities)))$.
\XN{We use the \lstinline!rank! window function to determine the tuple with the highest \texttt{avgden} value. Rank determines for each tuple its position within its window according to the sort order of the window function ignoring duplicates (corresponding to \lstinline!rank()! in SQL).}
  We start from the relation access which is instrumented as shown at the end of \Cref{sec:annotate} (\fignumref{4} in \Cref{fig:eg-q-capture}).
  We apply $\ruleagg$ to instrument the  aggregation to propagate the values in $\psa{\parti_{state}}$. The rule uses $\bitor$ to aggregate input sketches for each group (see \fignumref{3} in \Cref{fig:eg-q-capture}  and  \Cref{tab:cp-eg-result1}).  
  \XN{We then rewrite the window operator and selection using rules $\rulewin$ and $\rulesel$ (corresponding to the \lstinline!ORDER BY! and \lstinline!LIMIT! in SQL). } 
  We then rewrite the order by with limit using rules $\ruleord$  and
the corresponding rewrite at the SQL level does not require any modification (\fignumref{2} in \Cref{fig:eg-q-capture} and \Cref{tab:cp-eg-result2}). Using $\ruleinstr$, we then apply $\bitor$ aggregation  
to generate the final sketch: $\{f_1\}$ (\fignumref{1} in \Cref{fig:eg-q-capture} and \Cref{tab:cp-eg-result}). 
\end{Example}
}
  \ifnottechreport{
\input{./sections/cap_eg_short}
  }
%

\iftechreport{
\begin{figure}[t]
  \begin{minipage}{0.55\linewidth}
    { \footnotesize
\begin{minipage}{1\linewidth}
   \begin{subfigure}{1\linewidth}
\begin{tabular}{|c|c|c|c|}  \hline
\cthead popden & \cthead city & \cthead state &  \cthead $\psatt_{\parti_{state}}$ \\
 4200 & Anchorage & AK  & \bv{1000} \\
 6000	& San Diego &  CA &  \bv{1000}\\
5000	&  Sacramento & CA & \bv{1000}\\
7000 &  New York & NY & \bv{0010} \\
2000 &   Buffalo & NY &  \bv{0010}\\
3700 & Austin & TX & \bv{0001}\\
2500 & Houston & TX & \bv{0001}\\

 \hline
\end{tabular}\\[-2mm]
   	\caption{Result of \fignumref{4}}
     \label{tab:cp-eg-input}
   \end{subfigure}
\end{minipage}
\vspace{5pt}
\begin{minipage}{1\linewidth}
   \begin{subfigure}{1\linewidth}
   \centering
\begin{tabular}{|c|}  \hline
\cthead $\psatt_{\parti_{state}}$    \\
\bv{1000} \\
 \hline
\end{tabular}\\[-2mm]
     \caption{Result of \fignumref{1}}
     \label{tab:cp-eg-result}
   \end{subfigure}
 \end{minipage}
 }
\end{minipage}
\begin{minipage}{0.44\linewidth}
{\footnotesize
\begin{minipage}{1\linewidth}
   \begin{subfigure}{1\linewidth}
   \centering
\begin{tabular}{|c|c|c|}  \hline
\cthead state & \cthead avgden & \cthead $\psatt_{\parti_{state}}$ \\
  CA & 5500 & \bv{1000} \\
  NY & 4500 & \bv{0010} \\
  AK & 4200 & \bv{1000} \\
  TX & 3100 & \bv{0001} \\ \hline
\end{tabular}\\[-1mm]
     \caption{Result of \fignumref{3}}
     \label{tab:cp-eg-result1}
   \end{subfigure}
\end{minipage}
\vspace{7pt}
\begin{minipage}{1\linewidth}
   \begin{subfigure}{1\linewidth}
   \centering
\begin{tabular}{|c|c|c|}  \hline
\cthead state & \cthead avgden & \cthead $\psatt_{\parti_{state}}$ \\
  CA & 5500 & \bv{1000} \\\hline
\end{tabular}
     \caption{Result of \fignumref{2}}
     \label{tab:cp-eg-result2}
   \end{subfigure}
\end{minipage}
}
\end{minipage}\\[-6mm]
\caption{Intermediate and final result(s) of $\instrument{\parti_{state},\qAvgden}$}
\label{fig:eg-capture-inter-results}
\end{figure}
}

%

\iftechreport{
\input{./sections/cap_eg_instrument}
}

\subsection{Optimizations}
\label{tab:opt_capture}
Our instrumentation preserves the structure of the input query
in most cases. Thus, the majority of overhead introduced by instrumentation 
is based on evaluating 1) \lstinline!CASE! expressions and 2) bitor aggregations. 
\XN{old: For 1) to initialize a range-based provenance sketch with $n$ fragments we generate a \lstinline!CASE! expression with $n$ conditions each testing the membership of a value $v$ in range $\range_i$. We can reduce the runtime from $\oNotation n$ to $\oNotation {\log n}$ by applying binary search to find the range $r_i$ that contains $v$.}
For 1) to initialize a range-based provenance sketch with $n$ fragments, we can apply binary search to test the membership of a value $v$ in range $\range_i$ which reduces the runtime from $\oNotation n$ to $\oNotation {\log n}$.
 We implemented this optimization as UDFs written in C in MonetDB and Postgres, the two systems we use in our experimental evaluation.
For 2) If $n$ is large, then singleton sets of fragments can be encoded more compactly by storing and propagating the position of the single bit set to $1$ as a fixed-size integer value instead of storing and propagating a full bitset until encountering the aggregation to construct the full bitset instead of doing bitor operation
\ifnottechreport{(\textit{delay} method).}
\iftechreport{which we call \textit{delay} method.}
Furthermore, in Postgres, this bitor aggregation function results in unnecessary creation of $n-1$ new bitsets when calculating the bitwise or of $n$ bit sets. 
Also, bitwise or is applied one byte at a time. We improve this implementation by computing the operation one machine-word at a time and by avoiding unnecessary creation of intermediate bitsets
\ifnottechreport{(\textit{No-copy} method).}
\iftechreport{which we call \textit{No-copy} method.}For MonetDB we implement $\bitor$ 
as a user-defined aggregation function in C.

%% file: sections/cap_eg_short.tex
Reconsider query $\qAvgden$ from our running example, after instrumentation, $\qAvgden$ would be rewritten as: 

\begin{minipage}{1\linewidth}
  \lstset{tabsize=4,style=psqlcolor,basicstyle=\scriptsize\upshape\ttfamily}
\begin{lstlisting}
SELECT bitor($\psatt_{\parti_{state}}$) AS $\psatt_{\parti_{state}}$
FROM (SELECT state,avg(popden) AS avgden, bitor($\psatt_{\parti_{state}}$) AS $\psatt_{\parti_{state}}$
	  	FROM $\query_{INIT}$ GROUP BY state ORDER BY avgden DESC LIMIT 1) 
\end{lstlisting}
\end{minipage}
Above query returns $\{1000\}$ which represents 
the final sketch is $\{f1\}$. Subquery $\query_{INIT}$ is shown in the end of  \Cref{sec:annotate}. 

%% file: sections/cap_eg_instrument.tex
\begin{figure}[t]
\begin{adjustbox}{max width=1\linewidth}
\begin{tabular}{clc}
\cellcolor{white} \fignumref{1} &
\begin{lstlisting}[mathescape]
SELECT bitor($\psatt_{\parti_{state}}$) AS $\psatt_{\parti_{state}}$
FROM (
\end{lstlisting}
& \cellcolor{white} \large \textbf{(\ruleinstr)} \\
\cellcolor{lllgrey} \fignumref{3} &
\cellcolor{lllgrey}
\begin{lstlisting}[mathescape]
	SELECT state,
				 avg(popden) AS avgden,
				 bitor($\psatt_{\parti_{state}}$) AS $\psatt_{\parti_{state}}$
	FROM (
\end{lstlisting}
&  \cellcolor{lllgrey}  \large \textbf{(\ruleagg)} \\
\cellcolor{llgrey} \fignumref{4} & \cellcolor{llgrey} 
\begin{lstlisting}[mathescape]
			  $\query_{INIT}$ 
\end{lstlisting} \hspace{2mm} {\scriptsize (see \Cref{sec:annotate})}
&  \cellcolor{llgrey}  \large \textbf{(\ruler)} \\
\cellcolor{lllgrey} \fignumref{3} &
\cellcolor{lllgrey}
\begin{lstlisting}
		   ) p1 
	GROUP BY state
\end{lstlisting}
&   \cellcolor{lllgrey}  \large \textbf{(\ruleagg)} \\
\cellcolor{white} \fignumref{2} & \hspace{0.1mm}
\begin{lstlisting}
	ORDER BY avgden DESC 
	LIMIT 1) p2
\end{lstlisting}
&  \cellcolor{white}  \large \textbf{(\ruleord)} \\
\end{tabular}
\end{adjustbox}\\[-3mm]
\caption{Instrumented Query $\instrument{\parti_{state},\qAvgden}$}
\label{fig:eg-q-capture}
\end{figure}

%% file: sections/use.tex
\section{Using Provenance Sketches}
\label{sec:ps-reuse}
Once a provenance sketch $\provSketch$ has been  captured, we can utilize it to speed up the subsequent execution of queries.
For that we have to instrument the query to filter out data that does not belong to the provenance sketch.
This is achieved by decoding the provenance sketch into selection conditions and applying these conditions by adding selection operators on top of every relation access that is covered by the sketch. Recall that we use  $\quse{\query}{\provSketch}$ to denote the result of instrumenting query $\query$ using sketch $\provSketch$. $\quse{\query}{\provSketch}$ is defined as the identity function on all operators except for table access operators.
Let $\parti$ be a range-based partition of a relation $\rel$ on attribute $\att$ using ranges $\ranges = (\range_1, \ldots, \range_n)$ and  $\provSketch = \{ \frag_{i_1}, \ldots, \frag_{i_m} \}$ be a sketch based on $\parti$. We generate a
condition $\bigvee_{j=1}^{m} \att \in \range_{i_j}$ to filter $\rel$ based on $\parti$.
Thus, the instrumentation rule for applying range-based sketch $\parti$ to relation $\rel$ \ifnottechreport{is $ \quse{R}{\provSketch} \defas \selection_{\bigvee_{j=1}^{m} \att \in \range_{i_j}}(R)$.}
\iftechreport{is: 
\begin{equation}\label{eq:use-ps-rewrite}
  \quse{R}{\provSketch} \defas \selection_{\bigvee_{j=1}^{m} \att \in \range_{i_j}}(R)
\end{equation}

}
\iftechreport{Instrumentation for other partition schemes operates in a similar fashion, e.g., for  a hash-based partition 
we produce  a disjunction $\bigvee_{j=1}^{m} h(a) = i_j$.}
\subsection{Optimizations}
\label{tab:use_opt}
Databases can exploit physical design to evaluate the type of selection conditions we create for range-based sketches.
However, if $\card{\provSketch}$ is large, i.e., the sketch contains a large number of fragments, then the size of the selection condition  that has to be evaluated may outweigh this benefit.
Furthermore, if the database has to resort to a full table scan, then we pay the overhead of evaluating a condition that is linear in $\card{\parti}$ for each tuple.
We now discuss how to improve this by reducing the number of conditions and/or improving the performance of evaluating these conditions. First off,
if a sketch contains a sequence of adjacent fragments $f_{i}, \ldots, f_{j}$ for $i < j$, we can replace the conditions $\bigvee_{k=i}^{j} \att \in \range_{k}$ with a single condition $\att \in \bigcup_{k=i}^j r_{k}$. Reconsider the sketch  $\provSketch = \{ \frag_1, \frag_2 \}$ from the example above. 
Since these two fragments are adjacent, we can generate a single condition $state \in [AL,MI]$ instead of $state \in [Al,DE] \vee state \in [FL,MI]$.
Note that the condition generated for a range partition checks whether an attribute value is an element of one of the ranges corresponding to the fragments of the sketch. Since these ranges are ordered, we can apply binary search to improve the performance of evaluating a condition with $n$ disjunctions from $\oNotation{n}$ to $\oNotation{\log n}$.
We implemented a Postgres extension
to be able to  exploit zone maps (brin indexes in Postgres) to skip data based on such a condition.

%% file: sections/experiments.tex
\section{Experiments}
\label{sec:exp}

We now evaluate the impact of parameters such as
number of fragments on the precision of provenance sketches, measure the overhead of capturing sketches, and evaluate the speed-up gained by using provenance sketches.
\BG{Our evaluation consists of three parts:  
(i) we evaluate the effectiveness of our optimization methods for capturing the provenance sketches;  (ii) we evaluate provenance sketch capture and use for standard benchmark queries and over real world datasets; Additionally, through the result, we learn the relationship between the number of fragments which should be created in a provenance sketch and the number of times the query would be run afterwards;  (iii) we device a simplistic prototype for applying our techniques online to a workload that is not known apriori
to learn how effective our techniques are in a self-tuning environment, 
i.g., for each incoming query whether or not to create a provenance sketch and whether or not to use any captured provenance sketch.
}
All experiments were run on a machine with 2 x 3.3Ghz AMD Opteron 4238 CPUs (12 cores) 
and 128GB RAM running Ubuntu 18.04 (linux kernel 4.15.0).
We use Postgres 11.4 as an example of a classical disk-based system and MonetDB 11.33.11 as an example of a columnar main-memory database system.


\subsection{Workloads}
\label{sec:dataset}

\parttitle{TPC-H} We use all 22 query templates of the TPC-H benchmark and SF1 and SF10 to evaluate performance. 

\iftechreport{
\parttitle{Crimes} This dataset records crimes reported in Chicago (\url{https://data.cityofchicago.org/Public-Safety/Crimes-2001-to-present/ijzp-q8t2}). It contains
 $\sim$6.7M tuples.
We use two queries: 
\texttt{C-Q1}: Compute the 5 areas with the most crimes.
\texttt{C-Q2}: Return the number of blocks where more than 10000 crimes took place.
}

\ifnottechreport{\parttitle{Movie Ratings}  The MovieLens dataset (\url{https://grouplens.org/datasets/movielens}) 
contains a movie relation ($\sim$27k movies) and a ratings relation ($\sim$20m ratings).}
\iftechreport{
\parttitle{Movie Ratings}  The MovieLens dataset (\url{https://grouplens.org/datasets/movielens}) 
contains a movie relation ($\sim$27k movies) and a ratings relation ($\sim$20m ratings).
We use three queries:
\texttt{M-Q1}: Compute the 10 movies with the most ratings. \texttt{M-Q2}: Return the number of movies with more than 63,300 ratings. \texttt{M-Q3}: Return the 10 most popular movies where popularity is defined as the weighted sum of the number of ratings of a movie and the number of times it has been tagged.
}

\iftechreport{
\parttitle{Stack Overflow (SFO)}
}
\ifnottechreport{
\parttitle{Stack Overflow}}  
This in archive of content from Stack Overflow (\url{https://www.kaggle.com/stackoverflow/stackoverflow}). It consists of a users relation ($\sim$12.5m users), 
a badges relation ($\sim$35.9m badges), a comments relation ($\sim$75.9m comments) and a posts relation ($\sim$48.5m posts).
We use five real queries from \url{https://data.stackexchange.com/stackoverflow/queries}: \texttt{S-Q1}: 
Return the 10 users with the most number of posts.
\texttt{S-Q2}: 
 Return the 10 owners whose post is favored by the most people.
\texttt{S-Q3}: 
Return the 10 users with the most number of comments.
\texttt{S-Q4}: 
Return the 10 users with the most number of badges.
\texttt{S-Q5}: 
Return all users who did post between 47945 and 52973 comments.

%
\BG{Needed here? 1\%, 5\%, 10\%, 20\% and 50\%.}

\ifnottechreport{
\input{./sections/exp_figs_short3} 
}
\iftechreport{
\input{./sections/exp_figs_long}
}

\ifnottechreport{
\subsection{Capture Optimizations}
\label{sec:cap-opt}
In preliminary experiments, we have evaluated the effectiveness of the optimization for provenance capture presented in \Cref{tab:opt_capture}. These experiments demonstrated that using binary search to determine which partition an input tuple belongs to significantly outperforms the use of case expressions for this purpose. Furthermore, delaying creating sets of sketches and avoiding copying when constructing sketches further improves performance. Thus, we enable these optimizations for all remaining experiments. For additional details about these experiments see \cite{techreport}.
}
\iftechreport{
\subsection{Capture Optimizations}
\label{sec:cap-opt}
We first evaluate the effectiveness of the optimization for provenance capture presented in \Cref{tab:opt_capture}.
\Cref{fig:case-opt,fig:bitor-opt} show capture runtime varying $\card{\parti}$, i.e., the number of fragments of the partition based on which we are creating the sketch. We use \pss{i} to denote a partition with $i$ fragments.
\parttitle{Creating Singleton Sketches}
We considered two approaches for creating singleton sketches for tuples by determining which fragment of a range-partition the tuple belongs to. Either the membership of the tuple is tested using a list of \textit{case} expressions or we use a UDF to perform \textit{binary search} over the ranges of the partition.
We use the crimes dataset and Postgres in this experiment. As expected, \textit{binary search} significantly outperforms \textit{case} for larger number of fragments, e.g., about 2 orders of magnitude for \pss{10K}.

\parttitle{Merging Sketches}:
For operators like aggregation we need to union sketches which corresponds to computing the bitwise-or of sketches. 
Recall that we presented two optimizations for the operation in \Cref{tab:opt_capture}:   \textit{delay} and \textit{No-copy}. 
For this experiment, we union the singleton sketches for all tuples from the movie ratings datasets.\BG{Was not introduced} The results are shown in  \Cref{fig:bitor-opt}.
The \textit{delay} optimization significantly improves performance for larger number of fragments, e.g., from $\sim$0.5 seconds to $\sim$0.2 seconds for \pss{10K}.  \textit{No-copy} further improves this to $\sim$0.16 seconds. 
Thus, we enable these optimizations for all remaining experiments.
}
\iftechreport{\input{./sections/exp_figs_opt_end}}
\subsection{TPC-H}\label{sec:exp-tpc-h}

%

Because of the TPC-H benchmark's artificial data distribution, e.g., many columns are uniformly distributed, this stresses our approach since there are essentially no meaningful correlations that we can exploit.
\XN{no need maybe:
We use range-based provenance sketches since in preliminary experiments these sketches have shown to perform better or at least comparable to other sketch types in terms of skipping potential (size of the over-approximation of provenance) and the selection conditions we produce for such sketches can be exploited successfully by a vanilla database system (e.g., using indexes or zone maps).}
As explained in \Cref{sec:prov-sketch}\BG{make sure to describe this there}, we use the one-dimensional equi-depth histograms maintained as statistics by the DBMS to determine the ranges of a partition for sketches.
For consistency we generate provenance sketches on the primary keys (PK) of relations. However, for cases where using the PK is unsafe (see \Cref{sec:safety-check}), we build sketches over the query's group-by attributes. In addition, we evaluated how the number of fragments of a provenance sketch affect the selectively of the sketch (the fraction of input data covered by the sketch). We vary the number of fragments from 32 to 10000. We show the selectivity for each query and table accessed by a query in~\cite{techreport}. For many queries we already achieve selectivities of a few percent for \pss{4000}.
For queries that are not shown in the following either the provenance is already too large for this queries to benefit from \pbds (e.g., query Q1's provenance consists of over 95\% of its input) or the query's selection conditions are already restrictive leaving no further room for improvement.
\BG{In addition to presenting the results for capture and use of provenance sketches, we also collect the provenance sketch size for each table accessed by the queries to verify our results.}
\XN{duplicated:
We limit the use and capture experiments to queries where provenance sketches can be effective, e.g., if all input data of a query belongs to the provenance then provenance sketches cannot improve performance.} \BG{At last, we study from our result (capture and reuse) to see when we would begin to gain the benefit (recall we need to pay the first run of the query to get the provenance sketch firstly) and which number of fragments we need to create in provenance sketch for different user cases.}

\iftechreport{
\begin{figure}[t]
\begin{adjustbox}{max width=1\linewidth}
\begin{tabular}{c|c|c|c|c||c|c|c|c|}
\cline{2-9}
                                     & \multicolumn{4}{c||}{\textbf{1GB}} & \multicolumn{4}{c|}{\textbf{10GB}}                                                                                                                    \\ \hline
\multicolumn{1}{|c||}{\cthead Query} & \cthead  No-PS                       & \cthead \pss{4000} & \cthead \pss{10000} & \cthead \pss{100000} & \cthead No-PS       & \cthead \pss{4000} & \cthead \pss{10000} & \cthead \pss{100000} \\ \hline
\multicolumn{1}{|c||}{Q2}            & {[}1, 10)                          &                    & {[}10, 73)          & {[}73, $\infty$)     & {[}1, 3)          & {[}3, 6)           & {[}6, 215)          & {[}215, $\infty$)    \\ \hline
\multicolumn{1}{|c||}{Q3}            & {[}1, 2)                           & {[}2, $\infty$)    &                     &                      & {[}1, 2)          & {[}2, 17)          & {[}17, 1629)        & {[}1629, $\infty$)   \\ \hline
\multicolumn{1}{|c||}{Q5}            & {[}1, 3)                           &                    &                     & {[}3, $\infty$)      & \cellcolor{black} & \cellcolor{black}                  & \cellcolor{black}                   & \cellcolor{black}                    \\ \hline
\multicolumn{1}{|c||}{Q7}            & {[}1, 2)                           &                    &                     & {[}2, $\infty$)      & \cellcolor{black}                 & \cellcolor{black}                  & \cellcolor{black}                   & \cellcolor{black}                    \\ \hline
\multicolumn{1}{|c||}{Q8}            & {[}1, 4)                           &                    & {[}4, 5)            & {[}5, $\infty$)      & \cellcolor{black}                 & \cellcolor{black}                  & \cellcolor{black}                   & \cellcolor{black}                    \\ \hline
\multicolumn{1}{|c||}{Q10}           & {[}1, 2)                           & {[}2, 46)          & {[}46, 667)         & {[}667, $\infty$)    & {[}1, 2)          & {[}2, 3)           & {[}3, 462)          & {[}462, $\infty$)    \\ \hline
\multicolumn{1}{|c||}{Q17}           & {[}1, 2)                           &                    & {[}2, 799)          & {[}799, $\infty$)    & \cellcolor{black}                 & \cellcolor{black}                  & \cellcolor{black}                   & \cellcolor{black}                    \\ \hline
\multicolumn{1}{|c||}{Q18}           & {[}1, 2)                           & {[}2, 3)           & {[}3, 1968)         & {[}1968, $\infty$)   & {[}1, 2)          &                    & {[}2, $\infty$)     &                     \\ \hline
\multicolumn{1}{|c||}{Q19}           & {[}1, 3)                           &                    & {[}3, 195)          & {[}195, $\infty$)    & {[}1, 7)          &                    & {[}7, 8)            & {[}8, $\infty$)      \\ \hline
\multicolumn{1}{|c||}{Q20}           & {[}1, 3)                           & {[}3, 11)          & {[}11, 387)         & {[}387, $\infty$)    & {[}1, 5)          & {[}5, 8)           & {[}8, 540)          & {[}540, $\infty$)    \\ \hline
\multicolumn{1}{|c||}{Q21}           & {[}1, 15)                          &                    & {[}15, 507)         & {[}507, $\infty$)    & {[}1, 6)          & {[}6, 26)          & {[}26, $\infty$)    &                     \\ \hline
\end{tabular}
\end{adjustbox}
\trimfigspace
\caption{Optimal \#fragments varying \#repetitions}
\label{fig:option}
\end{figure}
}
\iftechreport{
\parttitle{Provenance Sketch Selectivity}
\Cref{fig:sel-tpch-1gb}  shows the percentage of data covered by provenance sketches for each relation accessed by a query when varying the number of fragments from 32 to 10000. We use colors to denote relations and patterns to denote number of fragments. For consistency we generate provenance sketches on the primary key attributes of a relation. However, for cases where using the PK would be  unsafe (see \Cref{sec:safety-check}) build the sketch over the query's group-by attributes. 
For queries that are not shown here either the provenance is already too large for this queries to benefit from \pbds or the query's selection conditions are already restrictive leaving no further room for improvement. 
For example, query Q1's provenance consists of over 95\% of its input (the lineitem table).
Note that about half of the TPC-H queries are quite selective in terms of provenance and this selectivity can be exploited by sketches even when using only a moderate number of fragments. For many queries we already achieve selectivities of a few percent for \pss{4000}.}

\parttitle{Postgres - Capture \& Reuse}
Next, we evaluate whether these input size reductions lead to significant performance improvements.
\Cref{fig:tpch-post-1gb,fig:tpch-post-10gb} show the runtime of TPC-H queries using captured provenance sketches  (\textit{\ps}) and without  \pbds (\textit{\nor}).
For this experiment we created zone maps for all tables.\footnote{Zone maps are called brin indexes in Postgres.} Furthermore, we create indexes on 
PK and FK columns.
Thus, the database system has plenty of index structures to choose from when evaluating queries.
Unless stated otherwise, queries apply the binary search (\textit{BS}) method to test whether a tuple belongs to a provenance sketch 
(\Cref{sec:ps-reuse}). 



\iftechreport{\Cref{fig:tpch-post-1gb} shows runtimes for SF1.  Q3 is a top-10 query that returns the 10 orders with the highest revenue. It is highly selective on the \pk of orders and customer (at most 10 customers have submitted these orders).
}
\ifnottechreport{
\Cref{fig:tpch-post-1gb} shows runtimes for SF1. Q3 is a top-10 query that returns the 10 orders with the highest revenue. It is highly selective on the \pk of orders and customer.
}
Since we 
use equi-depth histograms to determine partition ranges, each fragment contains approximately the same number of rows.
Thus, the runtime of the query is roughly linear in the number of rows contained in the 10 fragments of the provenance sketch, e.g., $\sim\frac{1}{40}$ the runtime without \pbds for \pss{400}. 
We observe similar behavior for queries Q10 and Q18 which are top-20 and top-100 queries, respectively.
Q19 is an interesting case since it consists of an aggregation over a complex selection condition. This demonstrates that \pbds can sometimes unearth additional ways to exploit selection conditions that the DBMS was unable to detect.
\iftechreport{As shown in \Cref{fig:sel-tpch-1gb}, for queries Q5, Q7, Q8, Q20 and Q21 we need larger numbers of fragments to be able to benefit from \pbds. This trend is also reflected in query performance \Cref{fig:tpch-post-1gb}.}
\ifnottechreport{For queries Q5, Q7, Q8, Q20 and Q21 we need larger numbers of fragments to be able to benefit from \pbds we measured in sketch selectivity. This trend is also reflected in query performance.}
While queries Q2 and Q17 are selective in terms of relevance, their selection conditions are already quite restrictive leaving little room for further improvement.
\Cref{fig:tpch-post-10gb} shows runtimes for SF10. Observe that the runtime of queries Q2, Q3, Q10, Q20 and Q21 show similar behavior for SF10 (10GB) as for SF1.



\BG{The runtime of the query is roughly linear in the data size of 10 fragments, e.g., $\sim\frac{1}{40}$ of the runtime without \pbds ~for \pss{400} (a partitioning with 400 fragments).
We observe similar behavior for queries Q10 and Q18 which are top-20 and top-100 queries respectively. However, Q18 only returns 10 result tuples when running on TPC-H 1GB  dataset even though it is a top-100 query, this is why even \pss{32} is also selective (See \Cref{fig:sel-tpch-1gb}). Q19 has the similar trend, however, the selectivity of \pss{32} is almost equal to 1, which confirms that the runtime of Q19 on \pss{32} is same with the runtime without \pbds (\nor) in \Cref{fig:tpch-post-1gb}.
Q5, Q7, Q8, Q20 and Q21 have similar selectivity trends, all of them are only selective when the number of fragments larger or equal to 4000 and the performance in \Cref{fig:tpch-post-1gb} confirms it.
Since the selection conditions of Q2 are already selective enough, even though Q2 is selective, only the provenance sketch with high selectivity could improve the query performance, e.g., \pss{10000} improves performance by 13\% which costs 0.46 seconds.  Q17 has similar trends with Q2.}
\iftechreport{
Binary search is typically more efficient when the number of fragments in the provenance sketch is large. 
However, for very selective provenance sketches, using a B-tree index will be more efficient.
\Cref{fig:use-tpch-post-1gb-or} shows the runtime of queries with selective provenance sketches over a SF1 instance when translating the sketch into a disjunctive condition (\textit{OR}). We only show queries whose runtime is improved compared to using binary search. The most significant improvements are for Q15 which did not benefit from \pbds for binary search and Q2 whose runtime is reduced to 0.1 seconds for \pss{400} and 0.066 seconds for \pss{10000}.}

\BG{For the same query Q2,  \pss{400} costs 0.1 seconds and \ps 10000 costs 0.066 seconds since using index is quicker than the query's selection conditions. Similarly the performance of Q3, Q10, Q18 and Q19 was improved again and also Q15 even thought it can not be sped up when using binary search. However, Q10 on \pss{10000} costs 0.075 seconds higher than the cost 0.054 seconds using binary search since large number of \lstinline!OR! clauses would be expensive. 
Q5, Q7, Q8, Q17, Q20 and Q21 in \Cref{fig:tpch-post-1gb} which are not shown in \Cref{fig:use-tpch-post-1gb-or}, since they can only get benefit from when the number of fragments larger or equal to 4000 which has large \ps ~size which would result in large number of \lstinline!OR! clause to make the query cost high.}

\Cref{fig:cap-tpch-post-1gb,fig:cap-tpch-post-10gb} show the overhead of capturing provenance sketches relative to executing the queries without any instrumentation for SF1 and SF10. Note that for some queries the overhead is less than 20\% while it is always less than 100\% for partition sizes up to 10000 fragment.  The overhead increases slightly in the number of fragments 
since larger number of fragments result in larger bitvectors 
and require more computation per input tuple to determine which fragment a tuple belongs to. Here we benefit from using binary search instead of a linear sequence of \lstinline!CASE! expressions.
For 100000 fragments the overhead is typically between 20\% and 700\% with an outlier (Q20) for the 10GB database which has $\sim$7500\% overhead.
\ifnottechreport{
Furthermore, we analyze whether the overhead of capture can be amortized by using
provenance sketches. We show the result in~\cite{techreport} that the overhead of
capturing a sketch is often amortized by using the sketch only once or
twice. }
\iftechreport{
\parttitle{Amortizing Capture Cost}
We now analyze whether the overhead of capture can be amortized by using
provenance sketches.  \Cref{fig:option} shows for each TPC-H query and a given
number of repetitions ($\runs$) of this query, the partition size (if any) that
minimizes total query execution cost. Note that these numbers are for
Postgres. For each partition size (and no partitioning) we show the interval of
repetitions for which this option is optimal. For example, consider query Q10
and the 1GB instance, [1,2) in column \nor means that if we only need to run the
query once, then the optimal choice is to not create any provenance sketch.  Let
$C_{\nor}$, $C_{cap}$ and $C_{reuse}$ represent the cost of running the query,
capturing the provenance sketch for the query, and running the instrumented
query which uses the sketch, respectively. Then the cost of evaluating the query
$runs$ times without \pbds is $C_{\nor} * \runs$. When using a sketch we have to
create the sketch and then evaluate the query $\runs$ times using the sketch:
$C_{cap} + C_{use} * \runs$. Based on these formulas we can determine which
option is optimal for $\runs$ repetitions. We do not show \pss{32}, \pss{64},
and \pss{400} since these options are dominated by other options for all values
of $\runs$.  Note that use of provenance sketches for \pbds often results in
performance improvements of several orders of magnitude. Thus, the overhead of
capturing a sketch is often amortized by using the sketch only once or
twice. Cells which are blacked out are queries for which \pbds is not beneficial
for this dataset size.}

\parttitle{MonetDB}
We also run \pbds on MonetDB to test our approach on an operator-at-a-time columnar main-memory system without indexes that is optimized for minimizing cost per tuple.
\Cref{fig:monetdb-tpch-1gb,fig:monetdb-tpch-10gb} show the runtime for using sketches.
Even though we cannot exploit any physical design, there are still several queries for which \pbds is beneficial. However, for 1GB the overhead of evaluating \lstinline!WHERE! clause conditions can outweight the benefits of reducing data size for larger number of fragments for queries Q2 and Q10.
%
%
\Cref{fig:monetdb-tpch-overhead-1gb,fig:monetdb-tpch-overhead-10gb} shows the relative overhead of provenance capture which exhibits similar trends as for Postgres but is usually higher.
We omit \pss{100000} since it did not result in additional improvement.
\iftechreport{
\subsection{Real World Datasets}
}
\ifnottechreport{
\subsection{Stack Overflow Dataset}
}
\label{sec:real-world-datasets}

\BG{Maybe the discussion on combining attributes has to go}
\iftechreport{
\parttitle{Crimes}
C-Q1 is a top-5 query grouping on 
geographical attributes while C-Q2 is two-level aggregation where the inner aggregation groups on geographical attribute \textit{block}. We consider provenance sketches over a combination of all group-by attributes (denoted as \ps{MIX}). 
Since these are strongly correlated geographical attributes with a low number of distinct values, 
i.e., a one-to-one correspondence between fragments and group-by values. We show the runtime of \pbds and capture in \Cref{fig:realdb-use,fig:realdb-cap}, respectively. \pbds improves performance by 88.5\% for C-Q1 and 30.3\% for C-Q2. The capture overhead is larger in this experiment than for the TPC-H queries since these queries do not use any selection conditions and thus a singleton sketch has to be produced for every input row.}
\BG{Selectivity of provenance is $\sim$2\%}

\iftechreport{
\parttitle{Movies}
Similar to the crime dataset we build provenance sketches over the group-by attributes.
The main difference is that the number of distinct values in the group-by attribute (\texttt{movieid}) is quite large. 
The runtime of M-Q1, M-Q2 and M-Q3 is improved by 61\% , 72.2\%  and  35\%, respectively for \pss{10000}. The capture overhead ranges between a factor of $\sim 0.37$ and $\sim 3.08$ (none of these queries contain any selection conditions).
}

\iftechreport{
\parttitle{Stack Overflow (SOF)}
Compared with crimes and movies dataset, this is a very large dataset, thus we only consider 1000 and 10000 fragments. }
\ifnottechreport{
Since this is a large dataset, we only consider 1000 and 10000 fragments.
}
\Cref{fig:stack-use} shows that \pbds is quite effective 
improving query performance by 96.9\% to 98.85\% for \pss{10000}. The capture overhead ranges between a factor of $\sim -0.14$ and $\sim 1.2$ (\Cref{fig:stack-cap}).
The negative overhead is caused by Postgres uses parallel pre-aggregation for the capture query, but not for the query without capture. We show results for additional real world datasets in~\cite{techreport}.

\ifnottechreport{
\input{./sections/exp_figs_end_sfo}
}

\subsection{End-to-end Experiment}\label{sec:self-tuning}

We now evaluate \pbds in a self-tuning setting over workloads that consist of multiple instances of one or more parameterized queries. We designed strategies to decide for each incoming query whether we will capture a sketch, use a previously captured sketch, or just execute the query without any instrumentation. We use our safety tests to determine which attributes are safe for a parameterized query and the method described in~\Cref{sec:reuse-different} to determine whether one of the sketches we have captured can be used to answer an incoming query.
\ifnottechreport{Given a set of query templates, we generate workloads by randomly choosing for each query the template and parameter values.}
\iftechreport{
We generate each template-based query instance based on following steps:
1) rolling a dice to decide which template to use; 2) random choosing parameter values (the values satisfy normal distribution per template);
 3) for interval paramaters, e.g., $ a>\$1$ and $a<\$2$, we do step 2) two times, one for start point ($\$1$), one for the interval size $n$, then $\$2=\$1+n$.}
We evaluate the performance of our strategies 
varying query selectivity ($sel$) and normal distribution coefficients.
The templates used here are modified from the queries introduced in \Cref{sec:dataset} by changing \lstinline!LIMIT! to a \lstinline!HAVING!.
Furthermore, we measure the cost for determining safe attributes for creating provenance sketches and the cost of checking reusability of provenance sketches.
\ifnottechreport{\parttitle{Self-tuning Strategy}}
\iftechreport{\parttitle{Self-tuning Strategies}}
\iftechreport{\emph{Eager strategy}: } 
\ifnottechreport{We 
keep track of sketches we have captured 
by mapping pairs of parameterized queries and parameter bindings to the sketches we have created.  for these queries and parameter bindings.} 
\iftechreport{We 
keep track of sketches we have captured 
using a map between values of the query parameters and sketches.}
Since sketches are not effective for non-selective queries, we estimate the selectivity of queries and if it is above a threshold (75\% in our experiments), we execute the query without using a sketch. Otherwise, we check whether any of the sketches we have captured so far can be used. If this is the case, we instrument the query to use this sketch. If no such sketch exists, then we record what sketch could have been used for the query. To avoid paying overhead for sketches that are rarely used, we only create a new sketch once we have accumulated enough evidence that the sketch is needed (the number of times it could have been used is above a threshold).
We call this the \emph{adaptive strategy}. \ifnottechreport{In \cite{techreport} we also evaluate an \emph{eager strategy} that creates new sketches whenever a query cannot use any of the existing sketches.}


\BG{For a 
query with parameters  $[l,h)$, 
if the query is not selective (75\% or more of the data),  
then we execute the query without using a sketch. 
 Otherwise, we find the smallest interval $[l',h')$ that contains $[l,h)$ and for which we have captured a provenance sketch in the past. We then apply \pbds using this sketch. If no such sketch exists, then we instrument the query to capture a sketch for $[l,h)$. However, 
less often used skeches might result in unnecessary cost, 
we could delay the capture 
once we encountered the interval enough times (\emph{adaptive strategy}).
\iftechreport{We evaluate  \emph{eager} strategy over crimes dataset and \emph{adaptive} strategy over stack overflow dataset. }}

\iftechreport{
\parttitle{Crimes}
In ~\Cref{fig:ete-crimes-mix}, we mix four templates and each one contains up to five parameters. 
Recall \emph{eager} method creates sketch if no existing sketch could be used which might spend much time in the early stage to creating large number of sketches, thus
we start to gain from the $133_{th}$ queries. 
However, these costs would be amortized with following queries, 
e.g., 50.8\% performance improvement until the $3000_{th}$ query.
Furthermore, we learn the influence by varying the query selectivity. 
In \Cref{fig:ete-crimes-sel}, we generate different selectivity queries from single template.
As we imagined that high selectivity results in a lower performance improvement, however, we still improve performance by 60\% at 20\% selectivity. }

\iftechreport{
\parttitle{Stack Overflow (SOF)}}
\ifnottechreport{
\parttitle{Stack Overflow}
}
 \Cref{fig:ete-stack-mix} shows the result of a workload over the stack overflow dataset with three query templates. We use paramters such that queries have an average selectivity of $\sim$1\%. 
 Since \emph{adaptive} delays capturing sketches and we pay for creating sketches,  \emph{adaptive} and execution without \pbds (\emph{No-PS}) perform roughly the same. However, after $19$ query instances have been executed, \emph{adaptive} starts to accumulate benefits from using sketches. These benefits increase over time while more and more sketches become available for reuse. \emph{adaptive} ourperforms \emph{No-PS} by a factor of $\sim 3$ with respect to total workload execution time (800 queries).
We also evaluated how query selectivity affects performance. 
\Cref{fig:ete-stack-sel07}, ~\ref{fig:ete-stack-sel2} and~\ref{fig:ete-stack-sel5}, show results varying the average  query selectivity (0.7\%, 2\% and 5\%). For this experiment we use a single query template. 
\BG{What template?}
As expected we benefit less for higher selectivites.
Finally, we vary the standard deviation (\emph{SDV}) of the normal distribution we use to determine parameter values (1000 and 5000) and fix selectivity to 1\%. As expected (\Cref{fig:ete-stack-q3-s1k,fig:ete-stack-q3-s5k}), we accumulate benefits faster when parameter values are more clustered (\emph{SDV}=1000), because a smaller number of sketches is sufficient for covering most queries.


\parttitle{Safety and Reuse Check Overhead}
We separately measured the overhead  of safety and reuse checks (both are $\sim 20$ ms per check).
\emph{Safety checks}: If there are $n$ sets of columns we want to check is $n$, then the total cost is $0.02 * n$ seconds. Since we only need to evaluate safety one time per query template, this cost is negligible.
\emph{Reusability check}:
Given $k$ templates and 
$m$ sketches for the each template, we have to test which template a query corresponds to. This takes about 0.05 ms per template. Then we check for each sketch we have created for the query's template whether it can be used to answer the query. Thus, the total time requires to find a sketch to use is
$k \cdot 0.00005 + 0.02 \cdot m$ seconds. 

%% file: sections/exp_figs_short3.tex
\begin{figure*}
  \centering
  \begin{minipage}{1.0\linewidth}
    \centering
    \includegraphics[width=0.8\linewidth]{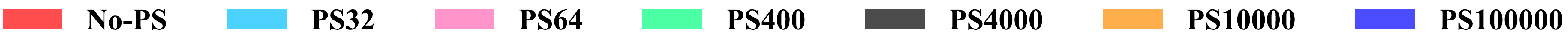}
  \end{minipage}

  \begin{minipage}[t]{1.0\linewidth}
    \begin{subfigure}{0.49\linewidth}
      \centering
      \includegraphics[height=3cm,trim=25pt 25pt 0 24pt, clip]{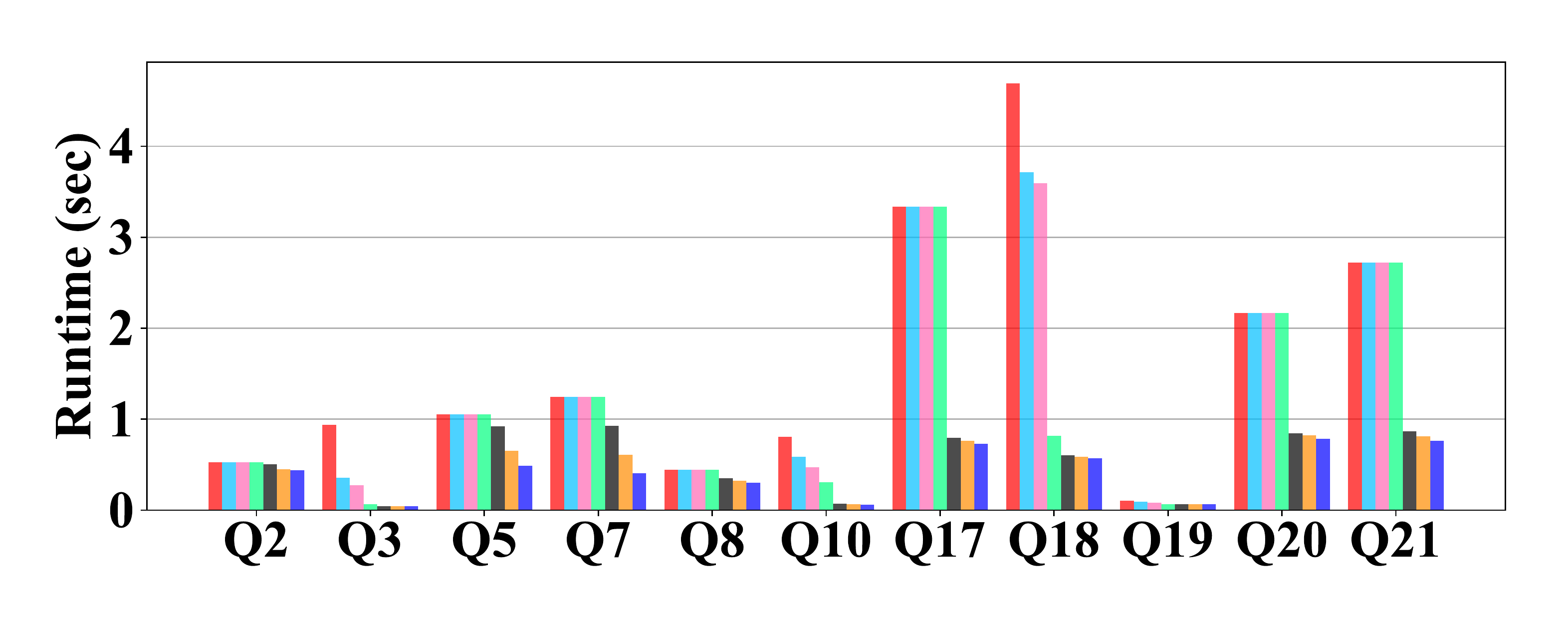}
      \trimfigspace
      \caption{Postgres Use - 1GB (BS)}
      \label{fig:tpch-post-1gb}
    \end{subfigure}
    \begin{subfigure}{0.49\linewidth}
      \centering
      \includegraphics[height=3cm,trim=25pt 25pt 0 24pt, clip]{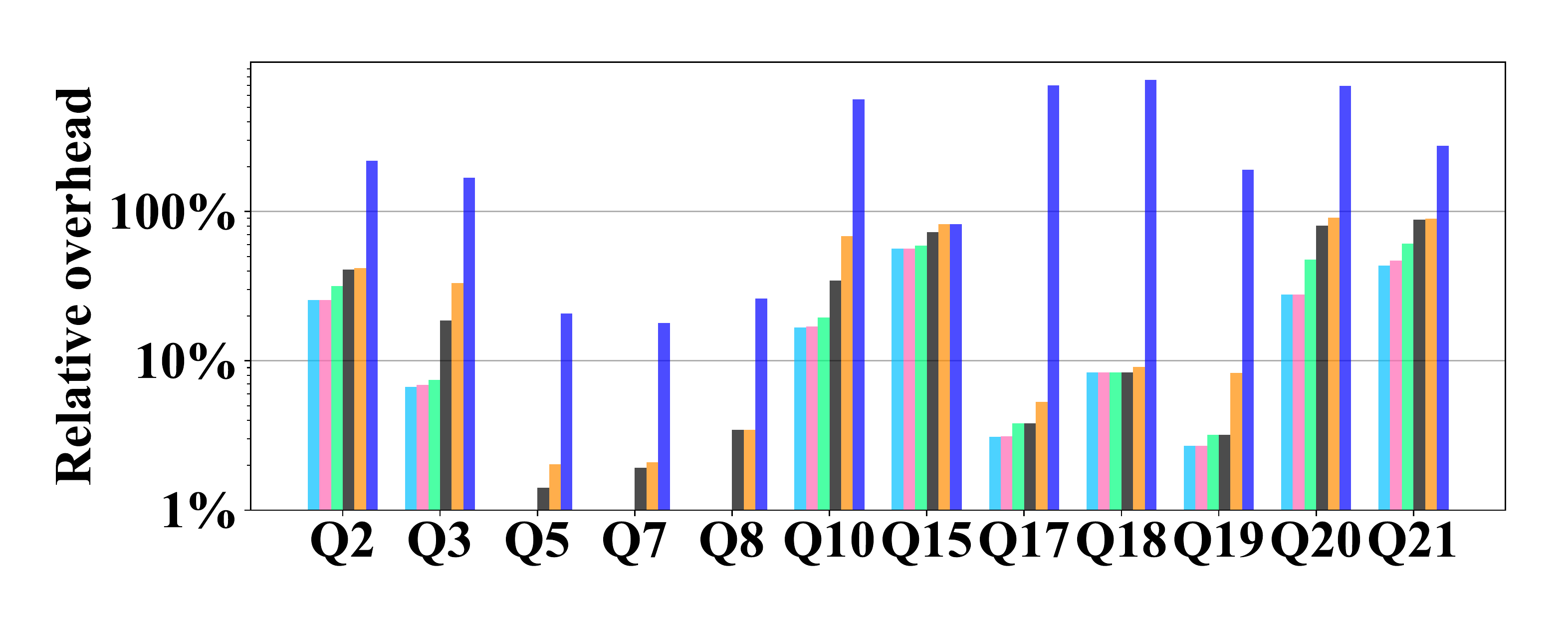}
      \trimfigspace
      \caption{Postgres Capture - 1GB}
      \label{fig:cap-tpch-post-1gb}
    \end{subfigure}
  \end{minipage}

  \begin{minipage}{1.0\linewidth}
    \begin{subfigure}{0.235\linewidth}
     \begin{adjustbox}{minipage=\linewidth,scale=1.05}
      \centering
      \includegraphics[height=2.55cm,width=1\linewidth,trim=25pt 25pt 0pt 25pt, clip]{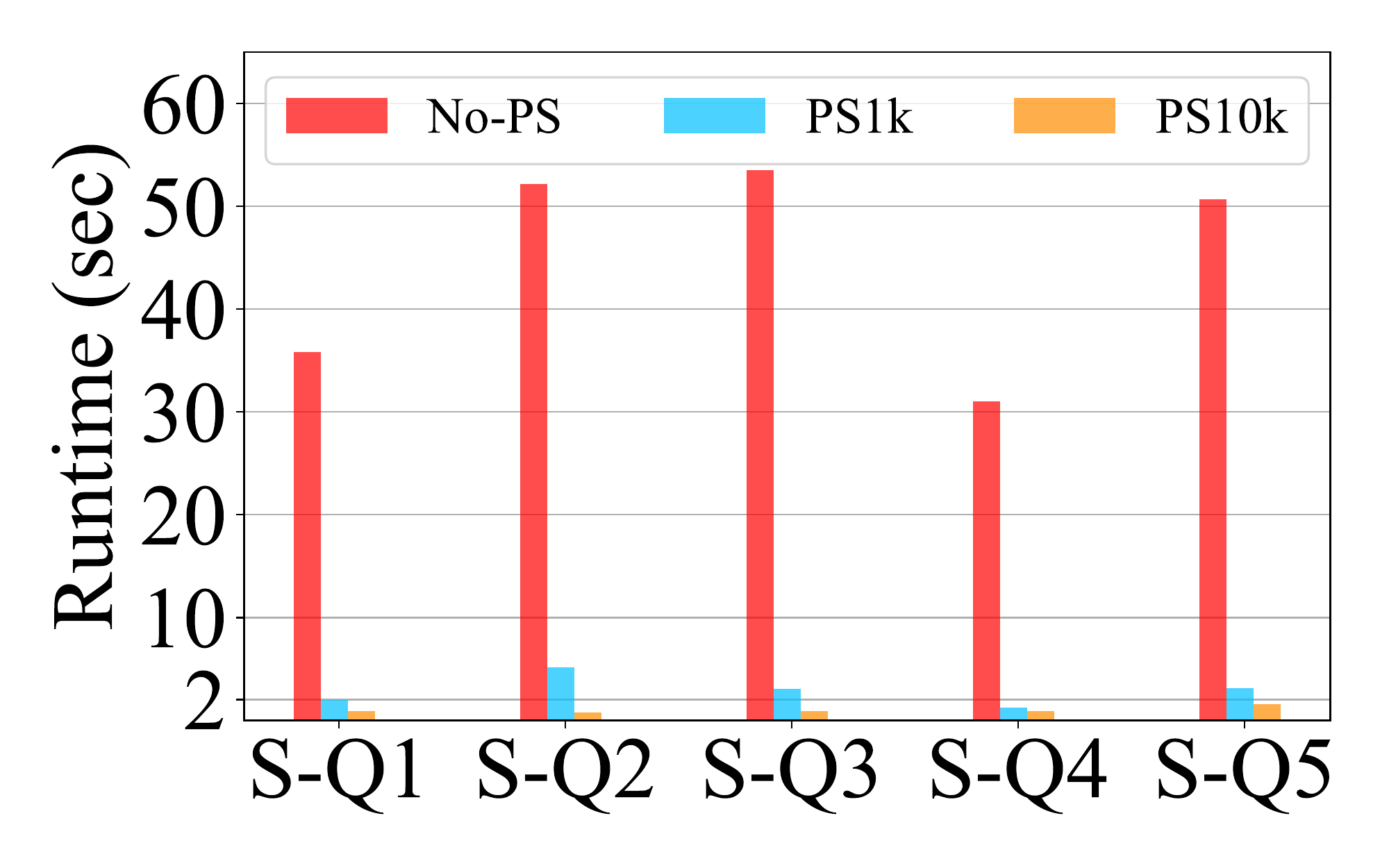}
      \trimfigspace \trimfigspace
            \end{adjustbox}
      \caption{Stack Overflow - Use (BS)}
      \label{fig:stack-use}
    \end{subfigure}
    \begin{subfigure}{0.24\linewidth}
    \begin{adjustbox}{minipage=\linewidth,scale=1.05}
      \centering
      \includegraphics[height=2.55cm,width=1\linewidth,trim=25pt 26pt 0pt 21pt, clip]{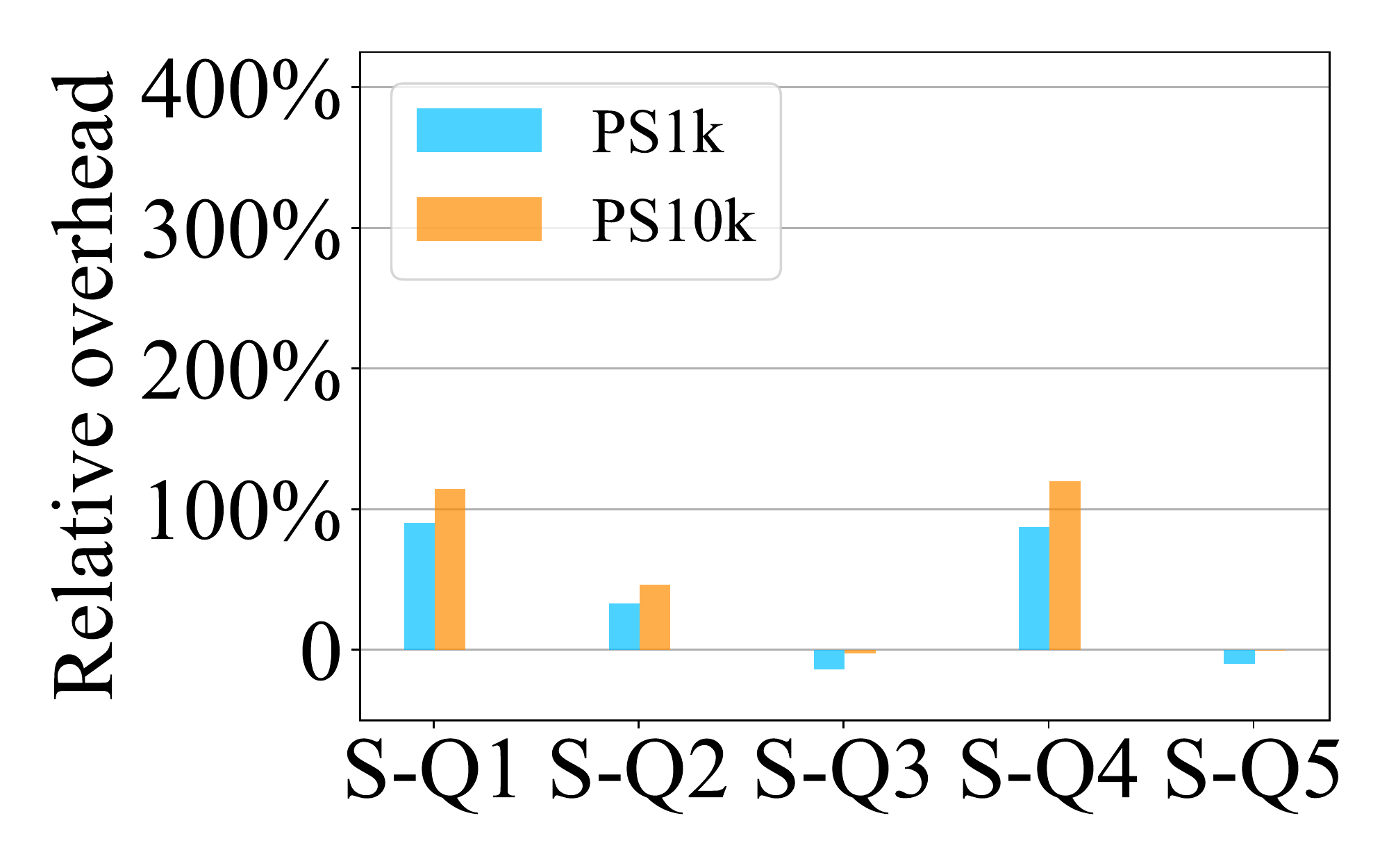}
      \trimfigspace \trimfigspace
             \end{adjustbox}
      \caption{Stack Overflow - Capture}
      \label{fig:stack-cap}
    \end{subfigure}
    \begin{subfigure}{0.28\linewidth}
      \begin{adjustbox}{minipage=\linewidth,scale=0.95}
      \includegraphics[height=2.8cm,trim=25pt 32pt 21pt 35pt, clip]{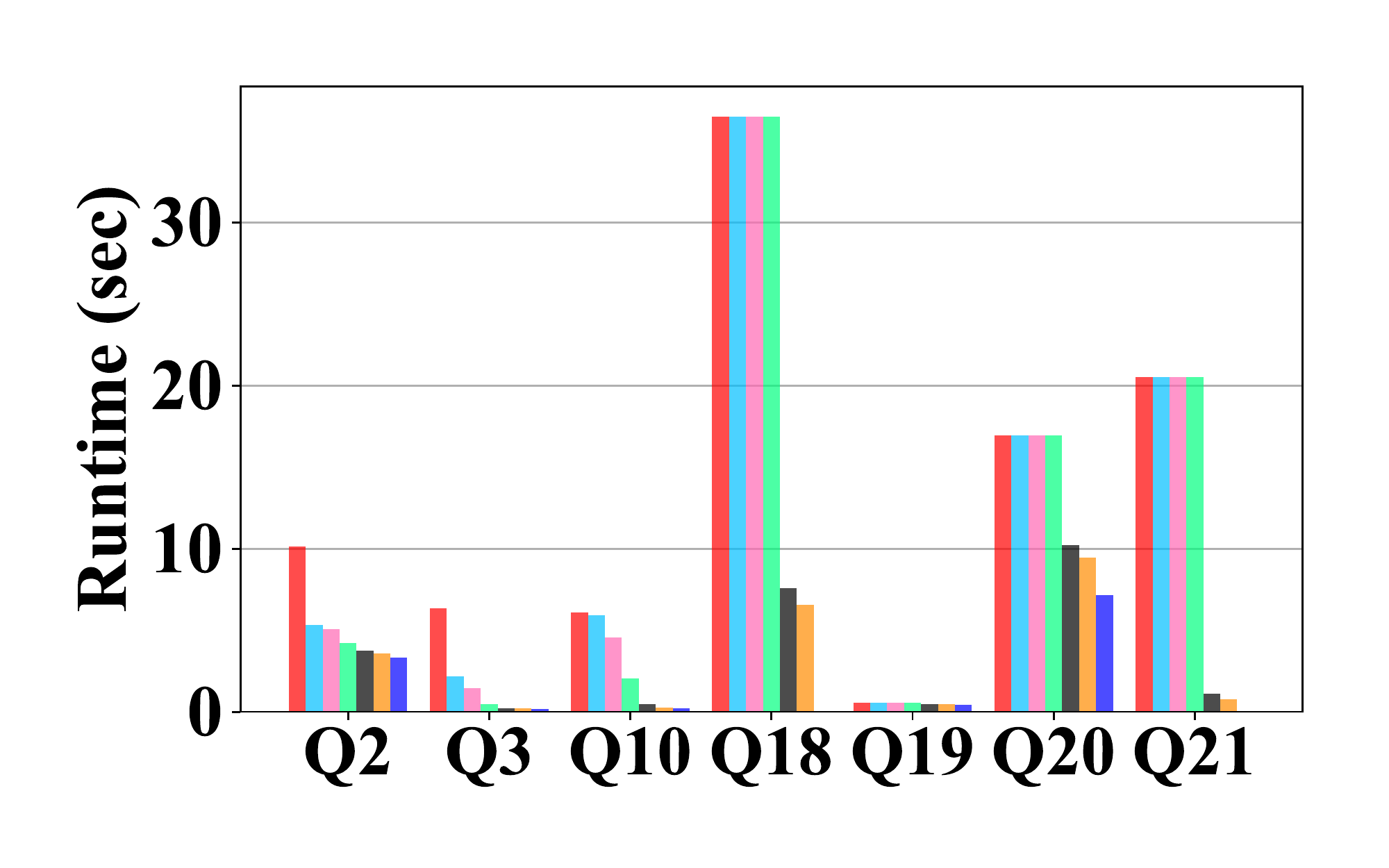}
      \trimfigspace \trimfigspace
                \end{adjustbox}
      \caption{Postgres Use - 10GB (BS)}
      \label{fig:tpch-post-10gb}
    \end{subfigure}\hspace{-8pt}
    \begin{subfigure}{0.23\linewidth}
    \begin{adjustbox}{minipage=\linewidth,scale=0.95}
      \includegraphics[height=2.8cm,trim=32pt 32pt 21pt 20pt, clip]{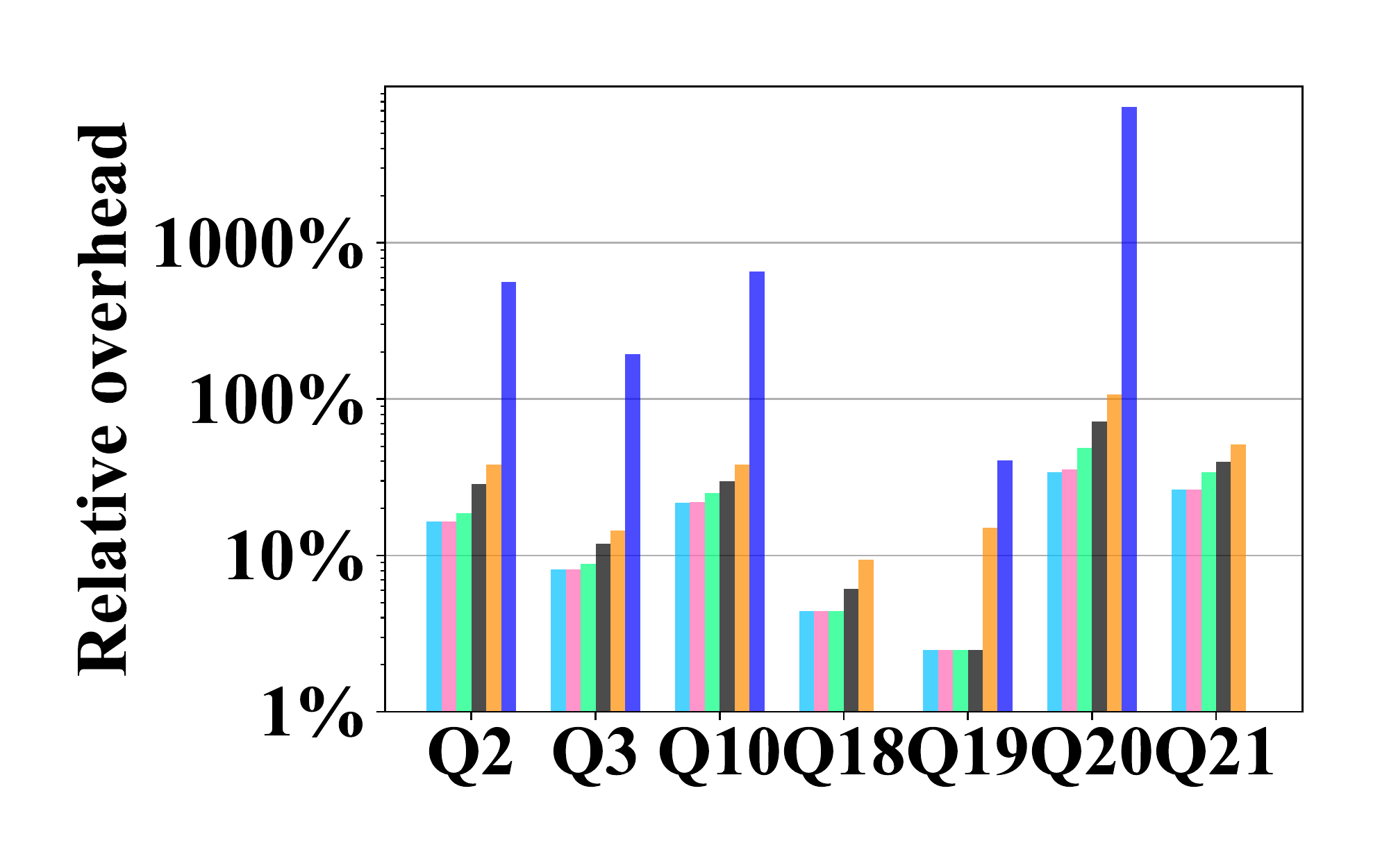}
      \trimfigspace \trimfigspace
             \end{adjustbox}
      \caption{Postgres Capture - 10GB}
      \label{fig:cap-tpch-post-10gb}
    \end{subfigure}
  \end{minipage}

  \begin{minipage}{1.0\linewidth}
    \begin{subfigure}{0.245\linewidth}
      \includegraphics[width=1\linewidth,trim=20pt 25pt 0 0pt, clip]{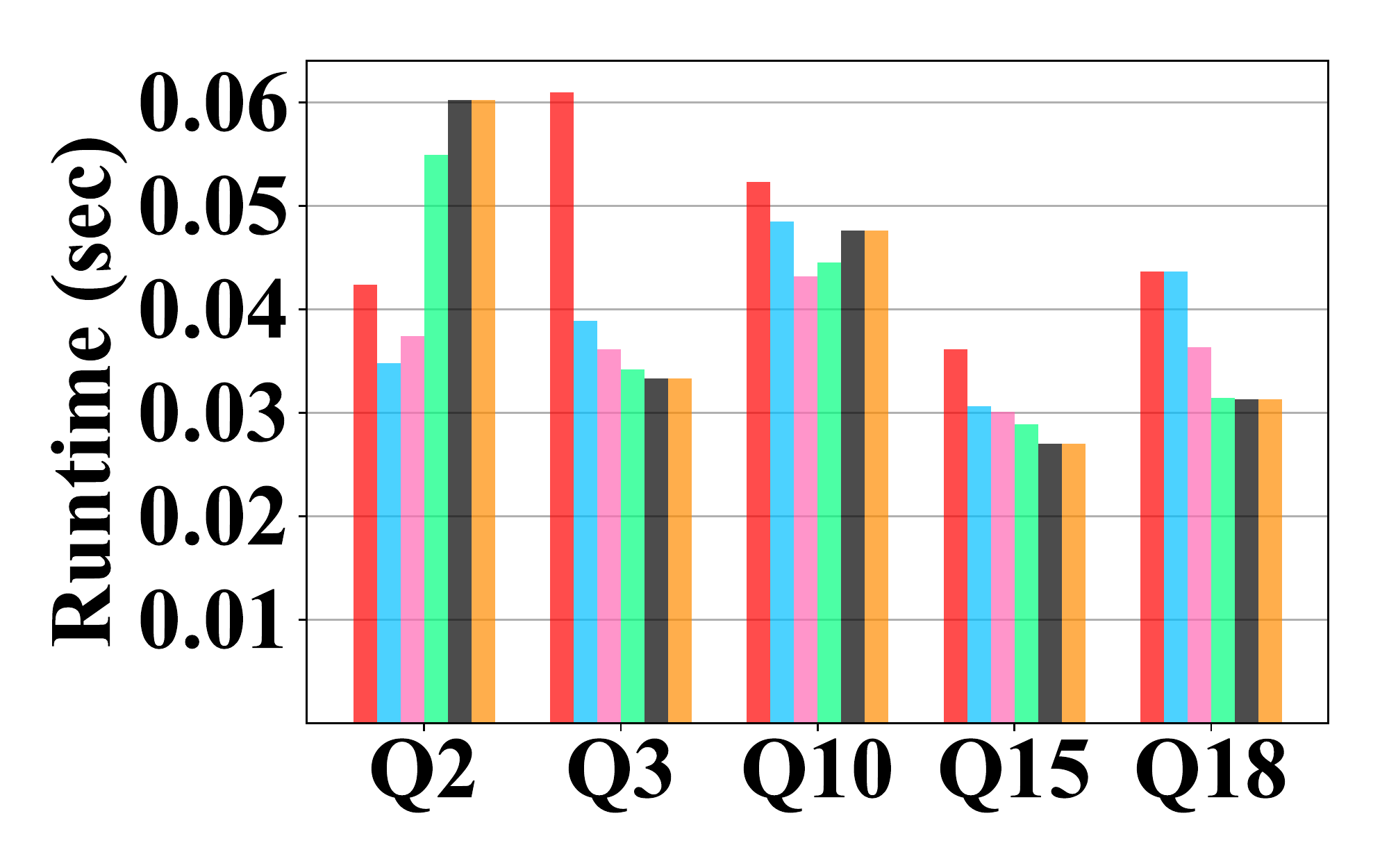}
      \trimfigspace[-5mm]
      \caption{MonetDB Use - 1GB (OR)}
      \label{fig:monetdb-tpch-1gb}
    \end{subfigure}
    \begin{subfigure}{0.245\linewidth}
      \includegraphics[width=1\linewidth,trim=20pt 25pt 0 0pt, clip]{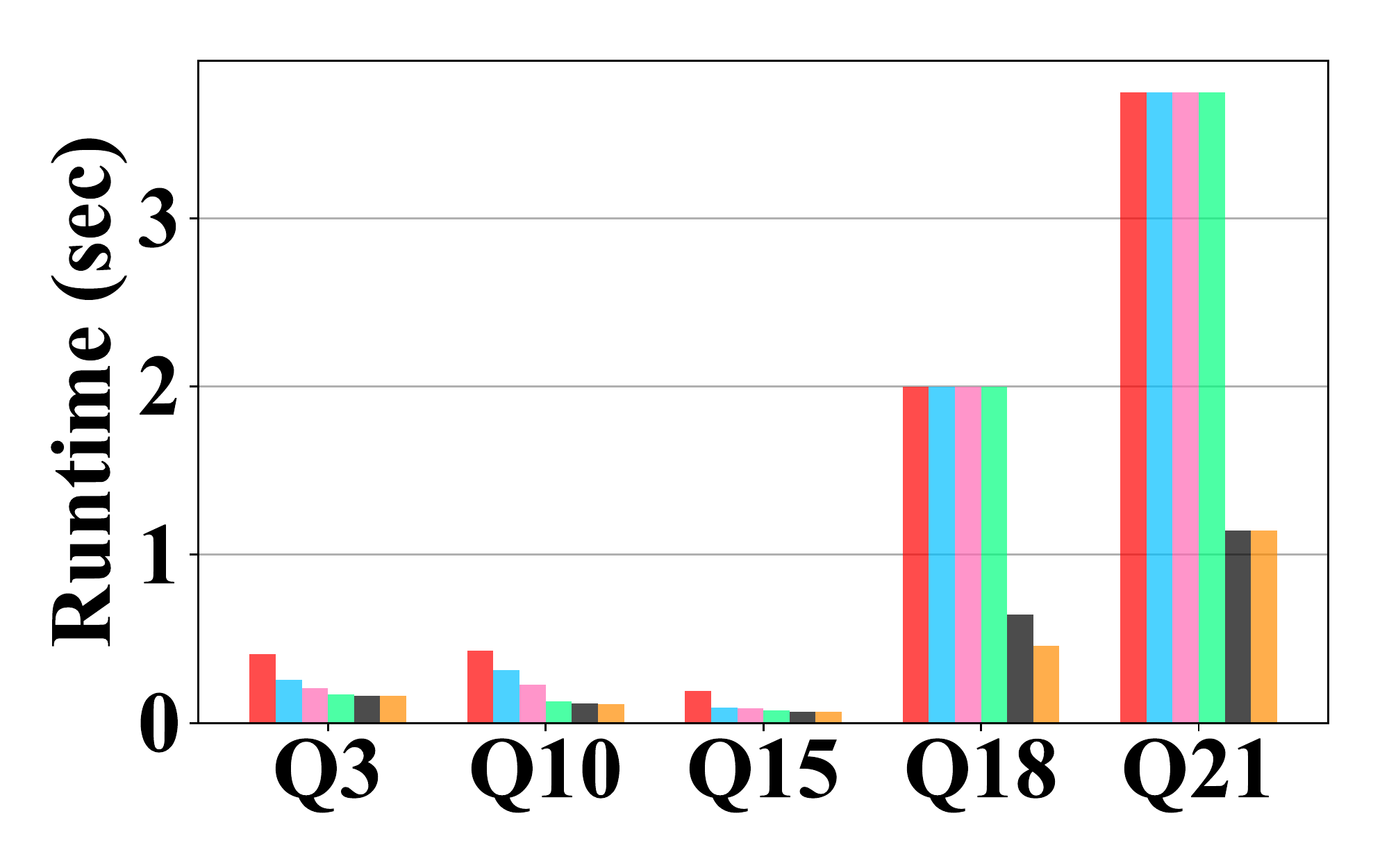}
      \trimfigspace[-5mm]
      \caption{MonetDB Use - 10GB (OR)}
      \label{fig:monetdb-tpch-10gb}
    \end{subfigure}
    \begin{subfigure}{0.245\linewidth}
      \includegraphics[width=1\linewidth,trim=20pt 25pt 0 0pt, clip]{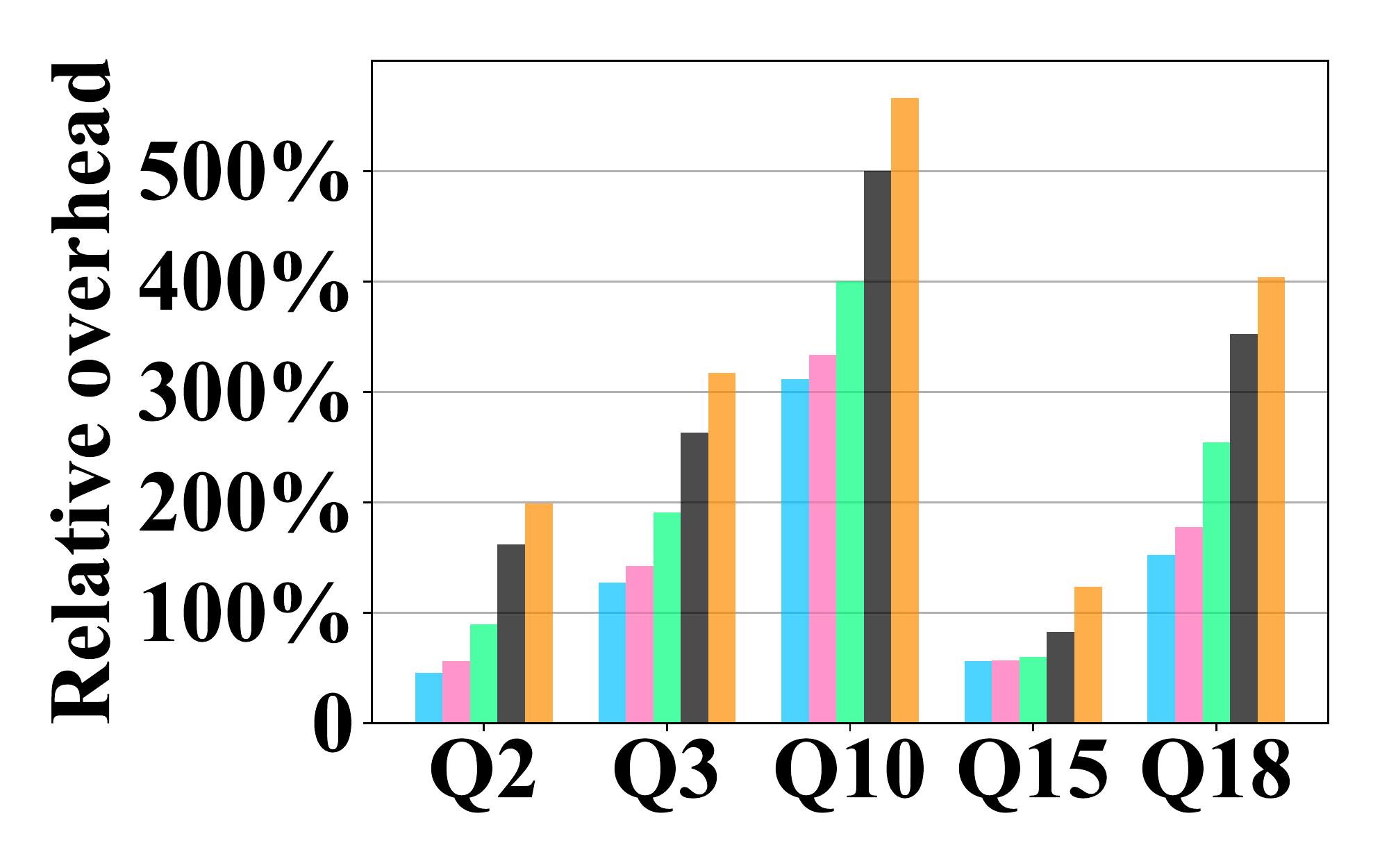}
      \trimfigspace[-5mm]
      \caption{MonetDB Capture - 1GB}
      \label{fig:monetdb-tpch-overhead-1gb}
    \end{subfigure}
    \begin{subfigure}{0.245\linewidth}
      \includegraphics[width=1\linewidth,trim=20pt 25pt 0 0pt, clip]{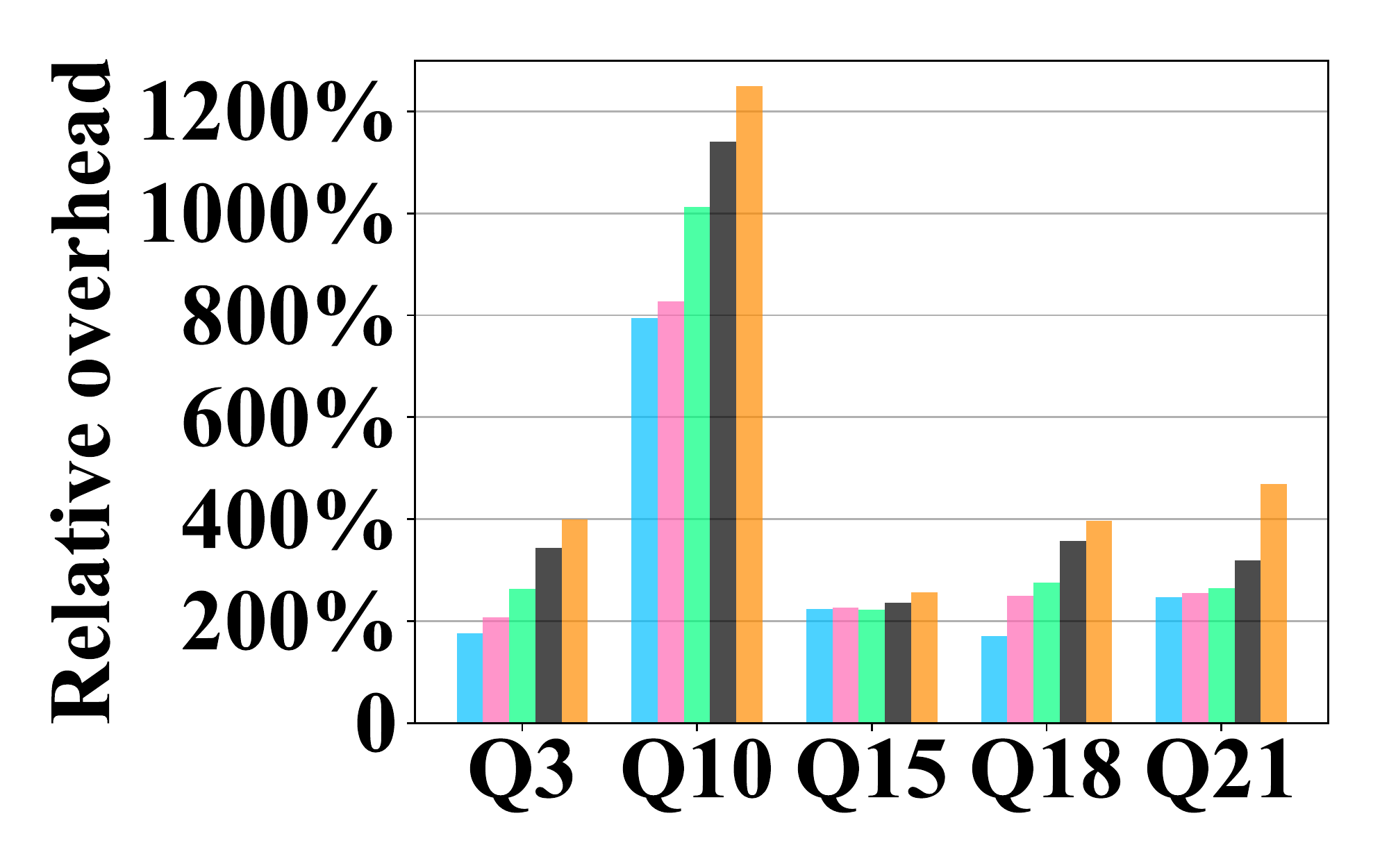}
      \trimfigspace[-5mm]
      \caption{MonetDB Capture - 10GB}
      \label{fig:monetdb-tpch-overhead-10gb}
    \end{subfigure}
    \trimfigspace
    \caption{Performance of provenance sketch capture and use for TPC-H and Stack Overflow queries.}
    \label{fig-exp-runtimes-usage}
  \end{minipage}
\end{figure*}

%% file: sections/exp_figs_long.tex
\begin{figure*}
  \centering
  \begin{minipage}{0.48\linewidth}
    \includegraphics[width=1\linewidth,trim=0pt 5pt 0 0pt, clip]{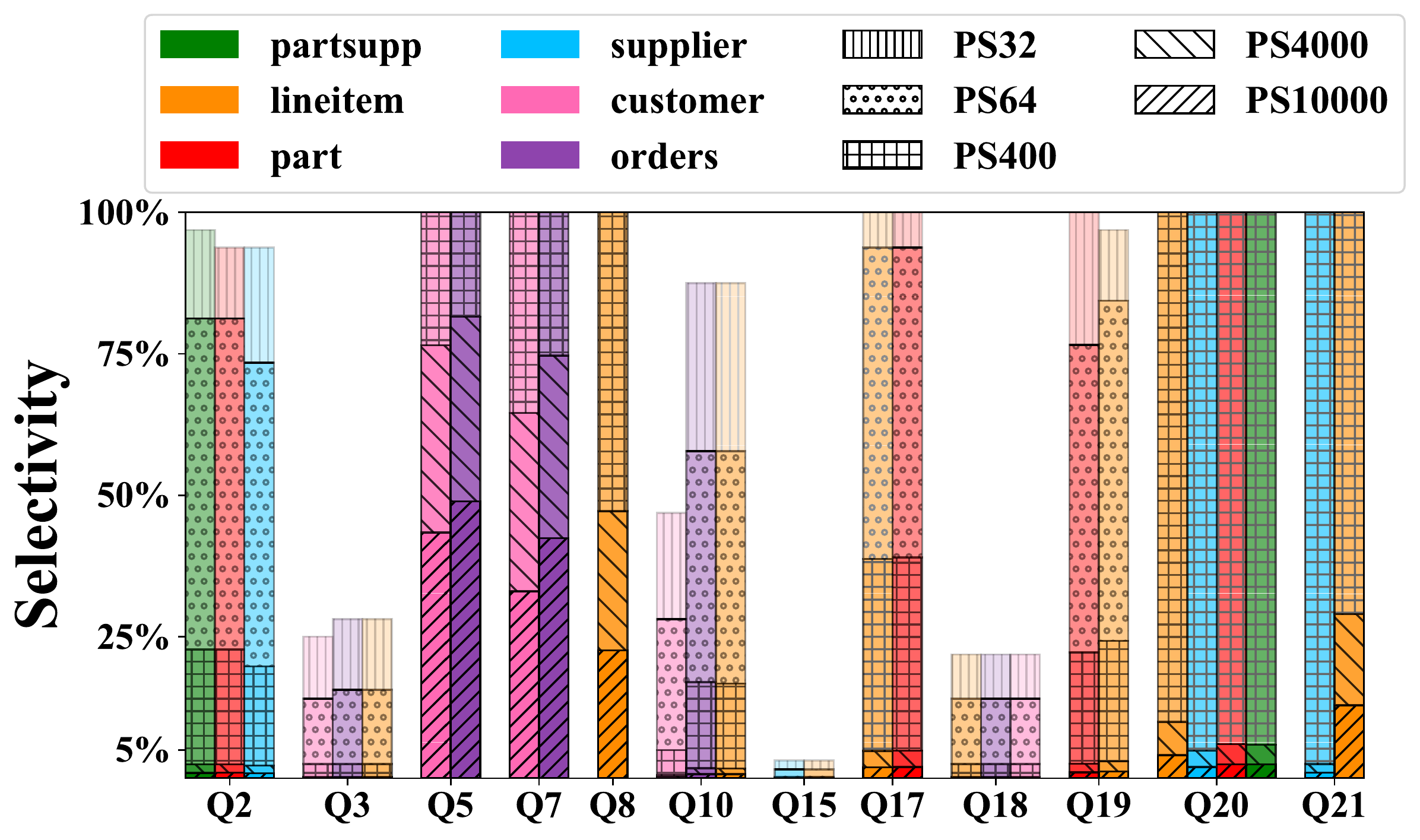}
    \trimfigspace
      	\caption{Selectivity of range-partition provenance sketches for the TPC-H 1GB.}
      	\label{fig:sel-tpch-1gb}
      \end{minipage}
        \begin{minipage}{0.46\linewidth}
    \begin{subfigure}{0.49\linewidth}
      \centering
      \includegraphics[width=1\linewidth,trim=0pt 25pt 0pt 25pt, clip]{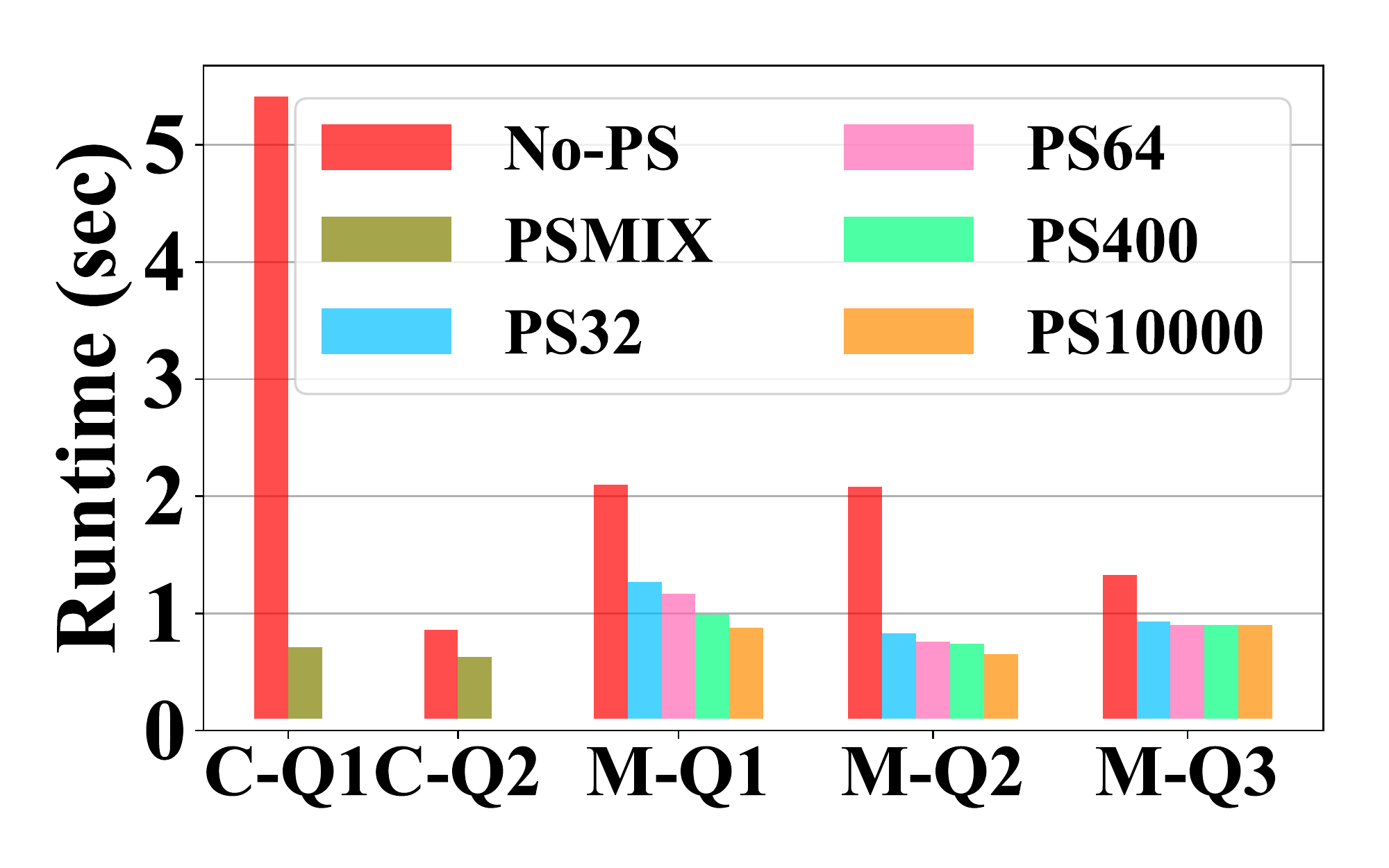}
      \caption{Crimes and Movies - Use (OR)}
      \label{fig:realdb-use}
    \end{subfigure}
    \begin{subfigure}{0.49\linewidth}
      \centering
      \includegraphics[width=1\linewidth,trim=0pt 25pt 0pt 25pt, clip]{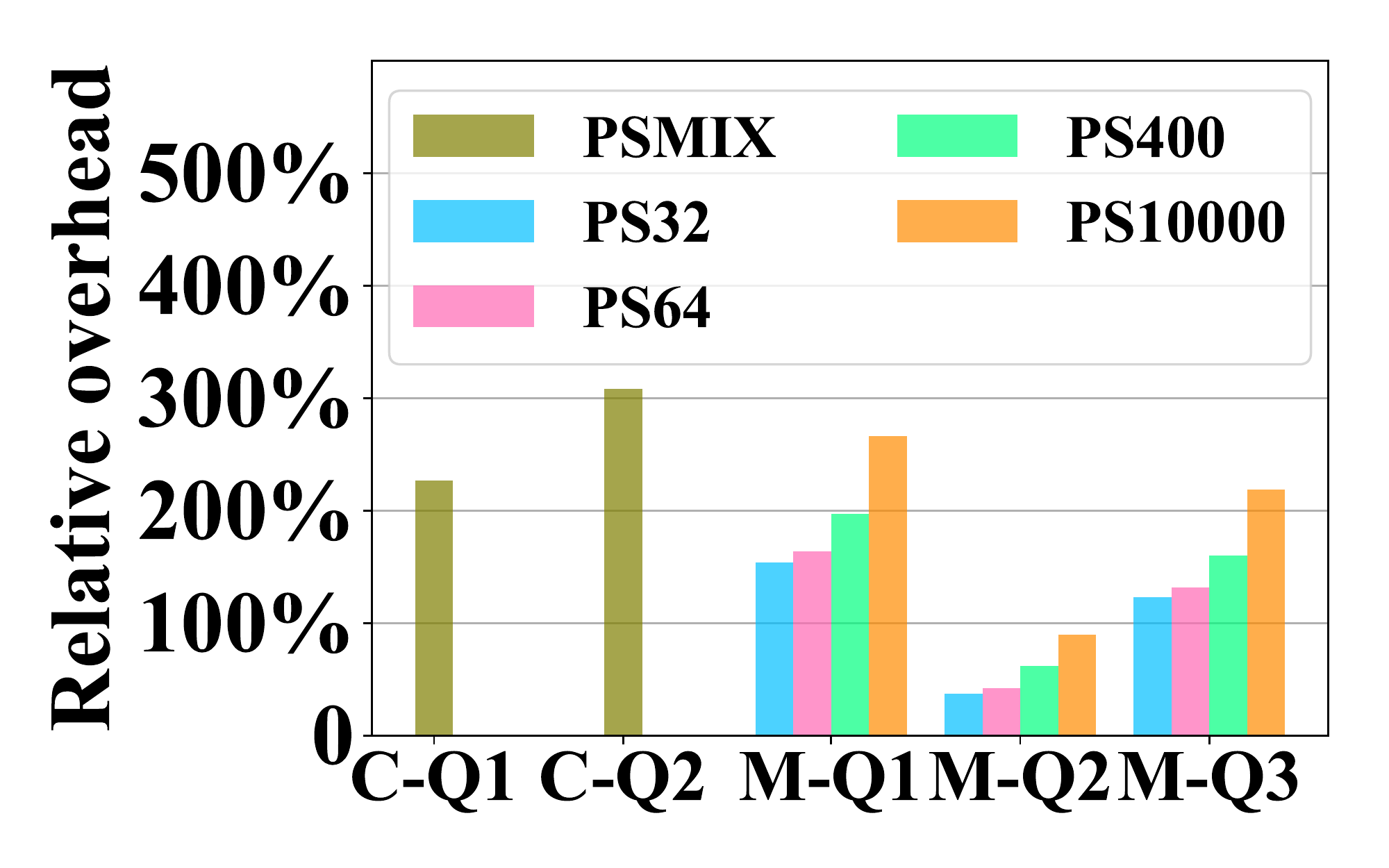}
      \caption{Crimes and Movies - Capture}
      \label{fig:realdb-cap}
    \end{subfigure}
        \begin{subfigure}{0.49\linewidth}
      \includegraphics[width=1\linewidth,trim=10pt 22pt 10pt 0pt, clip]{figs/realdb/stack_overflow_runtime.pdf}
      \caption{SOF - Use (BS)}
      \label{fig:stack-use}
    \end{subfigure}
    \begin{subfigure}{0.49\linewidth}
      \includegraphics[width=1\linewidth,trim=10pt 22pt 10pt 0pt, clip]{figs/realdb/stack_overflow_capture.pdf}
      \caption{SOF - Capture}
      \label{fig:stack-cap}
    \end{subfigure}
    \trimfigspace
    \caption{Real world data}
    \label{fig:real-word}
   \end{minipage}

  \begin{minipage}{1.0\linewidth}
    \centering
    \includegraphics[width=0.8\linewidth]{figs/pslegend.pdf}
  \end{minipage}

  \begin{minipage}[t]{1.0\linewidth}
    \begin{subfigure}{0.49\linewidth}
      \centering
      \includegraphics[height=3cm,trim=25pt 25pt 0 24pt, clip]{figs/binary_search/version1/post_tpch_1GB_runtime.pdf}
      \trimfigspace
      \caption{Postgres Use - 1GB (BS)}
      \label{fig:tpch-post-1gb}
    \end{subfigure}
    \begin{subfigure}{0.49\linewidth}
      \centering
      \includegraphics[height=3cm,trim=25pt 25pt 0 24pt, clip]{figs/binary_search/version1/post_tpch_1GB_capture.pdf}
      \trimfigspace
      \caption{Postgres Capture - 1GB}
      \label{fig:cap-tpch-post-1gb}
    \end{subfigure}
  \end{minipage}

  \begin{minipage}{1.0\linewidth}
    \begin{subfigure}{0.331\linewidth}
      \centering
      \includegraphics[height=3cm,trim=25pt 25pt 0 28pt, clip]{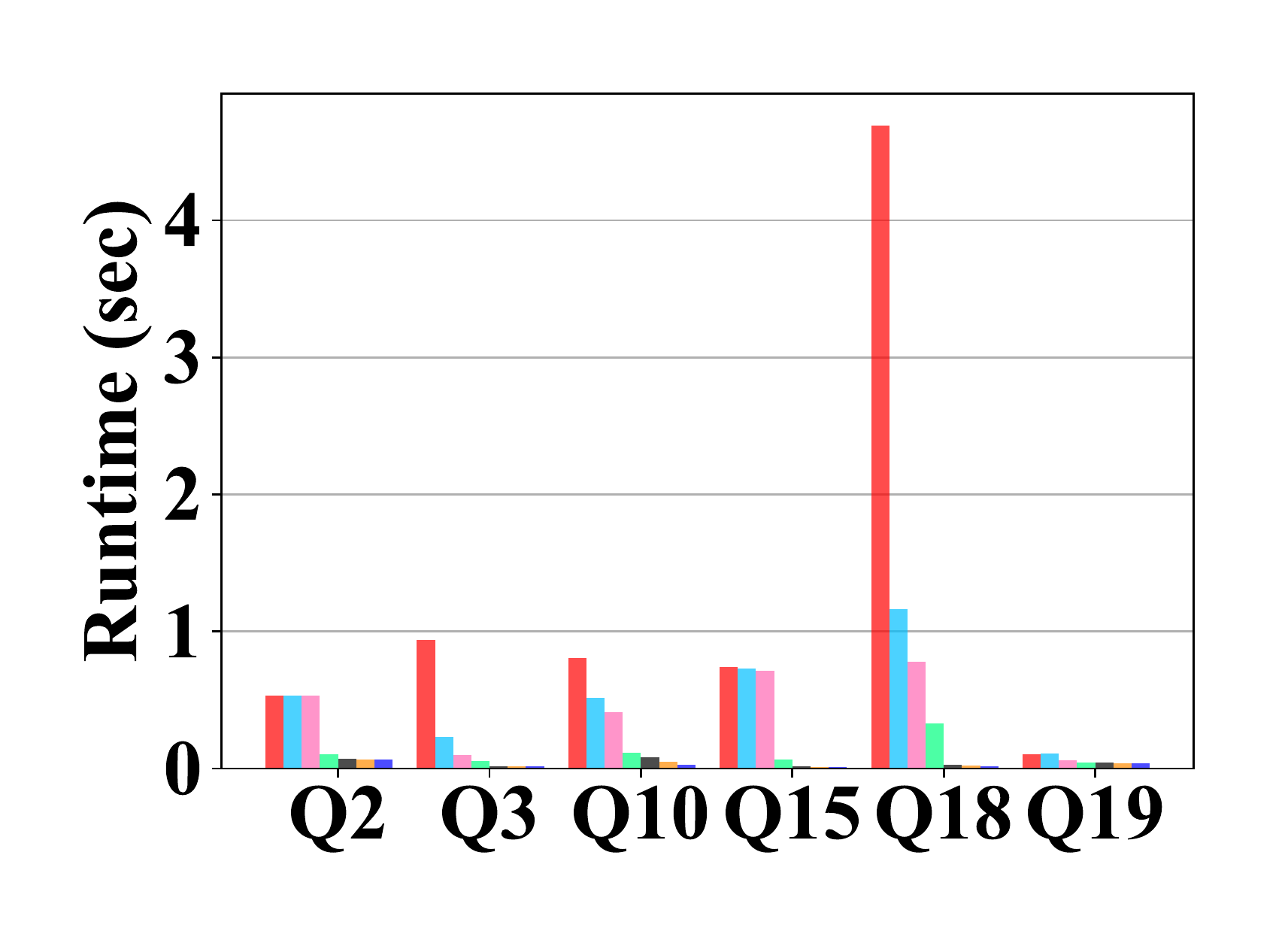}
      \trimfigspace
      \caption{Runtime - Postgres - 1GB (OR)}
      \label{fig:use-tpch-post-1gb-or}
    \end{subfigure}
    \begin{subfigure}{0.331\linewidth}
      \includegraphics[height=3cm,trim=25pt 25pt 0 28pt, clip]{figs/binary_search/version1/post_tpch_10GB_runtime.pdf}
      \trimfigspace
      \caption{Postgres Use - 10GB (BS)}
      \label{fig:tpch-post-10gb}
    \end{subfigure}
    \begin{subfigure}{0.331\linewidth}
      \includegraphics[height=3cm,trim=25pt 25pt 0 28pt, clip]{figs/binary_search/version1/post_tpch_10GB_capture.pdf}
      \trimfigspace
      \caption{Postgres Capture - 10GB}
      \label{fig:cap-tpch-post-10gb}
    \end{subfigure}
  \end{minipage}

  \begin{minipage}{1.0\linewidth}
    \begin{subfigure}{0.245\linewidth}
      \includegraphics[width=1\linewidth,trim=20pt 25pt 0 0pt, clip]{figs/monetdb/monetdb_tpch_1GB_runtime.pdf}
      \trimfigspace[-5mm]
      \caption{MonetDB Use - 1GB (OR)}
      \label{fig:monetdb-tpch-1gb}
    \end{subfigure}
    \begin{subfigure}{0.245\linewidth}
      \includegraphics[width=1\linewidth,trim=20pt 25pt 0 0pt, clip]{figs/monetdb/monetdb_tpch_10GB_runtime.pdf}
      \trimfigspace[-5mm]
      \caption{MonetDB Use - 10GB (OR)}
      \label{fig:monetdb-tpch-10gb}
    \end{subfigure}
    \begin{subfigure}{0.245\linewidth}
      \includegraphics[width=1\linewidth,trim=20pt 25pt 0 0pt, clip]{figs/monetdb/monetdb_tpch_1GB_capture.pdf}
      \trimfigspace[-5mm]
      \caption{MonetDB Capture - 1GB}
      \label{fig:monetdb-tpch-overhead-1gb}
    \end{subfigure}
    \begin{subfigure}{0.245\linewidth}
      \includegraphics[width=1\linewidth,trim=20pt 25pt 0 0pt, clip]{figs/monetdb/monetdb_tpch_10GB_capture.pdf}
      \trimfigspace[-5mm]
      \caption{MonetDB Capture - 10GB}
      \label{fig:monetdb-tpch-overhead-10gb}
    \end{subfigure}
    \trimfigspace
    \caption{Performance of provenance sketch capture and use for TPC-H queries.}
    \label{fig-exp-runtimes-usage}
  \end{minipage}

\end{figure*}

%% file: sections/exp_figs_opt_end.tex
\begin{figure*}[t]
  \centering
  \begin{minipage}{0.19\linewidth}
      \begin{subfigure}{1\linewidth}
      \includegraphics[width=1\linewidth,trim=10pt 22pt 10pt 0pt, clip]{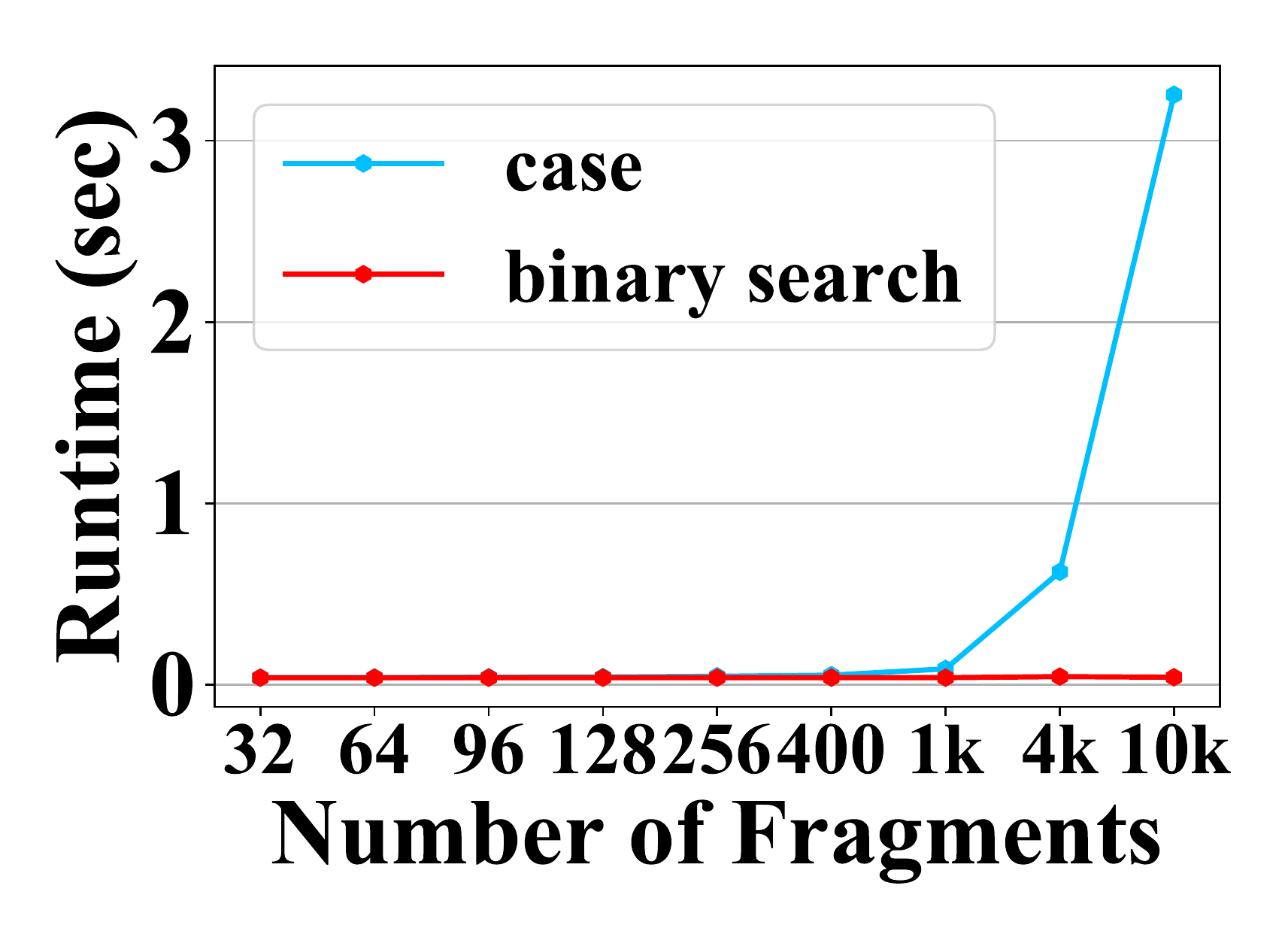}
      \caption{Creating Singletons}
      \label{fig:case-opt}
    \end{subfigure}
    \begin{subfigure}{1\linewidth}
      \includegraphics[width=1\linewidth,trim=10pt 22pt 10pt 0pt, clip]{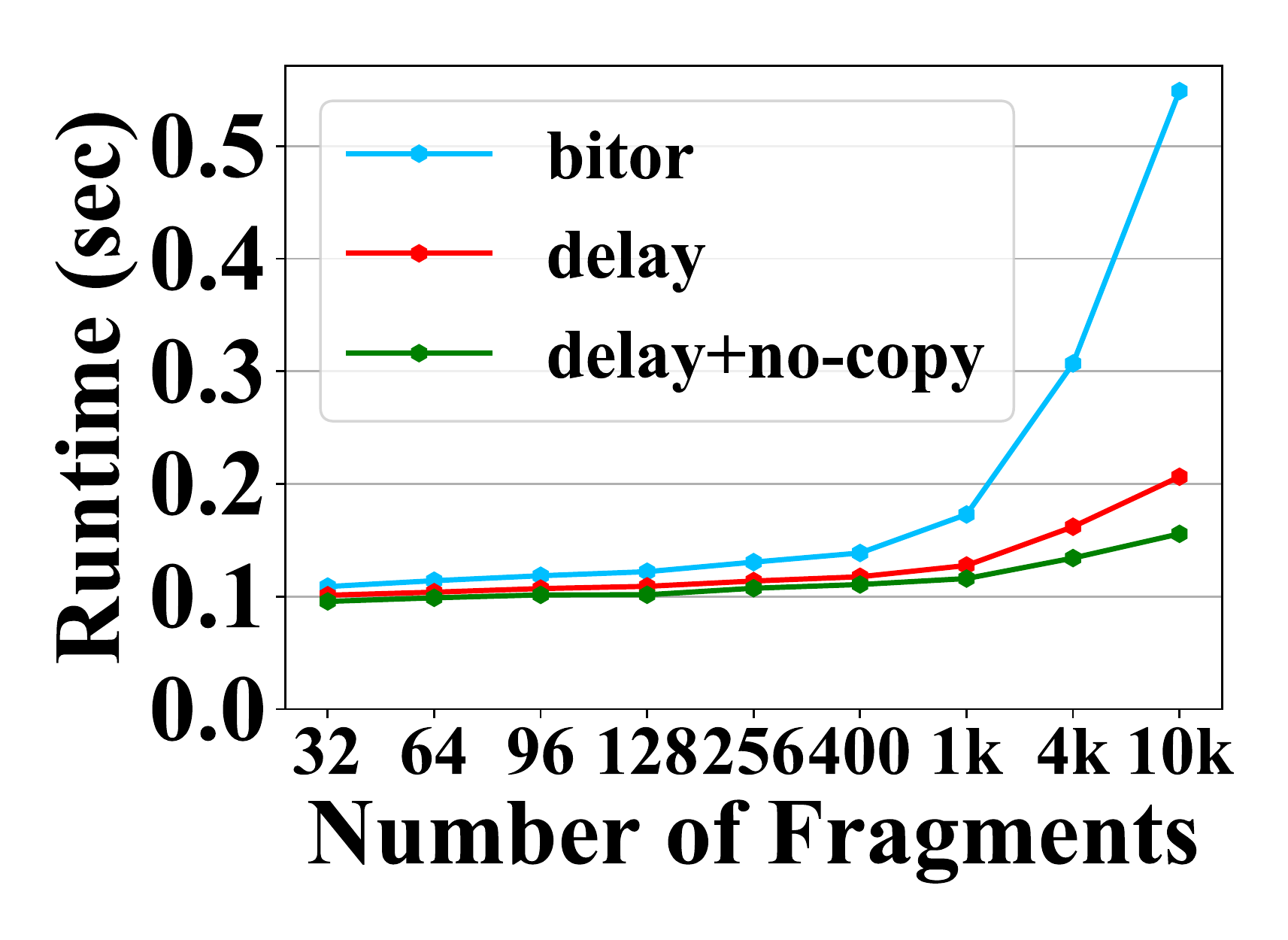}
      \caption{Merging Sketches}
      \label{fig:bitor-opt}
    \end{subfigure}
        \trimfigspace
        \caption{Optimizations}
    \label{fig:cap-opt}
  \end{minipage}
  \begin{minipage}{0.79\linewidth}
    \begin{subfigure}{0.24\linewidth}
      \includegraphics[width=1\linewidth,trim=10pt 22pt 10pt 0pt, clip]{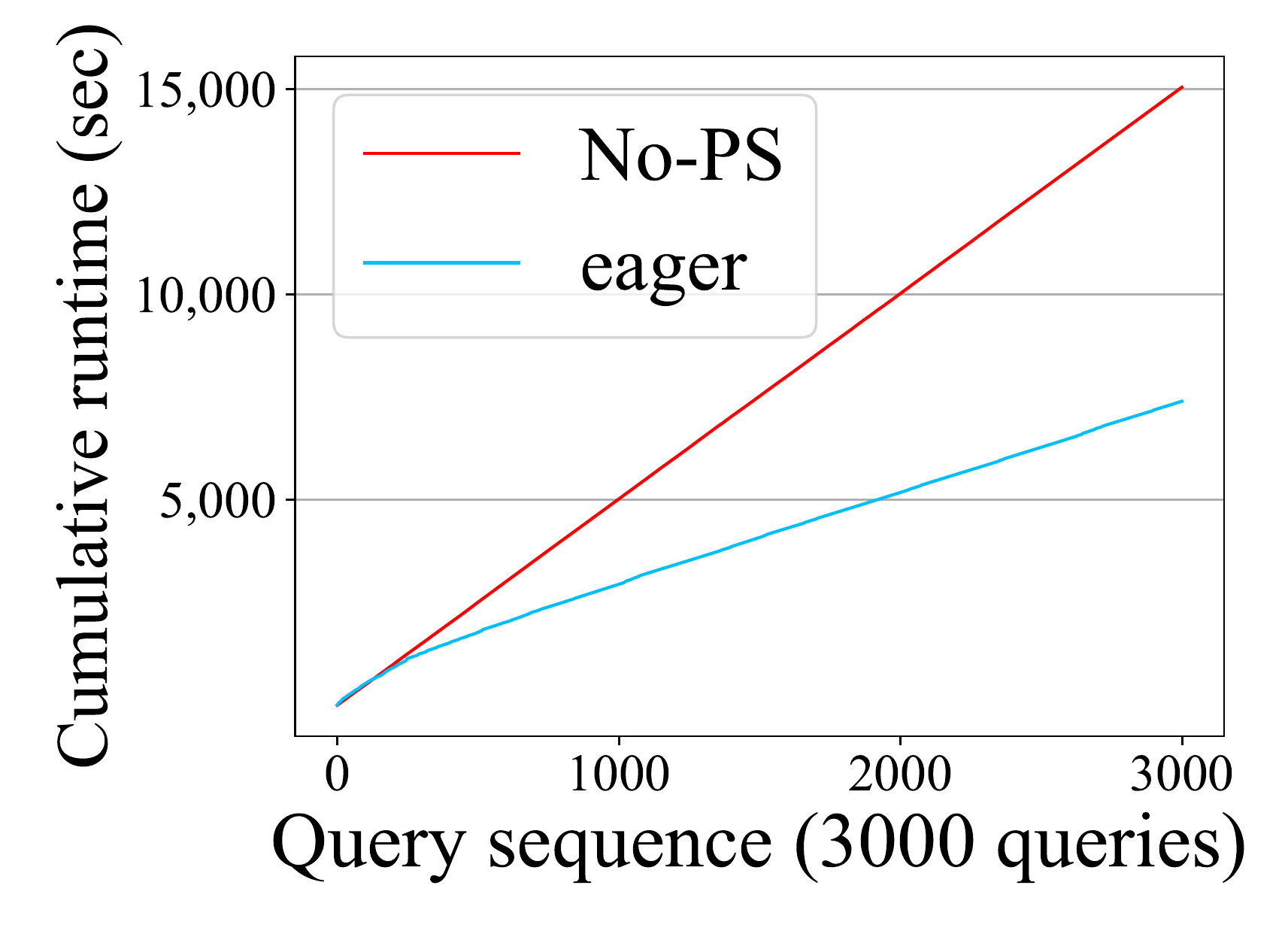}
      \caption{Crimes - Mix templates }
      \label{fig:ete-crimes-mix}
    \end{subfigure}
    \begin{subfigure}{0.24\linewidth}
      \includegraphics[width=1\linewidth,trim=10pt 22pt 10pt 0pt, clip]{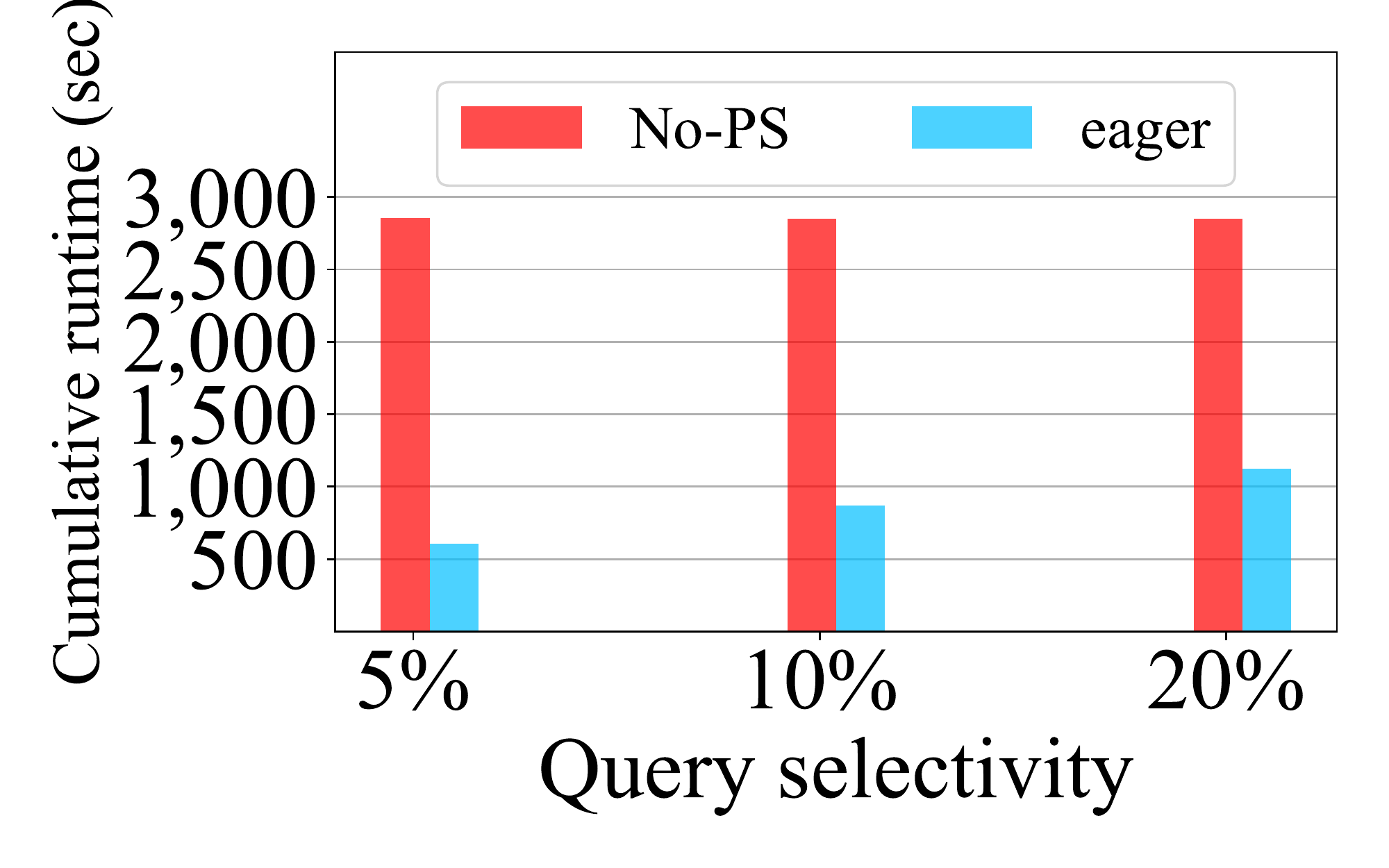}
      \caption{Crimes - Selectivity}
      \label{fig:ete-crimes-sel}
    \end{subfigure}
        \begin{subfigure}{0.24\linewidth}
      \includegraphics[width=1\linewidth,trim=10pt 22pt 10pt 0pt, clip]{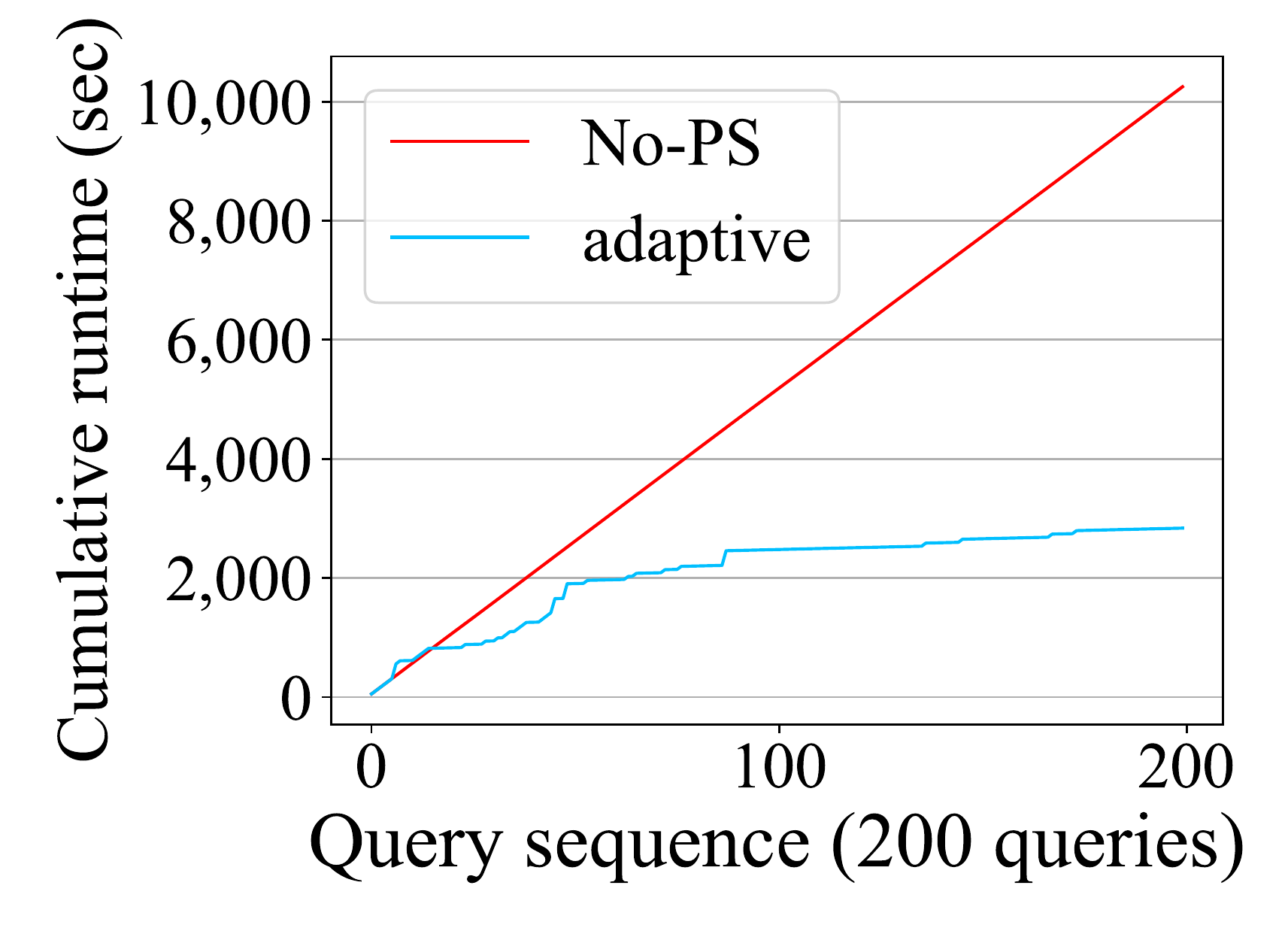}
      \caption{SOF ($SDV = 1k$)}
      \label{fig:ete-stack-q3-s1k}
    \end{subfigure}
        \begin{subfigure}{0.24\linewidth}
      \includegraphics[width=1\linewidth,trim=10pt 22pt 10pt 0pt, clip]{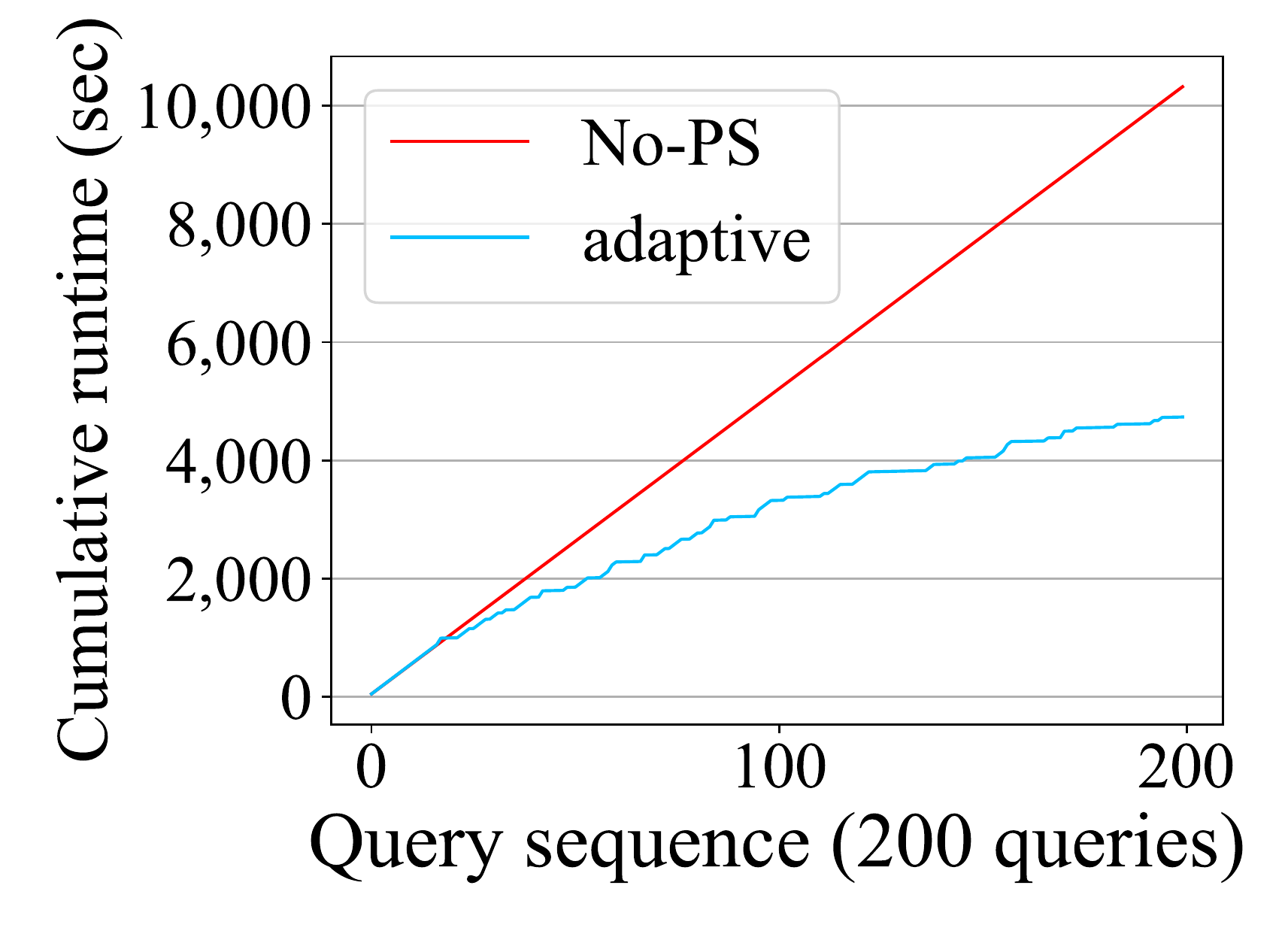}
      \caption{SOF ($SDV = 5k$)}
      \label{fig:ete-stack-q3-s5k}
    \end{subfigure}
        \begin{subfigure}{0.24\linewidth}
      \includegraphics[width=1\linewidth,trim=10pt 22pt 10pt 0pt, clip]{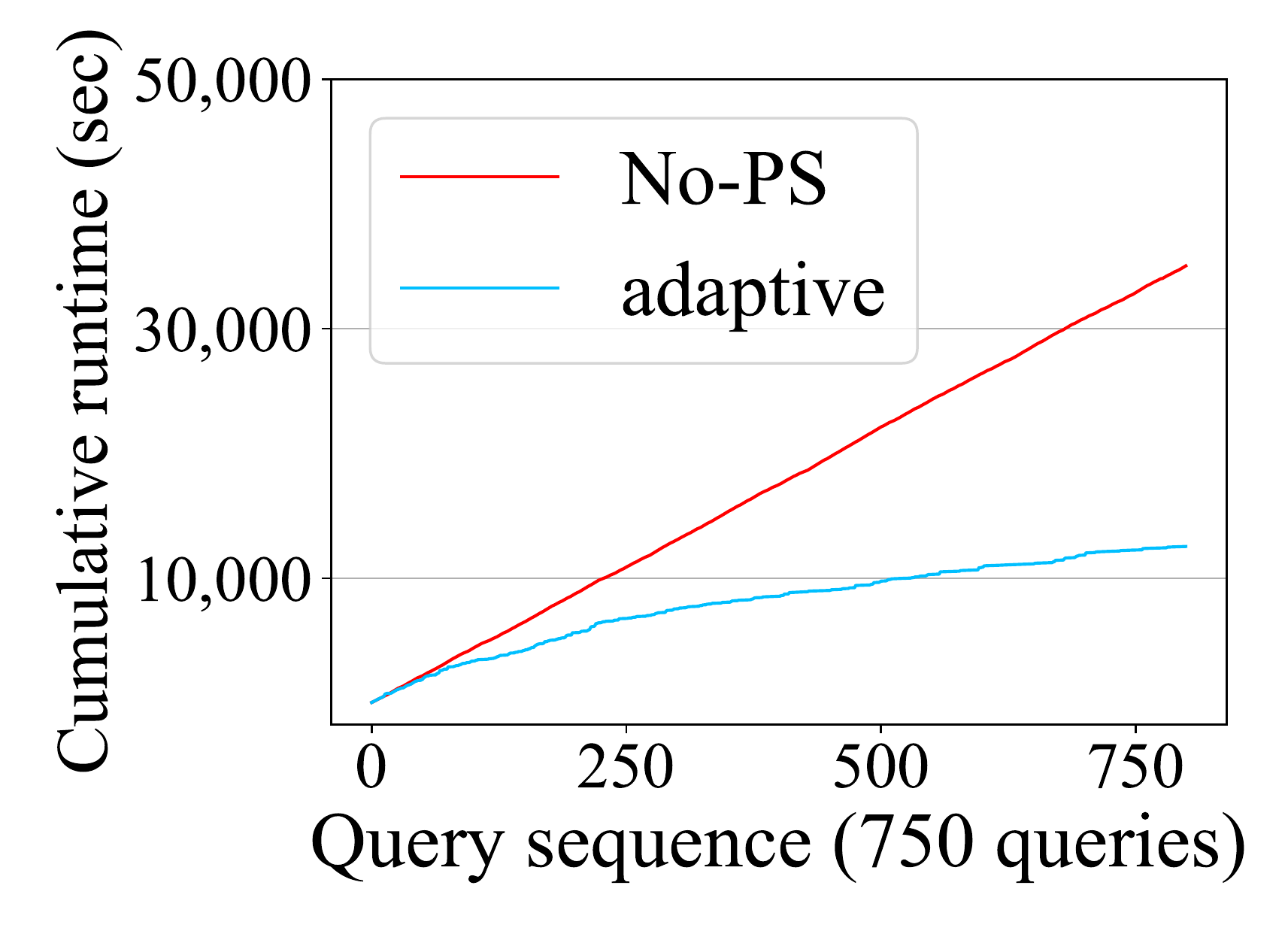}
      \caption{SOF - Mix templates}
      \label{fig:ete-stack-mix}
    \end{subfigure}
    \begin{subfigure}{0.24\linewidth}
      \includegraphics[width=1\linewidth,trim=10pt 22pt 10pt 0pt, clip]{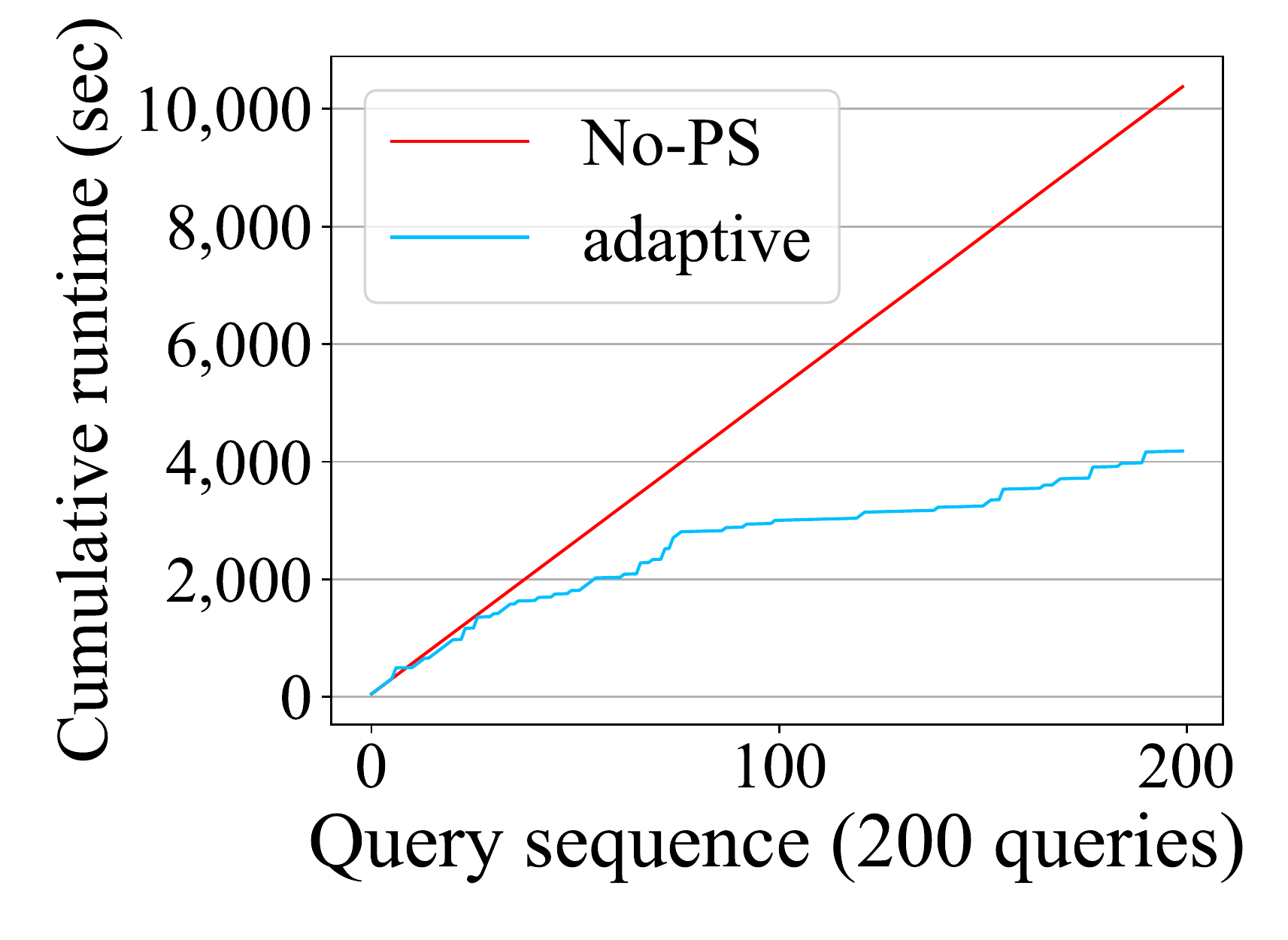}
      \caption{SOF ($sel = 0.7\%$)}
      \label{fig:ete-stack-sel07}
    \end{subfigure}
        \begin{subfigure}{0.24\linewidth}
      \includegraphics[width=1\linewidth,trim=10pt 22pt 10pt 0pt, clip]{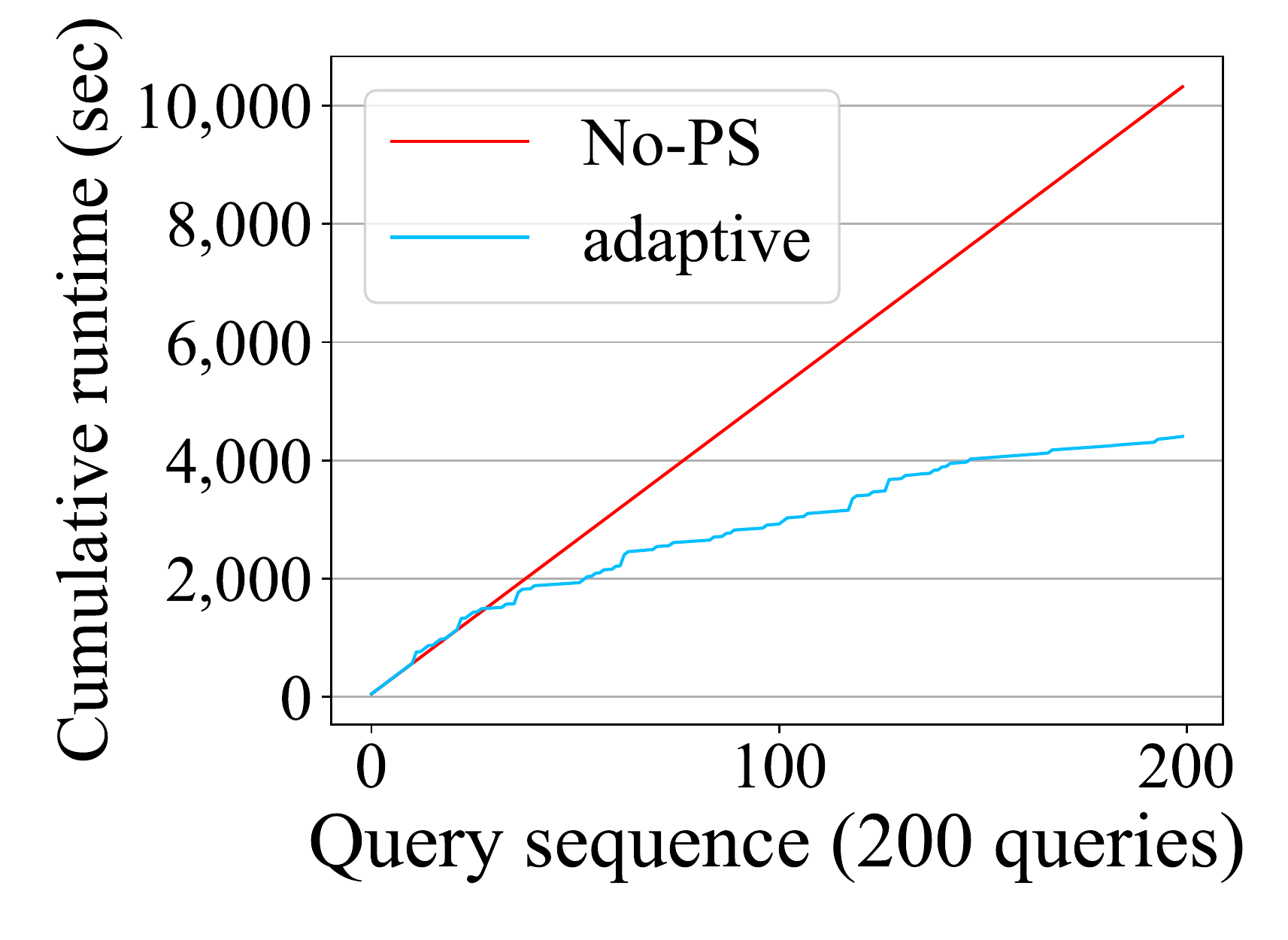}
      \caption{SOF ($sel = 2\%$)}
      \label{fig:ete-stack-sel2}
    \end{subfigure}
        \begin{subfigure}{0.24\linewidth}
      \includegraphics[width=1\linewidth,trim=10pt 22pt 10pt 0pt, clip]{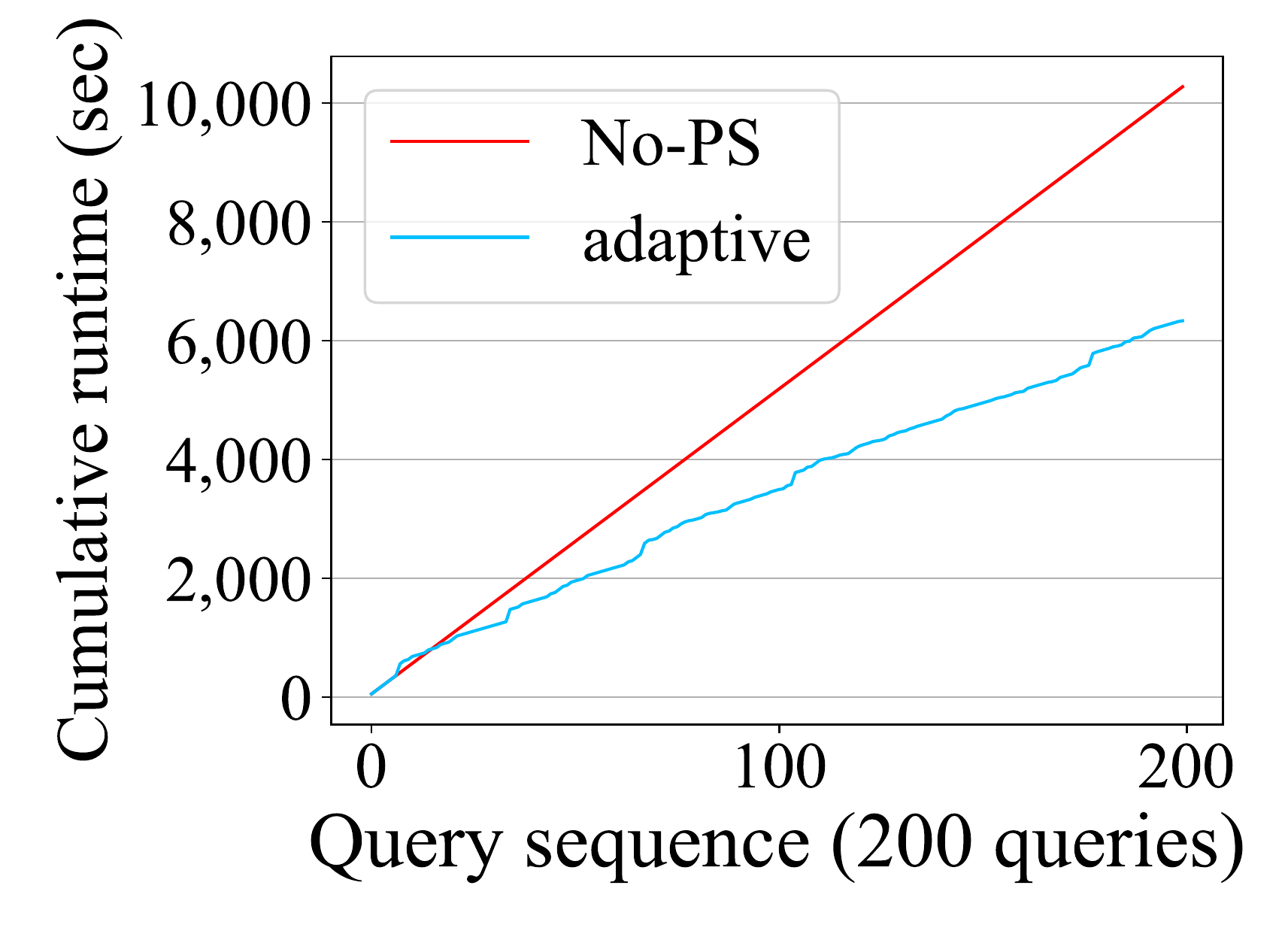}
      \caption{SOF ($sel = 5\%$)}
      \label{fig:ete-stack-sel5}
    \end{subfigure}
    \trimfigspace
      \caption{End-to-end Experiments}
    \label{fig:ete}
   \end{minipage}
  \end{figure*}

%% file: sections/exp_figs_end_sfo.tex
\begin{figure}[t]
  \centering
  \begin{minipage}{1\linewidth}
         \begin{subfigure}{0.32\linewidth}
      \includegraphics[width=1\linewidth,trim=50pt 50pt 10pt 0pt, clip]{figs/end_to_end_adapt/stack_overflow_end.pdf}
      \trimfigspace[-0.6cm]
      \caption{Mix templates}
      \label{fig:ete-stack-mix}
    \end{subfigure}
        \begin{subfigure}{0.32\linewidth}
      \includegraphics[width=1\linewidth,trim=50pt 50pt 10pt 0pt, clip]{figs/end_to_end_adapt/stack_q3_s1k_results.pdf}
     \trimfigspace[-0.6cm]
      \caption{$SDV = 1k$}
      \label{fig:ete-stack-q3-s1k}
    \end{subfigure}
        \begin{subfigure}{0.32\linewidth}
      \includegraphics[width=1\linewidth,trim=50pt 50pt 10pt 0pt, clip]{figs/end_to_end_adapt/stack_q3_s10k_results.pdf}
      \trimfigspace[-0.6cm]
      \caption{$SDV = 5k$}
      \label{fig:ete-stack-q3-s5k}
    \end{subfigure}
    \begin{subfigure}{0.32\linewidth}
      \includegraphics[width=1\linewidth,trim=50pt 50pt 10pt 0pt, clip]{figs/end_to_end_adapt/stack_q2_07p_results.pdf}
      \trimfigspace[-0.6cm]
      \caption{$sel = 0.7\%$}
      \label{fig:ete-stack-sel07}
    \end{subfigure}
        \begin{subfigure}{0.33\linewidth}
      \includegraphics[width=1\linewidth,trim=50pt 50pt 10pt 0pt, clip]{figs/end_to_end_adapt/stack_q2_2p_results.pdf}
      \trimfigspace[-0.6cm]
      \caption{$sel = 2\%$}
      \label{fig:ete-stack-sel2}
    \end{subfigure}
        \begin{subfigure}{0.33\linewidth}
      \includegraphics[width=1\linewidth,trim=50pt 50pt 10pt 0pt, clip]{figs/end_to_end_adapt/stack_q2_5p_results.pdf}
      \trimfigspace[-0.6cm]
      \caption{$sel = 5\%$}
      \label{fig:ete-stack-sel5}
    \end{subfigure}
   \trimfigspace[-0.4cm]
   \caption{End-to-end experiments on stack overflow data. We report the cumulative runtime in sec (the x-axis show the number of queries that have been executed up to that point).}
    \label{fig:ete}
   \end{minipage}
  \end{figure}


%% file: sections/conclusion.tex
\section{Conclusions and Future Work}
\label{sec:conclusion}

We present provenance-based data skipping (PBDS), a novel technique that determines at runtime which data is relevant for answering a query and then exploits this information to speed-up future queries. The main enabler of our approach are provenance sketches which are concise over-approximations of what data is relevant for a query. 
We develop self-tuning techniques for reusing a provenance sketch captured for one query to answer a different query.  
PBDS results in significant performance improvements for important classes of queries such as top-k queries that are highly selective, but where it is not possible to determine statically what data is relevant.
In the future, we will investigate how to maintain provenance sketches under updates and extend our self-tuning techniques to support wider range of queries and more powerful strategies. 


%% file: sections/future_work.tex


%% file: prov-data-skipping.bbl

\begin{thebibliography}{90}


\ifx \showCODEN    \undefined \def \showCODEN     #1{\unskip}     \fi
\ifx \showDOI      \undefined \def \showDOI       #1{#1}\fi
\ifx \showISBNx    \undefined \def \showISBNx     #1{\unskip}     \fi
\ifx \showISBNxiii \undefined \def \showISBNxiii  #1{\unskip}     \fi
\ifx \showISSN     \undefined \def \showISSN      #1{\unskip}     \fi
\ifx \showLCCN     \undefined \def \showLCCN      #1{\unskip}     \fi
\ifx \shownote     \undefined \def \shownote      #1{#1}          \fi
\ifx \showarticletitle \undefined \def \showarticletitle #1{#1}   \fi
\ifx \showURL      \undefined \def \showURL       {\relax}        \fi
\providecommand\bibfield[2]{#2}
\providecommand\bibinfo[2]{#2}
\providecommand\natexlab[1]{#1}
\providecommand\showeprint[2][]{arXiv:#2}

\bibitem[\protect\citeauthoryear{Abiteboul and Duschka}{Abiteboul and
  Duschka}{2013}]%
        {AD13c}
\bibfield{author}{\bibinfo{person}{Serge Abiteboul} {and}
  \bibinfo{person}{Olivier Duschka}.} \bibinfo{year}{2013}\natexlab{}.
\newblock \bibinfo{title}{Complexity of Answering Queries Using Materialized
  Views}.  (\bibinfo{year}{2013}).
\newblock


\bibitem[\protect\citeauthoryear{Abiteboul and Duschka}{Abiteboul and
  Duschka}{1998}]%
        {AD98a}
\bibfield{author}{\bibinfo{person}{Serge Abiteboul} {and}
  \bibinfo{person}{Oliver~M Duschka}.} \bibinfo{year}{1998}\natexlab{}.
\newblock \showarticletitle{Complexity of answering queries using materialized
  views}. In \bibinfo{booktitle}{\emph{PODS}}. \bibinfo{pages}{254--263}.
\newblock


\bibitem[\protect\citeauthoryear{Achakeev and Seeger}{Achakeev and
  Seeger}{2013}]%
        {AS13a}
\bibfield{author}{\bibinfo{person}{Daniar Achakeev} {and}
  \bibinfo{person}{Bernhard Seeger}.} \bibinfo{year}{2013}\natexlab{}.
\newblock \showarticletitle{Efficient bulk updates on multiversion B-trees}.
\newblock \bibinfo{journal}{\emph{PVLDB}} \bibinfo{volume}{6},
  \bibinfo{number}{14} (\bibinfo{year}{2013}), \bibinfo{pages}{1834--1845}.
\newblock


\bibitem[\protect\citeauthoryear{Aggarwal}{Aggarwal}{2009}]%
        {aggarwal2009trio}
\bibfield{author}{\bibinfo{person}{Charu~C Aggarwal}.}
  \bibinfo{year}{2009}\natexlab{}.
\newblock \showarticletitle{Trio A System for Data Uncertainty and Lineage}.
\newblock In \bibinfo{booktitle}{\emph{Managing and Mining Uncertain Data}}.
  \bibinfo{pages}{1--35}.
\newblock


\bibitem[\protect\citeauthoryear{Agrawal, Chaudhuri, and Narasayya}{Agrawal
  et~al\mbox{.}}{2000}]%
        {AC00}
\bibfield{author}{\bibinfo{person}{Sanjay Agrawal}, \bibinfo{person}{Surajit
  Chaudhuri}, {and} \bibinfo{person}{Vivek~R Narasayya}.}
  \bibinfo{year}{2000}\natexlab{}.
\newblock \showarticletitle{Automated Selection of Materialized Views and
  Indexes in SQL Databases.}. In \bibinfo{booktitle}{\emph{VLDB}},
  Vol.~\bibinfo{volume}{2000}. \bibinfo{pages}{496--505}.
\newblock


\bibitem[\protect\citeauthoryear{Agrawal, Narasayya, and Yang}{Agrawal
  et~al\mbox{.}}{2004}]%
        {AN04}
\bibfield{author}{\bibinfo{person}{S. Agrawal}, \bibinfo{person}{V. Narasayya},
  {and} \bibinfo{person}{B. Yang}.} \bibinfo{year}{2004}\natexlab{}.
\newblock \showarticletitle{Integrating vertical and horizontal partitioning
  into automated physical database design}. In
  \bibinfo{booktitle}{\emph{SIGMOD}}. \bibinfo{pages}{359--370}.
\newblock


\bibitem[\protect\citeauthoryear{Aguilera, Golab, and Shah}{Aguilera
  et~al\mbox{.}}{2008}]%
        {AG08}
\bibfield{author}{\bibinfo{person}{Marcos~K. Aguilera},
  \bibinfo{person}{Wojciech Golab}, {and} \bibinfo{person}{Mehul~A. Shah}.}
  \bibinfo{year}{2008}\natexlab{}.
\newblock \showarticletitle{A practical scalable distributed B-tree}.
\newblock \bibinfo{journal}{\emph{PVLDB}} \bibinfo{volume}{1},
  \bibinfo{number}{1} (\bibinfo{year}{2008}), \bibinfo{pages}{598--609}.
\newblock


\bibitem[\protect\citeauthoryear{Ahmad, Kennedy, Koch, and Nikolic}{Ahmad
  et~al\mbox{.}}{2012}]%
        {AK12}
\bibfield{author}{\bibinfo{person}{Yanif Ahmad}, \bibinfo{person}{Oliver
  Kennedy}, \bibinfo{person}{Christoph Koch}, {and} \bibinfo{person}{Milos
  Nikolic}.} \bibinfo{year}{2012}\natexlab{}.
\newblock \showarticletitle{DBToaster: Higher-order delta processing for
  dynamic, frequently fresh views}.
\newblock \bibinfo{journal}{\emph{PVLDB}} \bibinfo{volume}{5},
  \bibinfo{number}{10} (\bibinfo{year}{2012}), \bibinfo{pages}{968--979}.
\newblock


\bibitem[\protect\citeauthoryear{Ainy, Bourhis, Davidson, Deutch, and
  Milo}{Ainy et~al\mbox{.}}{2015}]%
        {AB15a}
\bibfield{author}{\bibinfo{person}{Eleanor Ainy}, \bibinfo{person}{Pierre
  Bourhis}, \bibinfo{person}{Susan~B. Davidson}, \bibinfo{person}{Daniel
  Deutch}, {and} \bibinfo{person}{Tova Milo}.} \bibinfo{year}{2015}\natexlab{}.
\newblock \showarticletitle{Approximated Summarization of Data Provenance}. In
  \bibinfo{booktitle}{\emph{CIKM}}. \bibinfo{pages}{483--492}.
\newblock


\bibitem[\protect\citeauthoryear{Amiri, Park, Tewari, and Padmanabhan}{Amiri
  et~al\mbox{.}}{2003}]%
        {amiri2003}
\bibfield{author}{\bibinfo{person}{Khalil Amiri}, \bibinfo{person}{Sanghyun
  Park}, \bibinfo{person}{Renu Tewari}, {and} \bibinfo{person}{Sriram
  Padmanabhan}.} \bibinfo{year}{2003}\natexlab{}.
\newblock \showarticletitle{Scalable template-based query containment checking
  for web semantic caches}. In \bibinfo{booktitle}{\emph{ICDE}}.
  \bibinfo{pages}{493--504}.
\newblock


\bibitem[\protect\citeauthoryear{Amsterdamer, Deutch, and Tannen}{Amsterdamer
  et~al\mbox{.}}{2011}]%
        {AD11d}
\bibfield{author}{\bibinfo{person}{Yael Amsterdamer}, \bibinfo{person}{Daniel
  Deutch}, {and} \bibinfo{person}{Val Tannen}.}
  \bibinfo{year}{2011}\natexlab{}.
\newblock \showarticletitle{Provenance for Aggregate Queries}. In
  \bibinfo{booktitle}{\emph{PODS}}. \bibinfo{pages}{153--164}.
\newblock


\bibitem[\protect\citeauthoryear{Anand, Bowers, McPhillips, and
  Ludäscher}{Anand et~al\mbox{.}}{2009}]%
        {AB09}
\bibfield{author}{\bibinfo{person}{Manish~Kumar Anand}, \bibinfo{person}{Shawn
  Bowers}, \bibinfo{person}{Timothy McPhillips}, {and} \bibinfo{person}{Bertram
  Ludäscher}.} \bibinfo{year}{2009}\natexlab{}.
\newblock \showarticletitle{Efficient Provenance Storage over Nested Data
  Collections}. In \bibinfo{booktitle}{\emph{EDBT}}. \bibinfo{pages}{958--969}.
\newblock


\bibitem[\protect\citeauthoryear{Aouiche, Darmont, Boussaid, and
  Bentayeb}{Aouiche et~al\mbox{.}}{2005}]%
        {AD05}
\bibfield{author}{\bibinfo{person}{Kamel Aouiche}, \bibinfo{person}{Jérôme
  Darmont}, \bibinfo{person}{Omar Boussaid}, {and} \bibinfo{person}{Fadila
  Bentayeb}.} \bibinfo{year}{2005}\natexlab{}.
\newblock \showarticletitle{Automatic Selection of Bitmap Join Indexes in Data
  Warehouses}. In \bibinfo{booktitle}{\emph{DaWaK}}. \bibinfo{pages}{64--73}.
\newblock


\bibitem[\protect\citeauthoryear{Arab, Feng, Glavic, Lee, Niu, and Zeng}{Arab
  et~al\mbox{.}}{2018}]%
        {AF18}
\bibfield{author}{\bibinfo{person}{Bahareh~Sadat Arab}, \bibinfo{person}{Su
  Feng}, \bibinfo{person}{Boris Glavic}, \bibinfo{person}{Seokki Lee},
  \bibinfo{person}{Xing Niu}, {and} \bibinfo{person}{Qitian Zeng}.}
  \bibinfo{year}{2018}\natexlab{}.
\newblock \showarticletitle{GProM - A Swiss Army Knife for Your Provenance
  Needs}.
\newblock \bibinfo{journal}{\emph{Data Eng. Bull.}} \bibinfo{volume}{41},
  \bibinfo{number}{1} (\bibinfo{year}{2018}), \bibinfo{pages}{51--62}.
\newblock


\bibitem[\protect\citeauthoryear{Assadi, Khanna, Li, and Tannen}{Assadi
  et~al\mbox{.}}{2016}]%
        {AK16a}
\bibfield{author}{\bibinfo{person}{Sepehr Assadi}, \bibinfo{person}{Sanjeev
  Khanna}, \bibinfo{person}{Yang Li}, {and} \bibinfo{person}{Val Tannen}.}
  \bibinfo{year}{2016}\natexlab{}.
\newblock \showarticletitle{Algorithms for Provisioning Queries and Analytics}.
  In \bibinfo{booktitle}{\emph{ICDT}}. \bibinfo{pages}{18:1--18:18}.
\newblock


\bibitem[\protect\citeauthoryear{Bhagwat, Chiticariu, Tan, and
  Vijayvargiya}{Bhagwat et~al\mbox{.}}{2005}]%
        {bhagwat2005annotation}
\bibfield{author}{\bibinfo{person}{Deepavali Bhagwat}, \bibinfo{person}{Laura
  Chiticariu}, \bibinfo{person}{Wang-Chiew Tan}, {and} \bibinfo{person}{Gaurav
  Vijayvargiya}.} \bibinfo{year}{2005}\natexlab{}.
\newblock \showarticletitle{An annotation management system for relational
  databases}.
\newblock \bibinfo{journal}{\emph{VLDBJ}} \bibinfo{volume}{14},
  \bibinfo{number}{4} (\bibinfo{year}{2005}), \bibinfo{pages}{373--396}.
\newblock


\bibitem[\protect\citeauthoryear{Böhm, Berchtold, Kriegel, and Michel}{Böhm
  et~al\mbox{.}}{2000}]%
        {BB00a}
\bibfield{author}{\bibinfo{person}{C. Böhm}, \bibinfo{person}{S. Berchtold},
  \bibinfo{person}{H.P. Kriegel}, {and} \bibinfo{person}{U. Michel}.}
  \bibinfo{year}{2000}\natexlab{}.
\newblock \showarticletitle{Multidimensional index structures in relational
  databases}.
\newblock \bibinfo{journal}{\emph{Journal of Intelligent Information Systems}}
  \bibinfo{volume}{15}, \bibinfo{number}{1} (\bibinfo{year}{2000}),
  \bibinfo{pages}{51--70}.
\newblock


\bibitem[\protect\citeauthoryear{Ceri, Negri, and Pelagatti}{Ceri
  et~al\mbox{.}}{1982}]%
        {ceri1982horizontal}
\bibfield{author}{\bibinfo{person}{Stefano Ceri}, \bibinfo{person}{Mauro
  Negri}, {and} \bibinfo{person}{Giuseppe Pelagatti}.}
  \bibinfo{year}{1982}\natexlab{}.
\newblock \showarticletitle{Horizontal data partitioning in database design}.
  In \bibinfo{booktitle}{\emph{SIGMOD}}. \bibinfo{pages}{128--136}.
\newblock


\bibitem[\protect\citeauthoryear{Chandra and Merlin}{Chandra and
  Merlin}{1977}]%
        {chandra1977}
\bibfield{author}{\bibinfo{person}{Ashok~K Chandra} {and}
  \bibinfo{person}{Philip~M Merlin}.} \bibinfo{year}{1977}\natexlab{}.
\newblock \showarticletitle{Optimal implementation of conjunctive queries in
  relational data bases}. In \bibinfo{booktitle}{\emph{STOC}}.
  \bibinfo{pages}{77--90}.
\newblock


\bibitem[\protect\citeauthoryear{Chapman, Jagadish, and Ramanan}{Chapman
  et~al\mbox{.}}{2008}]%
        {CJ08a}
\bibfield{author}{\bibinfo{person}{Adriane Chapman}, \bibinfo{person}{H.~V.
  Jagadish}, {and} \bibinfo{person}{Prakash Ramanan}.}
  \bibinfo{year}{2008}\natexlab{}.
\newblock \showarticletitle{Efficient Provenance Storage}. In
  \bibinfo{booktitle}{\emph{SIGMOD}}. \bibinfo{pages}{993--1006}.
\newblock


\bibitem[\protect\citeauthoryear{Chaudhuri, Datar, and Narasayya}{Chaudhuri
  et~al\mbox{.}}{2004}]%
        {CD04}
\bibfield{author}{\bibinfo{person}{Surajit Chaudhuri}, \bibinfo{person}{Mayur
  Datar}, {and} \bibinfo{person}{Vivek Narasayya}.}
  \bibinfo{year}{2004}\natexlab{}.
\newblock \showarticletitle{Index selection for databases: A hardness study and
  a principled heuristic solution}.
\newblock \bibinfo{journal}{\emph{TKDE}} \bibinfo{volume}{16},
  \bibinfo{number}{11} (\bibinfo{year}{2004}), \bibinfo{pages}{1313--1323}.
\newblock


\bibitem[\protect\citeauthoryear{Chaudhuri, Krishnamurthy, Potamianos, and
  Shim}{Chaudhuri et~al\mbox{.}}{1995}]%
        {chaudhuri1995optimizing}
\bibfield{author}{\bibinfo{person}{Surajit Chaudhuri}, \bibinfo{person}{Ravi
  Krishnamurthy}, \bibinfo{person}{Spyros Potamianos}, {and}
  \bibinfo{person}{Kyuseok Shim}.} \bibinfo{year}{1995}\natexlab{}.
\newblock \showarticletitle{Optimizing queries with materialized views}. In
  \bibinfo{booktitle}{\emph{ICDE}}. \bibinfo{pages}{190--190}.
\newblock


\bibitem[\protect\citeauthoryear{Chaudhuri and Narasayya}{Chaudhuri and
  Narasayya}{2007}]%
        {CN07}
\bibfield{author}{\bibinfo{person}{Surajit Chaudhuri} {and}
  \bibinfo{person}{Vivek Narasayya}.} \bibinfo{year}{2007}\natexlab{}.
\newblock \showarticletitle{Self-tuning database systems: a decade of
  progress}. In \bibinfo{booktitle}{\emph{VLDB}}. \bibinfo{pages}{3--14}.
\newblock


\bibitem[\protect\citeauthoryear{Chen, Lehri, Kuan~Loh, Alur, Jia, Loo, and
  Zhou}{Chen et~al\mbox{.}}{2017}]%
        {CL17}
\bibfield{author}{\bibinfo{person}{Chen Chen}, \bibinfo{person}{Harshal~Tushar
  Lehri}, \bibinfo{person}{Lay Kuan~Loh}, \bibinfo{person}{Anupam Alur},
  \bibinfo{person}{Limin Jia}, \bibinfo{person}{Boon~Thau Loo}, {and}
  \bibinfo{person}{Wenchao Zhou}.} \bibinfo{year}{2017}\natexlab{}.
\newblock \showarticletitle{Distributed Provenance Compression}. In
  \bibinfo{booktitle}{\emph{SIGMOD}}. \bibinfo{pages}{203--218}.
\newblock


\bibitem[\protect\citeauthoryear{Cheney, Chiticariu, and Tan}{Cheney
  et~al\mbox{.}}{2009}]%
        {CC09}
\bibfield{author}{\bibinfo{person}{James Cheney}, \bibinfo{person}{Laura
  Chiticariu}, {and} \bibinfo{person}{Wang-Chiew Tan}.}
  \bibinfo{year}{2009}\natexlab{}.
\newblock \showarticletitle{Provenance in Databases: Why, How, and Where}.
\newblock \bibinfo{journal}{\emph{Foundations and Trends in Databases}}
  \bibinfo{volume}{1}, \bibinfo{number}{4} (\bibinfo{year}{2009}),
  \bibinfo{pages}{379--474}.
\newblock


\bibitem[\protect\citeauthoryear{Clarke}{Clarke}{2013}]%
        {clarke2013storage}
\bibfield{author}{\bibinfo{person}{John Clarke}.}
  \bibinfo{year}{2013}\natexlab{}.
\newblock \showarticletitle{Storage indexes}.
\newblock In \bibinfo{booktitle}{\emph{Oracle Exadata Recipes}}.
  \bibinfo{pages}{553--576}.
\newblock


\bibitem[\protect\citeauthoryear{Cui, Widom, and Wiener}{Cui
  et~al\mbox{.}}{2000}]%
        {CW00b}
\bibfield{author}{\bibinfo{person}{Yingwei Cui}, \bibinfo{person}{Jennifer
  Widom}, {and} \bibinfo{person}{Janet~L. Wiener}.}
  \bibinfo{year}{2000}\natexlab{}.
\newblock \showarticletitle{Tracing the Lineage of View Data in a Warehousing
  Environment}.
\newblock \bibinfo{journal}{\emph{TODS}} \bibinfo{volume}{25},
  \bibinfo{number}{2} (\bibinfo{year}{2000}), \bibinfo{pages}{179--227}.
\newblock


\bibitem[\protect\citeauthoryear{De~Moura and Bj{\o}rner}{De~Moura and
  Bj{\o}rner}{2008}]%
        {de2008z3}
\bibfield{author}{\bibinfo{person}{Leonardo De~Moura} {and}
  \bibinfo{person}{Nikolaj Bj{\o}rner}.} \bibinfo{year}{2008}\natexlab{}.
\newblock \showarticletitle{Z3: An efficient SMT solver}. In
  \bibinfo{booktitle}{\emph{International conference on Tools and Algorithms
  for the Construction and Analysis of Systems}}. Springer,
  \bibinfo{pages}{337--340}.
\newblock


\bibitem[\protect\citeauthoryear{Deutch, Ives, Milo, and Tannen}{Deutch
  et~al\mbox{.}}{2013a}]%
        {DI13}
\bibfield{author}{\bibinfo{person}{D. Deutch}, \bibinfo{person}{Z. Ives},
  \bibinfo{person}{T. Milo}, {and} \bibinfo{person}{V. Tannen}.}
  \bibinfo{year}{2013}\natexlab{a}.
\newblock \showarticletitle{Caravan: Provisioning for What-If Analysis}.
\newblock \bibinfo{journal}{\emph{CIDR}} (\bibinfo{year}{2013}).
\newblock


\bibitem[\protect\citeauthoryear{Deutch, Moskovitch, and Rinetzky}{Deutch
  et~al\mbox{.}}{2019}]%
        {DM19}
\bibfield{author}{\bibinfo{person}{Daniel Deutch}, \bibinfo{person}{Yuval
  Moskovitch}, {and} \bibinfo{person}{Noam Rinetzky}.}
  \bibinfo{year}{2019}\natexlab{}.
\newblock \showarticletitle{Hypothetical Reasoning via Provenance Abstraction}.
  In \bibinfo{booktitle}{\emph{SIGMOD}}. \bibinfo{pages}{537--554}.
\newblock


\bibitem[\protect\citeauthoryear{Deutch, Moskovitch, and Tannen}{Deutch
  et~al\mbox{.}}{2013b}]%
        {DM13}
\bibfield{author}{\bibinfo{person}{Daniel Deutch}, \bibinfo{person}{Yuval
  Moskovitch}, {and} \bibinfo{person}{Val Tannen}.}
  \bibinfo{year}{2013}\natexlab{b}.
\newblock \showarticletitle{PROPOLIS: Provisioned Analysis of Data-Centric
  Processes}.
\newblock \bibinfo{journal}{\emph{PVLDB}} \bibinfo{volume}{6},
  \bibinfo{number}{12} (\bibinfo{year}{2013}).
\newblock


\bibitem[\protect\citeauthoryear{Du}{Du}{2013}]%
        {D13}
\bibfield{author}{\bibinfo{person}{Jiang Du}.} \bibinfo{year}{2013}\natexlab{}.
\newblock \showarticletitle{DeepSea: self-adaptive data partitioning and
  replication in scalable distributed data systems}. In
  \bibinfo{booktitle}{\emph{PODS}}. \bibinfo{pages}{7--12}.
\newblock


\bibitem[\protect\citeauthoryear{Du, Glavic, Tan, and Miller}{Du
  et~al\mbox{.}}{2017}]%
        {DG17}
\bibfield{author}{\bibinfo{person}{Jiang Du}, \bibinfo{person}{Boris Glavic},
  \bibinfo{person}{Wei Tan}, {and} \bibinfo{person}{Renée~J. Miller}.}
  \bibinfo{year}{2017}\natexlab{}.
\newblock \showarticletitle{DeepSea: Adaptive Workload-Aware Partitioning of
  Materialized Views in Scalable Data Analytics}. In
  \bibinfo{booktitle}{\emph{EDBT}}. \bibinfo{pages}{198--209}.
\newblock


\bibitem[\protect\citeauthoryear{El~Gebaly, Agrawal, Golab, Korn, and
  Srivastava}{El~Gebaly et~al\mbox{.}}{2014}]%
        {EA14}
\bibfield{author}{\bibinfo{person}{Kareem El~Gebaly}, \bibinfo{person}{Parag
  Agrawal}, \bibinfo{person}{Lukasz Golab}, \bibinfo{person}{Flip Korn}, {and}
  \bibinfo{person}{Divesh Srivastava}.} \bibinfo{year}{2014}\natexlab{}.
\newblock \showarticletitle{Interpretable and informative explanations of
  outcomes}.
\newblock \bibinfo{journal}{\emph{PVLDB}} \bibinfo{volume}{8},
  \bibinfo{number}{1} (\bibinfo{year}{2014}).
\newblock


\bibitem[\protect\citeauthoryear{El~Gebaly, Feng, Golab, Korn, and
  Srivastava}{El~Gebaly et~al\mbox{.}}{2018}]%
        {EF18}
\bibfield{author}{\bibinfo{person}{Kareem El~Gebaly}, \bibinfo{person}{Guoyao
  Feng}, \bibinfo{person}{Lukasz Golab}, \bibinfo{person}{Flip Korn}, {and}
  \bibinfo{person}{Divesh Srivastava}.} \bibinfo{year}{2018}\natexlab{}.
\newblock \showarticletitle{Explanation Tables}.
\newblock \bibinfo{journal}{\emph{Sat}}  \bibinfo{volume}{5}
  (\bibinfo{year}{2018}), \bibinfo{pages}{14}.
\newblock


\bibitem[\protect\citeauthoryear{Fehrenbach and Cheney}{Fehrenbach and
  Cheney}{2018}]%
        {FC18a}
\bibfield{author}{\bibinfo{person}{Stefan Fehrenbach} {and}
  \bibinfo{person}{James Cheney}.} \bibinfo{year}{2018}\natexlab{}.
\newblock \showarticletitle{Language-integrated provenance}.
\newblock \bibinfo{journal}{\emph{Sci. Comput. Program.}}
  \bibinfo{volume}{155} (\bibinfo{year}{2018}), \bibinfo{pages}{103--145}.
\newblock


\bibitem[\protect\citeauthoryear{Geerts and Poggi}{Geerts and Poggi}{2010}]%
        {GP10}
\bibfield{author}{\bibinfo{person}{F. Geerts} {and} \bibinfo{person}{A.
  Poggi}.} \bibinfo{year}{2010}\natexlab{}.
\newblock \showarticletitle{On database query languages for K-relations}.
\newblock \bibinfo{journal}{\emph{Journal of Applied Logic}}
  \bibinfo{volume}{8}, \bibinfo{number}{2} (\bibinfo{year}{2010}),
  \bibinfo{pages}{173--185}.
\newblock


\bibitem[\protect\citeauthoryear{Glavic, Köhler, Riddle, and
  Ludäscher}{Glavic et~al\mbox{.}}{2015}]%
        {GK15}
\bibfield{author}{\bibinfo{person}{Boris Glavic}, \bibinfo{person}{Sven
  Köhler}, \bibinfo{person}{Sean Riddle}, {and} \bibinfo{person}{Bertram
  Ludäscher}.} \bibinfo{year}{2015}\natexlab{}.
\newblock \showarticletitle{Towards Constraint-based Explanations for Answers
  and Non-Answers}. In \bibinfo{booktitle}{\emph{TaPP}}.
\newblock


\bibitem[\protect\citeauthoryear{Glavic, Miller, and Alonso}{Glavic
  et~al\mbox{.}}{2013}]%
        {glavic2013using}
\bibfield{author}{\bibinfo{person}{Boris Glavic}, \bibinfo{person}{Renée~J
  Miller}, {and} \bibinfo{person}{Gustavo Alonso}.}
  \bibinfo{year}{2013}\natexlab{}.
\newblock \showarticletitle{Using SQL for Efficient Generation and Querying of
  Provenance Information}.
\newblock In \bibinfo{booktitle}{\emph{In Search of Elegance in the Theory and
  Practice of Computation}}. \bibinfo{pages}{291--320}.
\newblock


\bibitem[\protect\citeauthoryear{Goldstein and Larson}{Goldstein and
  Larson}{2001}]%
        {GL01}
\bibfield{author}{\bibinfo{person}{J. Goldstein} {and} \bibinfo{person}{P.Å.
  Larson}.} \bibinfo{year}{2001}\natexlab{}.
\newblock \showarticletitle{Optimizing queries using materialized views: a
  practical, scalable solution}.
\newblock \bibinfo{journal}{\emph{SIGMOD Record}} \bibinfo{volume}{30},
  \bibinfo{number}{2} (\bibinfo{year}{2001}), \bibinfo{pages}{331--342}.
\newblock


\bibitem[\protect\citeauthoryear{Graefe}{Graefe}{2006}]%
        {G06}
\bibfield{author}{\bibinfo{person}{Goetz Graefe}.}
  \bibinfo{year}{2006}\natexlab{}.
\newblock \showarticletitle{B-tree indexes for high update rates}.
\newblock \bibinfo{journal}{\emph{SIGMOD Record}} \bibinfo{volume}{35},
  \bibinfo{number}{1} (\bibinfo{year}{2006}), \bibinfo{pages}{39--44}.
\newblock


\bibitem[\protect\citeauthoryear{Graefe and Kuno}{Graefe and Kuno}{2010}]%
        {GK10a}
\bibfield{author}{\bibinfo{person}{Goetz Graefe} {and} \bibinfo{person}{Harumi
  Kuno}.} \bibinfo{year}{2010}\natexlab{}.
\newblock \showarticletitle{Self-selecting, self-tuning, incrementally
  optimized indexes}. In \bibinfo{booktitle}{\emph{EDBT}}.
  \bibinfo{pages}{371--381}.
\newblock


\bibitem[\protect\citeauthoryear{Green, Aref, and Karvounarakis}{Green
  et~al\mbox{.}}{2012}]%
        {GA12}
\bibfield{author}{\bibinfo{person}{Todd~J Green}, \bibinfo{person}{Molham
  Aref}, {and} \bibinfo{person}{Grigoris Karvounarakis}.}
  \bibinfo{year}{2012}\natexlab{}.
\newblock \showarticletitle{Logicblox, platform and language: A tutorial}.
\newblock In \bibinfo{booktitle}{\emph{Datalog in Academia and Industry}}.
  \bibinfo{pages}{1--8}.
\newblock


\bibitem[\protect\citeauthoryear{Green, Karvounarakis, and Tannen}{Green
  et~al\mbox{.}}{2007}]%
        {GK07}
\bibfield{author}{\bibinfo{person}{Todd~J. Green}, \bibinfo{person}{Gregory
  Karvounarakis}, {and} \bibinfo{person}{Val Tannen}.}
  \bibinfo{year}{2007}\natexlab{}.
\newblock \showarticletitle{Provenance Semirings}. In
  \bibinfo{booktitle}{\emph{PODS}}. \bibinfo{pages}{31--40}.
\newblock


\bibitem[\protect\citeauthoryear{Green and Tannen}{Green and Tannen}{2017}]%
        {GT17}
\bibfield{author}{\bibinfo{person}{Todd~J Green} {and} \bibinfo{person}{Val
  Tannen}.} \bibinfo{year}{2017}\natexlab{}.
\newblock \showarticletitle{The Semiring Framework for Database Provenance}. In
  \bibinfo{booktitle}{\emph{PODS}}. \bibinfo{pages}{93--99}.
\newblock


\bibitem[\protect\citeauthoryear{Gupta and Mumick}{Gupta and Mumick}{1999}]%
        {GM99}
\bibfield{author}{\bibinfo{person}{A. Gupta} {and} \bibinfo{person}{I.S.
  Mumick}.} \bibinfo{year}{1999}\natexlab{}.
\newblock \bibinfo{booktitle}{\emph{Materialized views: techniques,
  implementations, and applications}}.
\newblock \bibinfo{publisher}{MIT press}.
\newblock


\bibitem[\protect\citeauthoryear{Halevy}{Halevy}{2001}]%
        {halevy2001answering}
\bibfield{author}{\bibinfo{person}{Alon~Y Halevy}.}
  \bibinfo{year}{2001}\natexlab{}.
\newblock \showarticletitle{Answering queries using views: A survey}.
\newblock \bibinfo{journal}{\emph{VLDB}} \bibinfo{volume}{10},
  \bibinfo{number}{4} (\bibinfo{year}{2001}), \bibinfo{pages}{270--294}.
\newblock


\bibitem[\protect\citeauthoryear{Heinis and Alonso}{Heinis and Alonso}{2008}]%
        {HA08}
\bibfield{author}{\bibinfo{person}{Thomas Heinis} {and}
  \bibinfo{person}{Gustavo Alonso}.} \bibinfo{year}{2008}\natexlab{}.
\newblock \showarticletitle{Efficient Lineage Tracking for Scientific
  Workflows}. In \bibinfo{booktitle}{\emph{SIGMOD}}.
  \bibinfo{pages}{1007--1018}.
\newblock


\bibitem[\protect\citeauthoryear{Héman, Nes, Żukowski, and Boncz}{Héman
  et~al\mbox{.}}{2008}]%
        {HN08}
\bibfield{author}{\bibinfo{person}{Sándor Héman}, \bibinfo{person}{Niels~J
  Nes}, \bibinfo{person}{Marcin Żukowski}, {and}
  \bibinfo{person}{Peter~Alexander Boncz}.} \bibinfo{year}{2008}\natexlab{}.
\newblock \bibinfo{booktitle}{\emph{Positional Delta Trees to reconcile updates
  with read-optimized data storage}}.
\newblock \bibinfo{publisher}{CWI. Information Systems [INS]}.
\newblock


\bibitem[\protect\citeauthoryear{Idreos, Manegold, Kuno, and Graefe}{Idreos
  et~al\mbox{.}}{2011}]%
        {IM11}
\bibfield{author}{\bibinfo{person}{Stratos Idreos}, \bibinfo{person}{Stefan
  Manegold}, \bibinfo{person}{Harumi Kuno}, {and} \bibinfo{person}{Goetz
  Graefe}.} \bibinfo{year}{2011}\natexlab{}.
\newblock \showarticletitle{Merging what's cracked, cracking what's merged:
  adaptive indexing in main-memory column-stores}.
\newblock \bibinfo{journal}{\emph{PVLDB}} \bibinfo{volume}{4},
  \bibinfo{number}{9} (\bibinfo{year}{2011}), \bibinfo{pages}{586--597}.
\newblock


\bibitem[\protect\citeauthoryear{Ikeda, Salihoglu, and Widom}{Ikeda
  et~al\mbox{.}}{2010}]%
        {IS10}
\bibfield{author}{\bibinfo{person}{Robert Ikeda}, \bibinfo{person}{Semih
  Salihoglu}, {and} \bibinfo{person}{Jennifer Widom}.}
  \bibinfo{year}{2010}\natexlab{}.
\newblock \bibinfo{booktitle}{\emph{Provenance-Based Refresh in Data-Oriented
  Workflows}}.
\newblock \bibinfo{type}{technical report}.
\newblock


\bibitem[\protect\citeauthoryear{Ikeda and Widom}{Ikeda and Widom}{2010}]%
        {IW10}
\bibfield{author}{\bibinfo{person}{Robert Ikeda} {and}
  \bibinfo{person}{Jennifer Widom}.} \bibinfo{year}{2010}\natexlab{}.
\newblock \showarticletitle{Panda: A System for Provenance and Data}. In
  \bibinfo{booktitle}{\emph{TaPP '10}}.
\newblock


\bibitem[\protect\citeauthoryear{Jindal and Dittrich}{Jindal and
  Dittrich}{2012}]%
        {JD12}
\bibfield{author}{\bibinfo{person}{Alekh Jindal} {and} \bibinfo{person}{Jens
  Dittrich}.} \bibinfo{year}{2012}\natexlab{}.
\newblock \showarticletitle{Relax and let the database do the partitioning
  online}.
\newblock In \bibinfo{booktitle}{\emph{Enabling Real-Time Business
  Intelligence}}. \bibinfo{pages}{65--80}.
\newblock


\bibitem[\protect\citeauthoryear{Karvounarakis and Green}{Karvounarakis and
  Green}{2012}]%
        {KG12}
\bibfield{author}{\bibinfo{person}{G. Karvounarakis} {and}
  \bibinfo{person}{T.J. Green}.} \bibinfo{year}{2012}\natexlab{}.
\newblock \showarticletitle{Semiring-Annotated Data: Queries and Provenance}.
\newblock \bibinfo{journal}{\emph{SIGMOD}} \bibinfo{volume}{41},
  \bibinfo{number}{3} (\bibinfo{year}{2012}), \bibinfo{pages}{5--14}.
\newblock


\bibitem[\protect\citeauthoryear{Klug}{Klug}{1988}]%
        {klug1988}
\bibfield{author}{\bibinfo{person}{Anthony Klug}.}
  \bibinfo{year}{1988}\natexlab{}.
\newblock \showarticletitle{On conjunctive queries containing inequalities}.
\newblock \bibinfo{journal}{\emph{JACM}} \bibinfo{volume}{35},
  \bibinfo{number}{1} (\bibinfo{year}{1988}), \bibinfo{pages}{146--160}.
\newblock


\bibitem[\protect\citeauthoryear{Köhler, Ludäscher, and Smaragdakis}{Köhler
  et~al\mbox{.}}{2012}]%
        {KL12}
\bibfield{author}{\bibinfo{person}{S. Köhler}, \bibinfo{person}{B.
  Ludäscher}, {and} \bibinfo{person}{Y. Smaragdakis}.}
  \bibinfo{year}{2012}\natexlab{}.
\newblock \showarticletitle{Declarative datalog debugging for mere mortals}.
\newblock \bibinfo{journal}{\emph{Datalog in Academia and Industry}}
  (\bibinfo{year}{2012}), \bibinfo{pages}{111--122}.
\newblock


\bibitem[\protect\citeauthoryear{Lee, Ludäscher, and Glavic}{Lee
  et~al\mbox{.}}{2018}]%
        {LGG18}
\bibfield{author}{\bibinfo{person}{Seokki Lee}, \bibinfo{person}{Bertram
  Ludäscher}, {and} \bibinfo{person}{Boris Glavic}.}
  \bibinfo{year}{2018}\natexlab{}.
\newblock \showarticletitle{Provenance Summaries for Answers and Non-Answers}.
\newblock \bibinfo{journal}{\emph{PVDLB}} \bibinfo{volume}{11},
  \bibinfo{number}{12} (\bibinfo{year}{2018}), \bibinfo{pages}{1954--1957}.
\newblock


\bibitem[\protect\citeauthoryear{Lee, Ludäscher, and Glavic}{Lee
  et~al\mbox{.}}{2020}]%
        {LL20}
\bibfield{author}{\bibinfo{person}{Seokki Lee}, \bibinfo{person}{Bertram
  Ludäscher}, {and} \bibinfo{person}{Boris Glavic}.}
  \bibinfo{year}{2020}\natexlab{}.
\newblock \showarticletitle{Approximate Summaries for Why and Why-not
  Provenance}.
\newblock \bibinfo{journal}{\emph{PVLDB}} \bibinfo{volume}{13},
  \bibinfo{number}{6} (\bibinfo{year}{2020}), \bibinfo{pages}{912--924}.
\newblock


\bibitem[\protect\citeauthoryear{Leis, Kemper, and Neumann}{Leis
  et~al\mbox{.}}{2013}]%
        {LK13}
\bibfield{author}{\bibinfo{person}{Viktor Leis}, \bibinfo{person}{Alfons
  Kemper}, {and} \bibinfo{person}{Thomas Neumann}.}
  \bibinfo{year}{2013}\natexlab{}.
\newblock \showarticletitle{The adaptive radix tree: ARTful indexing for
  main-memory databases}. In \bibinfo{booktitle}{\emph{ICDE}}.
  \bibinfo{pages}{38--49}.
\newblock


\bibitem[\protect\citeauthoryear{Levandoski, Lomet, Sengupta, Birka, and
  Diaconu}{Levandoski et~al\mbox{.}}{2014}]%
        {LL14b}
\bibfield{author}{\bibinfo{person}{Justin Levandoski}, \bibinfo{person}{David
  Lomet}, \bibinfo{person}{Sudipta Sengupta}, \bibinfo{person}{Adrian Birka},
  {and} \bibinfo{person}{Cristian Diaconu}.} \bibinfo{year}{2014}\natexlab{}.
\newblock \showarticletitle{Indexing on modern hardware: Hekaton and beyond}.
  In \bibinfo{booktitle}{\emph{SIGMOD}}. \bibinfo{pages}{717--720}.
\newblock


\bibitem[\protect\citeauthoryear{Levandoski, Lomet, and Sengupta}{Levandoski
  et~al\mbox{.}}{2013}]%
        {LL13}
\bibfield{author}{\bibinfo{person}{Justin~J Levandoski},
  \bibinfo{person}{David~B Lomet}, {and} \bibinfo{person}{Sudipta Sengupta}.}
  \bibinfo{year}{2013}\natexlab{}.
\newblock \showarticletitle{The Bw-Tree: A B-tree for new hardware platforms}.
  In \bibinfo{booktitle}{\emph{ICDE}}. \bibinfo{pages}{302--313}.
\newblock


\bibitem[\protect\citeauthoryear{Levy, Mendelzon, and Sagiv}{Levy
  et~al\mbox{.}}{1995}]%
        {LM95a}
\bibfield{author}{\bibinfo{person}{Alon~Y Levy}, \bibinfo{person}{Alberto~O
  Mendelzon}, {and} \bibinfo{person}{Yehoshua Sagiv}.}
  \bibinfo{year}{1995}\natexlab{}.
\newblock \showarticletitle{Answering queries using views}. In
  \bibinfo{booktitle}{\emph{PODS}}. \bibinfo{pages}{95--104}.
\newblock


\bibitem[\protect\citeauthoryear{Li, Xu, and Malik}{Li et~al\mbox{.}}{2016}]%
        {LX16}
\bibfield{author}{\bibinfo{person}{Xiang Li}, \bibinfo{person}{Xiaoyang Xu},
  {and} \bibinfo{person}{Tanu Malik}.} \bibinfo{year}{2016}\natexlab{}.
\newblock \showarticletitle{Interactive provenance summaries for reproducible
  science}. In \bibinfo{booktitle}{\emph{eScience}}. \bibinfo{pages}{355--360}.
\newblock


\bibitem[\protect\citeauthoryear{Li and Ross}{Li and Ross}{1999}]%
        {LR99}
\bibfield{author}{\bibinfo{person}{Zhe Li} {and} \bibinfo{person}{Kenneth~A.
  Ross}.} \bibinfo{year}{1999}\natexlab{}.
\newblock \showarticletitle{Fast Joins Using Join Indices}.
\newblock \bibinfo{journal}{\emph{VLDBJ}} \bibinfo{volume}{8},
  \bibinfo{number}{1} (\bibinfo{year}{1999}), \bibinfo{pages}{1--24}.
\newblock


\bibitem[\protect\citeauthoryear{Malik, Nistor, and Gehani}{Malik
  et~al\mbox{.}}{2010}]%
        {MN10}
\bibfield{author}{\bibinfo{person}{T. Malik}, \bibinfo{person}{L. Nistor},
  {and} \bibinfo{person}{A. Gehani}.} \bibinfo{year}{2010}\natexlab{}.
\newblock \showarticletitle{Tracking and Sketching Distributed Data
  Provenance}. In \bibinfo{booktitle}{\emph{eScience}}.
  \bibinfo{pages}{190--197}.
\newblock


\bibitem[\protect\citeauthoryear{Moerkotte}{Moerkotte}{1998}]%
        {moerkotte1998small}
\bibfield{author}{\bibinfo{person}{Guido Moerkotte}.}
  \bibinfo{year}{1998}\natexlab{}.
\newblock \showarticletitle{Small materialized aggregates: A light weight index
  structure for data warehousing}.
\newblock  (\bibinfo{year}{1998}).
\newblock


\bibitem[\protect\citeauthoryear{Müller, Dietrich, and Grust}{Müller
  et~al\mbox{.}}{2018}]%
        {MD18}
\bibfield{author}{\bibinfo{person}{Tobias Müller}, \bibinfo{person}{Benjamin
  Dietrich}, {and} \bibinfo{person}{Torsten Grust}.}
  \bibinfo{year}{2018}\natexlab{}.
\newblock \showarticletitle{You Say `What', I Hear `Where'and `Why'---(Mis-)
  Interpreting SQL to Derive Fine-Grained Provenance}.
\newblock \bibinfo{journal}{\emph{PVLDB}} \bibinfo{volume}{11},
  \bibinfo{number}{11} (\bibinfo{year}{2018}).
\newblock


\bibitem[\protect\citeauthoryear{Navathe and Ra}{Navathe and Ra}{1989}]%
        {NR89}
\bibfield{author}{\bibinfo{person}{S.B. Navathe} {and} \bibinfo{person}{M.
  Ra}.} \bibinfo{year}{1989}\natexlab{}.
\newblock \showarticletitle{Vertical partitioning for database design: a
  graphical algorithm}. In \bibinfo{booktitle}{\emph{SIGMOD}}.
  \bibinfo{pages}{450}.
\newblock


\bibitem[\protect\citeauthoryear{Niu, Kapoor, Glavic, Gawlick, Liu,
  Krishnaswamy, and Radhakrishnan}{Niu et~al\mbox{.}}{2017}]%
        {XN17}
\bibfield{author}{\bibinfo{person}{Xing Niu}, \bibinfo{person}{Raghav Kapoor},
  \bibinfo{person}{Boris Glavic}, \bibinfo{person}{Dieter Gawlick},
  \bibinfo{person}{Zhen~Hua Liu}, \bibinfo{person}{Vasudha Krishnaswamy}, {and}
  \bibinfo{person}{Venkatesh Radhakrishnan}.} \bibinfo{year}{2017}\natexlab{}.
\newblock \showarticletitle{Provenance-aware Query Optimization}. In
  \bibinfo{booktitle}{\emph{ICDE}}. \bibinfo{pages}{473--484}.
\newblock


\bibitem[\protect\citeauthoryear{Niu, Kapoor, Glavic, Gawlick, Liu,
  Krishnaswamy, and Radhakrishnan}{Niu et~al\mbox{.}}{2018}]%
        {XN18}
\bibfield{author}{\bibinfo{person}{Xing Niu}, \bibinfo{person}{Raghav Kapoor},
  \bibinfo{person}{Boris Glavic}, \bibinfo{person}{Dieter Gawlick},
  \bibinfo{person}{Zhen~Hua Liu}, \bibinfo{person}{Vasudha Krishnaswamy}, {and}
  \bibinfo{person}{Venkatesh Radhakrishnan}.} \bibinfo{year}{2018}\natexlab{}.
\newblock \showarticletitle{Heuristic and Cost-based Optimization for Diverse
  Provenance Tasks}.
\newblock \bibinfo{journal}{\emph{TKDE}} \bibinfo{volume}{31},
  \bibinfo{number}{7} (\bibinfo{year}{2018}), \bibinfo{pages}{1267--1280}.
\newblock


\bibitem[\protect\citeauthoryear{Niu, Liu, Li, and Glavic}{Niu
  et~al\mbox{.}}{2021}]%
        {techreport}
\bibfield{author}{\bibinfo{person}{Xing Niu}, \bibinfo{person}{Ziyu Liu},
  \bibinfo{person}{Pengyuan Li}, {and} \bibinfo{person}{Boris Glavic}.}
  \bibinfo{year}{2021}\natexlab{}.
\newblock \showarticletitle{Provenance-based Data Skipping (extended version)}.
\newblock  (\bibinfo{year}{2021}).
\newblock
\showeprint[arxiv]{2104.12815}


\bibitem[\protect\citeauthoryear{Olteanu and Schleich}{Olteanu and
  Schleich}{2016}]%
        {OS16}
\bibfield{author}{\bibinfo{person}{Dan Olteanu} {and}
  \bibinfo{person}{Maximilian Schleich}.} \bibinfo{year}{2016}\natexlab{}.
\newblock \showarticletitle{Factorized Databases}.
\newblock \bibinfo{journal}{\emph{SIGMOD Record}} \bibinfo{volume}{45},
  \bibinfo{number}{2} (\bibinfo{year}{2016}), \bibinfo{pages}{5--16}.
\newblock


\bibitem[\protect\citeauthoryear{Olteanu and Závodný}{Olteanu and
  Závodný}{2011}]%
        {OZ11}
\bibfield{author}{\bibinfo{person}{Dan Olteanu} {and} \bibinfo{person}{Jakub
  Závodný}.} \bibinfo{year}{2011}\natexlab{}.
\newblock \showarticletitle{On Factorisation of Provenance Polynomials}. In
  \bibinfo{booktitle}{\emph{TaPP}}.
\newblock


\bibitem[\protect\citeauthoryear{O'Neil and Graefe}{O'Neil and Graefe}{1995}]%
        {OG95}
\bibfield{author}{\bibinfo{person}{Patrick O'Neil} {and} \bibinfo{person}{Goetz
  Graefe}.} \bibinfo{year}{1995}\natexlab{}.
\newblock \showarticletitle{Multi-table joins through bitmapped join indices}.
\newblock \bibinfo{journal}{\emph{SIGMOD Record}} \bibinfo{volume}{24},
  \bibinfo{number}{3} (\bibinfo{year}{1995}), \bibinfo{pages}{8--11}.
\newblock


\bibitem[\protect\citeauthoryear{Papadomanolakis and Ailamaki}{Papadomanolakis
  and Ailamaki}{2004}]%
        {PA04a}
\bibfield{author}{\bibinfo{person}{S. Papadomanolakis} {and}
  \bibinfo{person}{A. Ailamaki}.} \bibinfo{year}{2004}\natexlab{}.
\newblock \showarticletitle{Autopart: Automating schema design for large
  scientific databases using data partitioning}.
\newblock  (\bibinfo{year}{2004}).
\newblock


\bibitem[\protect\citeauthoryear{Perez and Jermaine}{Perez and
  Jermaine}{2014}]%
        {PJ14}
\bibfield{author}{\bibinfo{person}{Luis~L. Perez} {and}
  \bibinfo{person}{Christopher~M. Jermaine}.} \bibinfo{year}{2014}\natexlab{}.
\newblock \showarticletitle{History-aware Query Optimization with Materialized
  Intermediate Views}. In \bibinfo{booktitle}{\emph{ICDE}}.
\newblock


\bibitem[\protect\citeauthoryear{Psallidas and Wu}{Psallidas and Wu}{2018}]%
        {PW18}
\bibfield{author}{\bibinfo{person}{Fotis Psallidas} {and}
  \bibinfo{person}{Eugene Wu}.} \bibinfo{year}{2018}\natexlab{}.
\newblock \showarticletitle{Smoke: Fine-grained lineage at interactive speed}.
\newblock \bibinfo{journal}{\emph{PVLDB}} \bibinfo{volume}{11},
  \bibinfo{number}{6} (\bibinfo{year}{2018}), \bibinfo{pages}{719--732}.
\newblock


\bibitem[\protect\citeauthoryear{Rabl and Jacobsen}{Rabl and Jacobsen}{2017}]%
        {RJ17}
\bibfield{author}{\bibinfo{person}{Tilmann Rabl} {and}
  \bibinfo{person}{Hans-Arno Jacobsen}.} \bibinfo{year}{2017}\natexlab{}.
\newblock \showarticletitle{Query Centric Partitioning and Allocation for
  Partially Replicated Database Systems}. In
  \bibinfo{booktitle}{\emph{SIGMOD}}. \bibinfo{pages}{315--330}.
\newblock


\bibitem[\protect\citeauthoryear{Roy, Orr, and Suciu}{Roy
  et~al\mbox{.}}{2015}]%
        {RO15}
\bibfield{author}{\bibinfo{person}{Sudeepa Roy}, \bibinfo{person}{Laurel Orr},
  {and} \bibinfo{person}{Dan Suciu}.} \bibinfo{year}{2015}\natexlab{}.
\newblock \showarticletitle{Explaining query answers with explanation-ready
  databases}.
\newblock \bibinfo{journal}{\emph{PVLDB}} \bibinfo{volume}{9},
  \bibinfo{number}{4} (\bibinfo{year}{2015}), \bibinfo{pages}{348--359}.
\newblock


\bibitem[\protect\citeauthoryear{Roy and Suciu}{Roy and Suciu}{2014}]%
        {RS14}
\bibfield{author}{\bibinfo{person}{Sudeepa Roy} {and} \bibinfo{person}{Dan
  Suciu}.} \bibinfo{year}{2014}\natexlab{}.
\newblock \showarticletitle{A formal approach to finding explanations for
  database queries}.
\newblock


\bibitem[\protect\citeauthoryear{Sagiv and Yannakakis}{Sagiv and
  Yannakakis}{1980}]%
        {sagiv1980}
\bibfield{author}{\bibinfo{person}{Yehoshua Sagiv} {and}
  \bibinfo{person}{Mihalis Yannakakis}.} \bibinfo{year}{1980}\natexlab{}.
\newblock \showarticletitle{Equivalences among relational expressions with the
  union and difference operators}.
\newblock \bibinfo{journal}{\emph{JACM}} \bibinfo{volume}{27},
  \bibinfo{number}{4} (\bibinfo{year}{1980}), \bibinfo{pages}{633--655}.
\newblock


\bibitem[\protect\citeauthoryear{Sellis, Roussopoulos, and Faloutsos}{Sellis
  et~al\mbox{.}}{1987}]%
        {SR87}
\bibfield{author}{\bibinfo{person}{T. Sellis}, \bibinfo{person}{N.
  Roussopoulos}, {and} \bibinfo{person}{C. Faloutsos}.}
  \bibinfo{year}{1987}\natexlab{}.
\newblock \showarticletitle{The R-tree: A dynamic index for multi-dimensional
  objects}.
\newblock \bibinfo{journal}{\emph{VLDB}} (\bibinfo{year}{1987}),
  \bibinfo{pages}{507--518}.
\newblock


\bibitem[\protect\citeauthoryear{Senellart, Jachiet, Maniu, and
  Ramusat}{Senellart et~al\mbox{.}}{2018}]%
        {SJ18}
\bibfield{author}{\bibinfo{person}{Pierre Senellart}, \bibinfo{person}{Louis
  Jachiet}, \bibinfo{person}{Silviu Maniu}, {and} \bibinfo{person}{Yann
  Ramusat}.} \bibinfo{year}{2018}\natexlab{}.
\newblock \showarticletitle{ProvSQL: provenance and probability management in
  postgreSQL}.
\newblock \bibinfo{journal}{\emph{PVLDB}} \bibinfo{volume}{11},
  \bibinfo{number}{12} (\bibinfo{year}{2018}), \bibinfo{pages}{2034--2037}.
\newblock


\bibitem[\protect\citeauthoryear{Sun, Franklin, Wang, and Wu}{Sun
  et~al\mbox{.}}{2016}]%
        {SF16}
\bibfield{author}{\bibinfo{person}{Liwen Sun}, \bibinfo{person}{Michael~J
  Franklin}, \bibinfo{person}{Jiannan Wang}, {and} \bibinfo{person}{Eugene
  Wu}.} \bibinfo{year}{2016}\natexlab{}.
\newblock \showarticletitle{Skipping-oriented partitioning for columnar
  layouts}.
\newblock \bibinfo{journal}{\emph{PVLDB}} \bibinfo{volume}{10},
  \bibinfo{number}{4} (\bibinfo{year}{2016}), \bibinfo{pages}{421--432}.
\newblock


\bibitem[\protect\citeauthoryear{Valduriez}{Valduriez}{1987}]%
        {V87}
\bibfield{author}{\bibinfo{person}{Patrick Valduriez}.}
  \bibinfo{year}{1987}\natexlab{}.
\newblock \showarticletitle{Join Indices}.
\newblock \bibinfo{journal}{\emph{TODS}} \bibinfo{volume}{12},
  \bibinfo{number}{2} (\bibinfo{year}{1987}), \bibinfo{pages}{218--246}.
\newblock


\bibitem[\protect\citeauthoryear{Wu and Madden}{Wu and Madden}{2013}]%
        {WM13}
\bibfield{author}{\bibinfo{person}{Eugene Wu} {and} \bibinfo{person}{Samuel
  Madden}.} \bibinfo{year}{2013}\natexlab{}.
\newblock \showarticletitle{Scorpion: Explaining Away Outliers in Aggregate
  Queries}.
\newblock \bibinfo{journal}{\emph{PVLDB}} \bibinfo{volume}{6},
  \bibinfo{number}{8} (\bibinfo{year}{2013}), \bibinfo{pages}{553--564}.
\newblock


\bibitem[\protect\citeauthoryear{Yu and Sarwat}{Yu and Sarwat}{2016}]%
        {yu2016two}
\bibfield{author}{\bibinfo{person}{Jia Yu} {and} \bibinfo{person}{Mohamed
  Sarwat}.} \bibinfo{year}{2016}\natexlab{}.
\newblock \showarticletitle{Two birds, one stone: a fast, yet lightweight,
  indexing scheme for modern database systems}.
\newblock \bibinfo{journal}{\emph{PVLDB}} \bibinfo{volume}{10},
  \bibinfo{number}{4} (\bibinfo{year}{2016}), \bibinfo{pages}{385--396}.
\newblock


\bibitem[\protect\citeauthoryear{Zhou, Larson, and Chaiken}{Zhou
  et~al\mbox{.}}{2010a}]%
        {ZL10}
\bibfield{author}{\bibinfo{person}{Jingren Zhou}, \bibinfo{person}{Per-Ake
  Larson}, {and} \bibinfo{person}{Ronnie Chaiken}.}
  \bibinfo{year}{2010}\natexlab{a}.
\newblock \showarticletitle{Incorporating partitioning and parallel plans into
  the SCOPE optimizer}. In \bibinfo{booktitle}{\emph{ICDE}}.
  \bibinfo{pages}{1060--1071}.
\newblock


\bibitem[\protect\citeauthoryear{Zhou, Arulraj, Navathe, Harris, and Xu}{Zhou
  et~al\mbox{.}}{2019}]%
        {zhou2019automated}
\bibfield{author}{\bibinfo{person}{Qi Zhou}, \bibinfo{person}{Joy Arulraj},
  \bibinfo{person}{Shamkant Navathe}, \bibinfo{person}{William Harris}, {and}
  \bibinfo{person}{Dong Xu}.} \bibinfo{year}{2019}\natexlab{}.
\newblock \showarticletitle{Automated verification of query equivalence using
  satisfiability modulo theories}.
\newblock \bibinfo{journal}{\emph{Proceedings of the VLDB Endowment}}
  \bibinfo{volume}{12}, \bibinfo{number}{11} (\bibinfo{year}{2019}),
  \bibinfo{pages}{1276--1288}.
\newblock


\bibitem[\protect\citeauthoryear{Zhou, Sherr, Tao, Li, Loo, and Mao}{Zhou
  et~al\mbox{.}}{2010b}]%
        {ZS10}
\bibfield{author}{\bibinfo{person}{Wenchao Zhou}, \bibinfo{person}{Micah
  Sherr}, \bibinfo{person}{Tao Tao}, \bibinfo{person}{Xiaozhou Li},
  \bibinfo{person}{Boon~Thau Loo}, {and} \bibinfo{person}{Yun Mao}.}
  \bibinfo{year}{2010}\natexlab{b}.
\newblock \showarticletitle{Efficient querying and maintenance of network
  provenance at internet-scale}. In \bibinfo{booktitle}{\emph{SIGMOD}}.
  \bibinfo{pages}{615--626}.
\newblock


\end{thebibliography}
